\documentclass[a4paper]{amsart}

\usepackage{amssymb}
\usepackage[all,defaultlines=3]{nowidow}
\usepackage[many]{tcolorbox}
\usepackage[store-sets-label]{keytheorems}
\usepackage{graphicx} 
\graphicspath{{../}{../Figures/}}
\usepackage{mathrsfs}
\usepackage{braket}
\usepackage{imakeidx}
\usepackage{hyperref}
\usepackage{mathtools}
\usepackage{microtype}
\usepackage{xcolor}
\usepackage{mdwtab}
\usepackage{zref-clever}
\zcsetup{nameinlink,cap}
\usepackage[noadjust]{cite}
\usepackage{comment}
\usepackage[all]{xy}
\usepackage[normalem]{ulem}
\usepackage{afterpage}
\usepackage[disableredefinitions,small]{complexity}
\usepackage{accents}
\usepackage{ifthen}
    \makeindex[title=Index of definitions\label{sec:index}]
\newboolean{Anonymous}
\setboolean{Anonymous}{false}

\zcsetup{
countertype = {
mycasesi = case,
mycasesii = subcase,
mycasesiii = subsubcase,
rpropi = property ,
twpropi = property ,
} ,
counterresetby = {
mypropii = mypropi 
}
}
\usepackage{enumitem}
\newlist{mycases}{enumerate}{3}
\setlist[mycases,1]{label=(\arabic*)}
\setlist[mycases,2]{label = \emph{\alph*}),
ref = \themycasesi.\emph{\alph*}}
\setlist[mycases,3]{label = \emph{\roman*}),
ref = \themycasesii.\emph{\roman*}}
\newlist{twprop}{enumerate}{1}
\setlist[twprop]{label={\rm (Tw\,\arabic*)}}
\newlist{rprop}{enumerate}{1}
\setlist[rprop]{label={\rm (r\,\arabic*)}}
\newlist{mytype}{enumerate}{1}
\setlist[mytype]{label={\rm Type \arabic*}:,
ref = \arabic*}
\zcsetup{abbrev}
\hypersetup{colorlinks=true,linkcolor=blue,filecolor=blue,citecolor=blue}
\newcommand{\tim}{induced tree-ordered minor}
\newcommand{\Tim}{Induced tree-ordered minor}
\newclass{\MSO}{MSO}

\DeclareMathOperator{\Cont}{\mathsf{IMinor}_\prec}
\DeclareMathOperator{\bomega}{\ensuremath{{\rm b}\omega_E}}

\DeclareMathOperator{\Mon}{Mon_\prec}
\title[Monadically dependent tree-ordered weakly sparse structures]{Characterizations of monadically dependent tree-ordered weakly sparse structures}
\ifAnonymous\else
\thanks{\ERCagreement}
\author[H. Buffière]{Hector Buffière}
\address{Université Paris Cité, CNRS, IRIF, Paris, France  \and Centre d'Analyse et de Mathématique Sociales CNRS UMR 8557, France}
\email{buffiere@irif.fr}
\author[Y. Lin]{Yuquan Lin}
\address{Southeast University, Nanjing, Jiangsu, China \and Centre d'Analyse et de Mathématique Sociales CNRS UMR 8557, France.}
\email{yqlin@seu.edu.cn}
\author[J. Ne\v set\v ril]{Jaroslav Ne\v set\v ril}
\address{Computer Science Institute of Charles University (IUUK), Praha, Czech Republic}
\email{nesetril@iuuk.mff.cuni.cz}
\author[P. Ossona de Mendez]{Patrice Ossona de Mendez}
\address{Centre d'Analyse et de Mathématique Sociales CNRS UMR 8557, France \and Computer Science Institute of Charles University (IUUK), Praha, Czech Republic}
\email{pom@ehess.fr}
\author[S. Siebertz]{Sebastian Siebertz}
\address{University of Bremen, Bremen, Germany}
\email{siebertz@uni-bremen.de}
\date{\today}
\newcommand{\ERCagreement}{
        {\footnotesize
        The second author is supported by the China Scholarship Council (CSC)	and  SEU Innovation Capability Enhancement Plan for Doctoral Students (CXJH\_SEU 24119).}\\[4pt]
		\noindent\begin{minipage}{.73\textwidth}
			\footnotesize
    This paper is part of a project that has received funding from the European Research Council (ERC) under the European Union's Horizon 2020 research and innovation program (grant agreement No 810115 -- {\sc Dynasnet}), and from the 
    German Research Foundation (DFG) with grant agreement No 444419611
		\end{minipage}\hfill
		\begin{minipage}{.25\textwidth}
\phantom{.}\hfill\includegraphics[height=13mm]{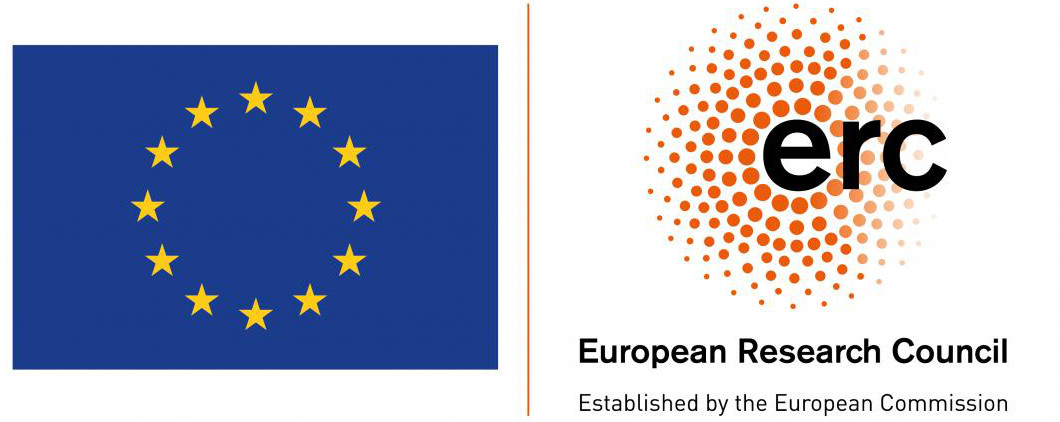}\hfill\phantom{.}
		\end{minipage}
}
\fi

\newtheorem{ex}{Example}[section]
\newtheorem{conj}{Conjecture}[section]

\newtheorem{remark}{Remark}[section]
\newtheorem{claim}{$\rhd$ Claim}[section]

\zcRefTypeSetup{case}{
Name-sg = Case ,
name-sg = case ,
Name-pl = Cases ,
name-pl = cases ,
}
\zcRefTypeSetup{subcase}{
Name-sg = Subcase ,
name-sg = subcase ,
Name-pl = Subcases ,
name-pl = subcases ,
}
\zcRefTypeSetup{subsubcase}{
Name-sg = Subsubcase ,
name-sg = subsubcase ,
Name-pl = Subsubcases ,
name-pl = subsubcases ,
}
\zcRefTypeSetup{property}{
Name-sg = Property ,
name-sg = property ,
Name-pl = Properties ,
name-pl = properties ,
}

\zcRefTypeSetup{claim}{
Name-sg = Claim ,
name-sg = claim ,
Name-pl = Claims ,
name-pl = claims ,
}
\zcRefTypeSetup{conjecture}{
Name-sg = Conjecture ,
name-sg = conjecture ,
Name-pl = Conjectures ,
name-pl = conjectures ,
}
\zcRefTypeSetup{mytypei}{Name-sg-ab = Type}
\zcsetup{countertype={cor=corollary}}
\zcsetup{countertype={lem=lemma}}
\zcsetup{countertype={ex=example}}
\zcsetup{countertype={thm=theorem}}
\zcsetup{countertype={defi=definition}}
\zcsetup{countertype={ndefi=definition}}
\newkeytheoremstyle{tcbthm}{
tcolorbox-no-titlebar={colframe=red,interior hidden, breakable,before skip=5pt,after skip=5pt,	outer arc=0pt, blanker,
borderline west={1pt}{0pt}{red},
	top=3pt,
	right=3pt,
    left=5pt,
	bottom=3pt},
}
\newkeytheoremstyle{tcbcor}{
tcolorbox-no-titlebar={colframe=purple,interior hidden, breakable,before skip=5pt,after skip=5pt,	outer arc=0pt, blanker,
borderline west={1pt}{0pt}{purple},
	top=3pt,
	right=3pt,
    left=5pt,
	bottom=3pt},
}
\newkeytheoremstyle{tcblem}{
tcolorbox-no-titlebar={colframe=blue,interior hidden, breakable,before skip=5pt,after skip=5pt,	outer arc=0pt, blanker,
borderline west={1pt}{0pt}{blue},
	top=3pt,
	right=3pt,
    left=5pt,
	bottom=3pt},
}
\newkeytheoremstyle{tcbdef}{
tcolorbox-no-titlebar={colframe=green,interior hidden, breakable,before skip=5pt,after skip=5pt,	outer arc=0pt, blanker,
borderline west={1pt}{0pt}{green},
	top=3pt,
	right=3pt,
    left=5pt,
	bottom=3pt},
}
\newkeytheoremstyle{tcbnot}{
tcolorbox-no-titlebar={colframe=yellow,interior hidden, breakable,before skip=5pt,after skip=5pt,	outer arc=0pt,
	arc=0pt,
	colback=white,
		rightrule=0pt,
		toprule=0pt,
	top=3pt,
	right=3pt,
    left=3pt,
	bottom=3pt,
	bottomrule=0pt,},
}
\newkeytheoremstyle{tcbprop}{
tcolorbox-no-titlebar={colframe=gray,interior hidden, breakable,before skip=5pt,after skip=5pt,	outer arc=0pt,
	arc=0pt,
	colback=white,
		rightrule=0pt,
		toprule=0pt,
	top=3pt,
	right=3pt,
    left=3pt,
	bottom=3pt,
	bottomrule=0pt,},
}
\newkeytheorem{property}[style=tcbprop,name=Property,numbered=false]
\newkeytheorem{thm}[style=tcbthm,name=Theorem,parent=section]
\newkeytheorem{theorem}[sibling=thm]
\newkeytheorem{lemma}[sibling=thm]
\newkeytheorem{corollary}[sibling=thm]
\newkeytheorem{cor}[style=tcbcor,name=Corollary,sibling=thm]
\newkeytheorem{lem}[style=tcblem,name=Lemma,sibling=thm]
\newkeytheorem{ndefi}[style=tcbdef,name=Definition,sibling=thm]
\newkeytheorem{defi}[sibling=thm,name=Definition]
\newkeytheorem{fact}[sibling=thm]
\newkeytheorem{nota}[style=tcbnot,name=Notation,sibling=thm]

\newkeytheorem{rproof}[name=Proof,numbered=false,style=remark,qed]

\newenvironment{clproof}{ \trivlist
	\item[\hskip\labelsep
	\emph{Proof of the claim}.]\ignorespaces
}{\hfill$\vartriangleleft$\vspace{10pt}\par}

\newcommand{\Cc}{\mathscr C}
\newcommand{\Dd}{\mathscr D}
\newcommand{\Mm}{\mathscr M}
\newcommand{\Nn}{\mathscr N}
\newcommand{\struc}[1]{\mathbf{#1}}
\renewcommand{\phi}{\varphi}

\newcommand{\equ}[2]{$\xy\xymatrix@=18pt{(#1)\ar@2{<->}[r]&(#2)}\endxy$}
\newcommand{\impl}[2]{$\xy\xymatrix@=18pt{(#1)\ar@2[r]&(#2)}\endxy$}

\usepackage{cleveref}
\usepackage{xifthen}

\DeclareMathOperator{\atp}{\rm atp}
\DeclareMathOperator{\otp}{\rm otp}
\DeclareMathOperator{\Inc}{\rm Inc}
\DeclareMathOperator{\Gaif}{\mathsf{Gaif}}
\DeclareMathOperator{\Flat}{\rm Flat}
\DeclareMathOperator{\TInc}{{\rm Inc}_\prec}
\DeclareMathOperator{\TGaif}{\mathsf{Gaif}_\prec}
\DeclareMathOperator{\Minor}{Minor_\prec}

\newcommand{\shm}{\mathbin{\nabla}}
\newcommand{\meet}{\mathbin{\curlywedge}}
\newcommand{\ndef}[2][]{\emph{#2}%
\ifthenelse{\isempty{#1}}{\index{#2}}{\index{#2 (#1)}}}
\newtheoremstyle{idxstyle} 
    {0pt}                    
    {0pt}                    
    {}                   
    {}                           
    {}                   
    {}                          
    {0pt}                       
    {}  

\theoremstyle{idxstyle}
\newtheorem{idx}{}

\zcRefTypeSetup{idx}{
Name-sg = Index of definitions,
name-sg = index of definitions ,
Name-pl = Indexes ,
name-pl = indexes ,
}

\begin{document}	

\begin{abstract}
    A class of structures is monadically dependent if one cannot interpret all graphs in colored expansions from the class using a fixed first-order formula. 
    A tree-ordered  $\sigma$-structure is the expansion of a $\sigma$-structure with a tree-order. 
    A tree-ordered $\sigma$-structure is weakly sparse if the Gaifman graph of its $\sigma$-reduct excludes some  biclique (of a given fixed size) as a subgraph. 
    Tree-ordered weakly sparse graphs are commonly used as tree-models (for example for classes with bounded shrubdepth,  structurally bounded expansion, bounded cliquewidth, or bounded twin-width), motivating their study on their own.
    In this paper, we consider several
    constructions on tree-ordered structures, such as tree-ordered variants of the Gaifman graph and of the incidence graph, induced and non-induced tree-ordered minors, and generalized fundamental graphs.
    We provide characterizations of monadically dependent classes of tree-ordered weakly sparse $\sigma$-structures based on each of these constructions, some of them establishing unexpected bridges with sparsity theory. As an application, we prove that a class of tree-ordered weakly sparse structures is monadically dependent if and only if its sparsification is nowhere-dense. Moreover, the sparsification transduction translates boundedness of clique-width and linear clique-width into boundedness of tree-width and path-width.
    We also prove that first-order model checking is not fixed parameter tractable on independent hereditary classes of tree-ordered weakly sparse graphs (assuming $\AW[*]\neq\FPT$) and give what we believe is the first model-theoretical characterization of classes of graphs excluding a minor, thus opening a new perspective of structural graph theory.
\end{abstract}

\vspace*{-2.2cm}
\maketitle
\setcounter{tocdepth}{1} 

\vspace*{-0.5cm}
{\small\tableofcontents}

\section{Introduction}

Recent years have witnessed an increasingly tight interplay between structural graph theory, model theory, and algorithmic complexity. 
A recurring theme in these three areas is that strong structural restrictions, often phrased in terms of sparsity or decomposability, simultaneously lead  to robust notions of tameness in logic and to efficient algorithms. 
Sparsity theory, introduced by Ne\v{s}et\v{r}il and Ossona de Mendez~\cite{Sparsity}, provides a particularly rich instance of this phenomenon. 
For example, \emph{nowhere dense} classes of graphs admit a wealth of equivalent combinatorial, algorithmic, and logical characterizations, and they form the precise boundary of tractability for first-order model checking on monotone graph classes~\cite{dvovrak2010deciding,grohe2017deciding}.

A decisive conceptual step was taken by Adler and Adler~\cite{Adler2013}, building on earlier work of Podewski and Ziegler~\cite{Podewski1978}, by connecting notions of model-theoretical tameness with graph sparsity. 
They proved that on monotone classes of binary structures the model-theoretical notions of dependence and stability, as well as their monadic variants, collapse to nowhere density. 
(Stability and dependence are two of the most central tameness notions from model theory~\cite{shelah1990classification}.)
Braunfeld and Laskowski~\cite{braunfeld2022existential} recently proved that  dependence and stability still collapse to their monadic variants under the weaker (and standard) assumption of heredity (i.e., closure by induced substructures). 
Therefore, under the heredity assumption, a result by Baldwin and Shelah~\cite{baldwin1985second} allows us to use, instead of the standard definitions of dependence and stability, characterizations of monadic dependence and monadic stability based on first-order transductions.
Informally, a class of structures $\Mm$ is monadically dependent if one cannot first-order transduce the class of all graphs from $\Mm$.
These insights support the systematic study of \emph{structurally sparse} classes, that is, classes obtained as first-order transductions of sparse ones~\cite{SurveyND,SBE_TOCL,gajarsky2020new}, as well as the study of monadic dependence and stability through a combinatorial lens. 
Many results were  extended from the sparse to the more general model-theoretical setting, and it was conjectured that the (essentially) graph theoretical notion of structural nowhere density coincides with the model theoretical notion of monadic stability~\cite{RW_SODA}. 

This not only established a bridge between combinatorics and model theory, but led to combinatorial characterizations of those hereditary classes of graphs that are (monadically) stable \cite{braunfeld2025decomposition, DreierMST23,MCST} and (monadically) dependent \cite{dreier2024flipbreakability,Bonnet2025,svm2024}.
It also led to important algorithmic applications. 
It was conjectured that on hereditary classes of graphs first-order model checking is fixed-parameter tractable if and only if the class is monadically dependent~\cite{tww4}. 
At this time, important partial results have been obtained in this direction, namely the fixed parameter tractability of first-order model checking for monadically stable classes of graphs~\cite{dreier2023first,MCST} and (assuming some special decomposition is given with the input) for classes with bounded twin-width~\cite{bonnet2021twin} and, more generally, for classes with bounded merge-width~\cite{dreier2025merge}. 
Furthermore, it has been established that monadic dependence constitutes a limit for tractability of first-order model checking on hereditary classes of graphs~\cite{dreier2024flipbreakability}. 

At the heart of the obtained tractability results for the aforementioned (sparse and dense) classes lie special tree-like decompositions. 
These include low tree-depth decompositions~\cite{Taxi_stoc06,POMNI}, low shrub-depth decompositions \cite{SBE_drops,SBE_TOCL,covers,kwon2020low}, clique-width expressions~\cite{courcelle1993handle}, and twin-models~\cite{TWW_perm}. 
These results show that, at least implicitly, many monadically dependent classes have a tree-ordered structure and suggest that  \emph{monadically dependent classes could be exactly first-order transductions of monadically dependent tree-ordered nowhere dense classes}~\cite{Buffiere_lcw}.
On the other hand, the  striking equivalence, for classes of totally ordered binary structures, of monadic dependence and  twin-width boundedness~\cite{tww4} suggests that perhaps monadically dependent classes of tree-ordered structures could also have a non-trivial but compact structural characterization.

This ubiquity of tree-ordered sparse graphs motivates our present study. 
We focus on \emph{tree-ordered weakly sparse (TOWS) structures}: weakly sparse relational structures equipped with a tree-order (see more formal definition below). 
Our aim is to understand precisely when classes of such tree-ordered structures are monadically dependent. 
The core contribution of this paper shows that, in the presence of a tree-order and under the assumption of weak sparsity, monadic dependence is characterized by the nowhere denseness of an associated minor closure. 
In order to provide a more precise formulation of this result, we take the time to outline a few definitions.

As usual, a class $\Mm$ of structures is said to be \emph{weakly sparse} if the class $\Gaif(\Mm)$ of its Gaifman graphs is weakly sparse, that is, excludes some balanced biclique (of a given fixed size) as a subgraph;
similarly, the class of structures $\Mm$ is said to be \emph{nowhere dense} if the class $\Gaif(\Mm)$ is nowhere dense.  

A \emph{tree-order} is a partial order $\prec$ whose cover (or direct successor) relation induces a tree. 
Recall that one of the main shortcomings of first-order logic is that it can express only local properties, captured e.g.\ by Gaifman's locality theorem. 
Tree-orders, however, are transitive, hence making it possible for first-order logic to access long directed paths in the order. 

We start our study by a dichotomy theorem for hereditary classes of tree-ordered graphs:  each such class either
includes some special monadically independent classes of tree-ordered graphs generated from a small structure $\tau$ (called a \emph{core}; see \zcref{fig:core}) by a natural extension map $\tau\mapsto\{\tau[n]\colon n\in\mathbb N\}$, or is dependent.

\getkeytheorem{thm:lics_twists}

\begin{figure}[ht]
    \centering
    \includegraphics[width=0.75\linewidth]{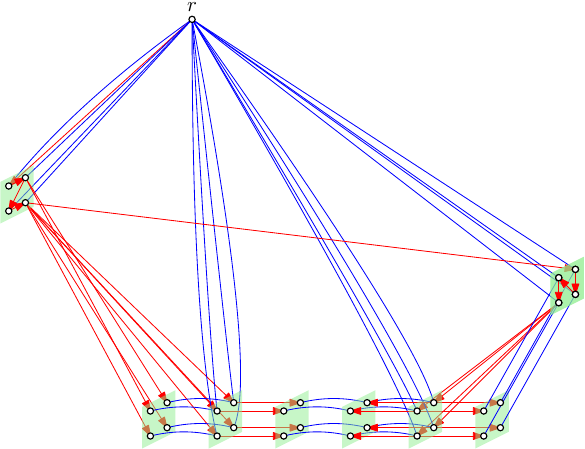}
    \caption{A core; red arrows are the cover relations of $\prec$, blue edges are the $E$-relations.}
    \label{fig:core}
\end{figure}

The long and technical proof of this theorem relies on the characterization of  monadically dependent classes of binary structures  in terms of \emph{transformers} given by Dreier et al.\ in~\cite{dreier2024flipbreakability} and is obtained by refining the extracted substructures.

Then, motivated by the particular structure of tree-ordered structures, we consider \emph{tree-ordered minors}, which are obtained by iteratively identifying direct successors in the order (we do not require that the identified vertices are connected by an edge/tuple), or equivalently by contracting subtrees of the underlying tree-order, and by deleting tuples from relations and  non-root vertices.  
This is a natural combinatorial operation;
see \zcref{fig:Tminor}.
We also consider \emph{induced tree-ordered minors}, which do not allow the deletion of tuples but are first-order transductions. 

\begin{figure}[ht]
    \centering
    \includegraphics[width=\columnwidth]{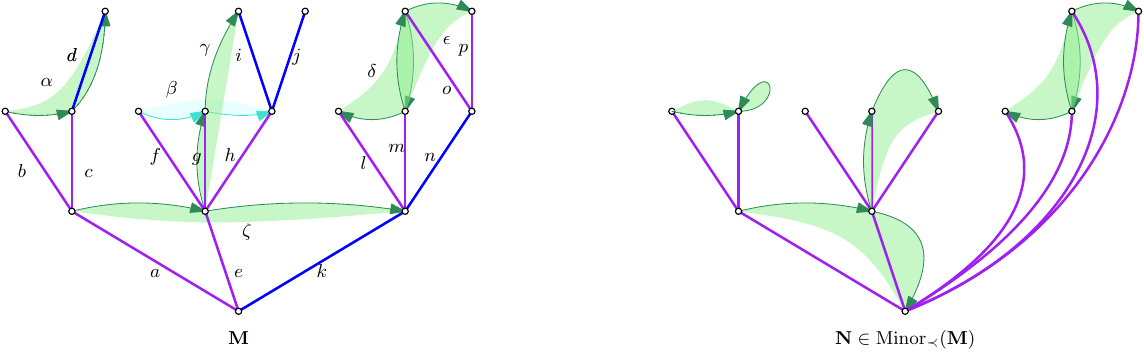}
    \caption{On the right, the tree-ordered minor of the tree-ordered structure $\mathbf M$ obtained by contracting covers $d,i,j,k,n$ and deleting the triple $\beta$.}
    \label{fig:Tminor}
\end{figure}

We also consider variants of standard constructions like the Gaifman graph, the incidence graph, and the fundamental graph relative to a spanning tree. 
In our setting, it will be convenient to consider the tree-order separately, and to introduce variants $\TGaif$ and $\TInc$ of the tree-ordered Gaifman and incidence graphs of a tree-ordered structure.
We shall also consider a representation by generalized fundamental graphs (where the role of the spanning tree is played by the cover graph of the tree order). We refer to the later part of the paper for  definitions and discussion.
Examples of these constructions are shown in \zcref{fig:all}.

\begin{figure}[ht]
    \centering
    \includegraphics[width=\columnwidth]{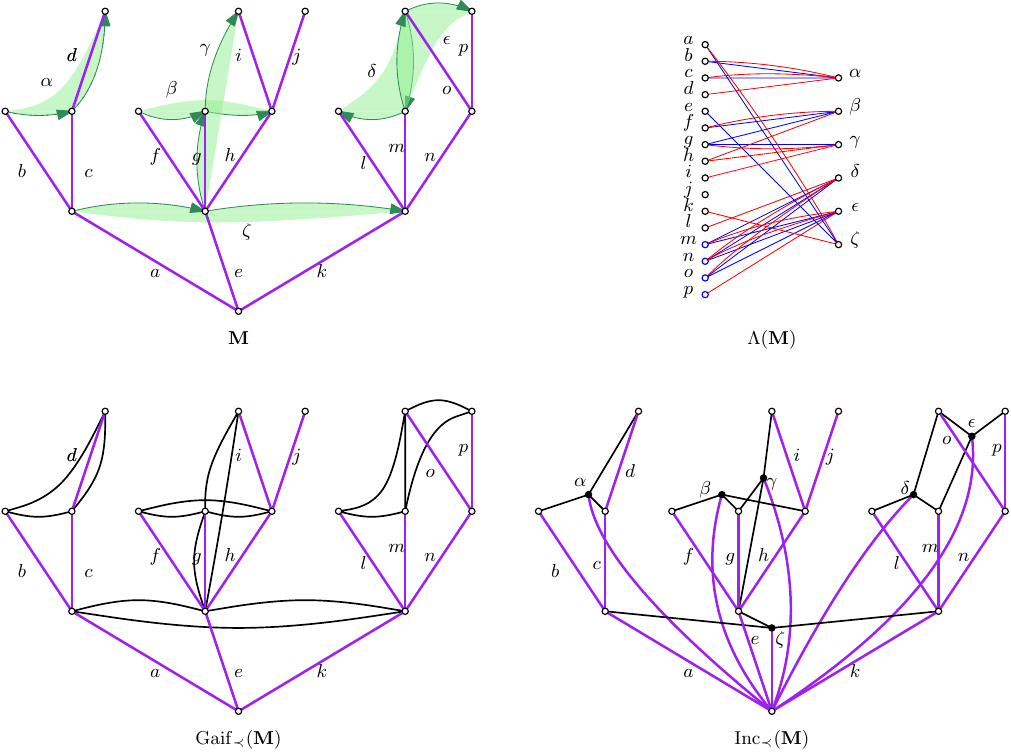}
    \caption{Constructions derived from a tree-ordered $\sigma$-structure. Top left, a tree-ordered $\sigma$-structure $\mathbf M$ is composed of a tree-order (Hasse diagram in purple fat) and a weakly sparse set of triples (green, with arrows indicating the triple order). Top right: the generalized fundamental graph $\Lambda(\mathbf M)$ of~$\mathbf M$. (If $k$ is the maximum arity of the relations in $\sigma$, then $\Lambda(\mathbf M)$ has $k-1$ types of edges). Bottom left: the tree-ordered Gaifman graph of $\mathbf M$. Bottom right: the tree-ordered incidence graph of $\mathbf M$.} 
    \label{fig:all}
\end{figure}

Our main result establishes that, for tree-ordered weakly sparse structures, monadic dependence admits characterizations based on each of these constructions, including 
a purely sparsity-theoretical characterization. 

\getkeytheorem{thm:mainS}

Here, $\TGaif(\mathscr M)$ denotes the class of tree-ordered Gaifman graphs of $\Mm$, 
$\TInc(\mathscr M)$ denotes the class of tree-ordered incidence graphs of $\Mm$, 
$\Lambda(\mathscr M)$ denotes the class of generalized fundamental graphs of $\Mm$ (see \zcref{def:TFun}), $\Minor(\mathscr M)$ is the class of tree-ordered minors of $\Mm$, and $\Minor(\mathscr M)^\sigma$ is its $\sigma$-reduct.

Before we describe some applications derived in this paper, we take the time for a few comments on \zcref{thm:mainS}. When considering TOWS structures, all the complexity involving density or unstability relies on interplay of the tree-order component. At first glance, we would not expect to be able to ``sparsify'' a TOWS structure while keeping these aspects (as any sparsification would be stable). However, \zcref{thm:mainS} shows that if each structure is sparsified into a finite set of sparse structures (the $\sigma$-reducts of the tree-ordered minors) we actually keep the complexity of the structure to a large extent.
This unexpected behavior opens an interesting bridge to sparsity theory, which remains to be explored.

It also follows from \zcref{thm:mainS} that TOWS structures basically reduce to TOWS graphs, as witnessed by the equivalence of the monadic dependence of a class~$\mathscr M$ and of the class $\TGaif(\mathscr M)$ of its tree-ordered Gaifman graphs. This implies that we can use arbitrarily complicated signatures on the unordered component without complicating  the condition under which these structures form a monadically dependent class. Moreover, 
the possibility to consider both the tree-ordered Gaifman graphs and the tree-ordered incidence graphs (which cannot be obtained as a transduction of the structure) is also a key property, which relies on the sparsity-theoretical characterization.

Finally, the equivalence of the last two items of \zcref{thm:mainS} demonstrates that the complexity of tree-ordered minor closed classes of TOWS $\sigma$-structures is essentially encoded in their $\sigma$-reduct.
This, as well as the characterization by generalized fundamental graphs (which are instrumental in matroid theory), suggests to revisit the theory of graph minors in this new setting.

\medskip

We now list some applications of our results.
First, when the involved tree-orders have  cover graphs with bounded degree, we prove a connection with twin-width.

\getkeytheorem{thm:tows_tww}

Our new structural characterizations open the way to systematically transfer techniques and results from sparsity theory to the study of monadically dependent tree-ordered weakly sparse structures.
We present several applications, which are listed  at the end of the paper. 
Let us highlight a few of them.

The equivalence of (1) and (6) in  \zcref{thm:mainS}, leads to  the following particular characterization of classes of graphs excluding a minor:
Let~$G$ be a graph. 
A \emph{spanning-tree-ordering} of $G$ is the expansion of $G$ by a tree-ordering of~$G$ defined by the setting of a global root and a rooted forest (as a subgraph of $G$) on the remaining vertices. 
The \emph{spanning-tree-ordering expansion} $\mathscr C_\Upsilon$ of a class of graphs $\mathscr C$ is the class of all the spanning-tree-orderings of all the graphs in $\mathscr C$.

\getkeytheorem{thm:excl_minor}

A slight variation of the graph version of \zcref{thm:mainS} involves the following \emph{sparsifying} transduction $\mathsf{Sp}$, which adds the cover graph of the tree-order to the edge set and then takes the $E$-reduct of all {\tim}s.  
We prove that $\mathsf{Sp}$ has the following striking properties (here ${\rm cw}, {\rm lcw}, {\rm sd}, {\rm tw}, {\rm pw}$, and ${\rm td}$ denote clique-width, linear clique-width, shrub-depth, tree-width, path-width, and tree-depth, respectively):

\getkeytheorem{thm:sp}

Using their structure theorem, Dreier \emph{et al.} \cite{dreier2024flipbreakability} proved that the algorithmic problem
of \FO-model checking is not tractable (precisely, {\FPT}) on independent hereditary classes of graphs (under the standard assumption that $\AW[*]\neq\FPT$).
This result most probably extends to binary structures where all the relations are symmetric, but the case where some relation is not symmetric is less clear. For this reason, it does not seem to follow from known results that \FO-model checking is not {\FPT} on independent hereditary classes of TOWS graphs. 
However, this follows from our study of unavoidable definable {\tim}s in hereditary independent classes of TOWS graphs. 

\getkeytheorem{thm:intractability}

We conjecture that this property extends to every independent hereditary class of relational structures.

\subsection{Structure of the paper}

An \zcref[noref,nocap]{sec:index} 
 is given at the end of the paper.
The remaining of the paper is organized around three parts.

\noindent $\vartriangleright$ The first part is devoted to the study of hereditary classes of 
TOWS graphs. Some general background from graph theory, order theory and model theory is first recalled in \zcref{sec:prelim}.
The notion of tree-ordered weakly sparse graphs is introduced in \zcref{sec:TOWS_G}, and some important related constructions are defined, including (induced) tree-ordered minors and fundamental graphs. 
The (long and technical) \zcref{sec:dep_bin} is devoted to the characterization of independent hereditary classes of TOWS-graphs in terms of induced substructures. 
The main result of this section (\zcref{thm:lics_twists}) is derived from the characterization of independent hereditary class of binary structures given in \cite{dreier2024flipbreakability} in terms of transformers, by considering substructures more suited to our framework, including twisters, clean twisters, and eventually twists. 
From this characterization we deduce in \zcref{sec:Iminor} a characterization in terms of (induced) tree-ordered minors, \zcref{thm:mainG}, which establishes an unexpected bridge with sparsity theory, as well as characterization in terms of  fundamental graphs (\zcref{sec:fund_dep}). This theorem is the tree-ordered graph version of \zcref{thm:mainS}, which is generalized in the second part.

\noindent $\vartriangleright$ In the second part, we extend our study of hereditary classes of TOWS-structures. 
In \zcref{sec:Gstdep}, we extend the theorem of Adler and Adler relating nowhere denseness to dependence, stability, and their monadic version in monotone classes of binary structures (\zcref{thm:Adler}) to hereditary classes of weakly sparse relational structures (\zcref{cor:NDG}). In \zcref{sec:TOWS}, we generalize the constructions introduced in \zcref{sec:TOWS_G} to general structures and extend in \zcref{sec:timS} the notions of (induced) tree-ordered minors. With these tools in hand, we generalize in \zcref{sec:depS} the characterization \zcref{thm:mainG} to general  relational structures, thus obtaining  \zcref{thm:mainS}, the main result of this paper.

\noindent $\vartriangleright$ The third part is devoted to applications. In \zcref{sec:sp}, we  introduce a notion of sparsification and show how it captures not only the dependence, but also the boundedness of clique-width, linear clique-width, and shrubdepth (\zcref{thm:sp}). A key aspect of this sparsification is that it doesn't map a structure to a sparse structure, but to a (finite) set of sparse structures.
In a second application (\zcref{sec:model checking}), we use the possibility to efficiently encode graphs in twists (\zcref{cor:interpret}) to prove the hardness of model checking on general TOWS graphs (\zcref{thm:intractability}).
The third application (\zcref{sec:minor}) is maybe the first model theoretical characterization of classes of graphs excluding a minor (\zcref{thm:excl_minor}). Finally, in a fourth application (\zcref{sec:tows_tww}), we consider the special case of TOWS graphs whose cover graph has bounded degree and prove that, in this case, monadic dependence is equivalent to the boundedness of twin-width (\zcref{thm:tows_tww}).

\newpage

\part{Monadic dependence of tree-ordered weakly sparse graphs}
\addtocontents{toc}{\vspace{2pt}}

\section{Preliminaries}
\label{sec:prelim}

\subsection{General model theory}

A \ndef{relational signature} $\sigma$ is a set of relation symbols, each 
with an associated non-negative integer, called its \ndef{arity}.
A \ndef{binary signature} is a signature with only unary and binary relation symbols. 
In the following, whenever we speak of a signature, we mean a finite relational signature. 
A \ndef{$\sigma$-structure} $\struc M$ consists of a \ndef{universe} $M$, 
which is a non-empty, possibly infinite set, and \emph{interpretations} of the symbols from the
signature: each relation symbol $R$ of arity $k$ is interpreted as
a $k$-ary relation $R(\struc M)\subseteq M^k$. 
We often do not distinguish between relation symbols and their interpretations. The logical values `true' and `false' are denoted by~$\top$ and~$\bot$.

A \emph{graph} is a finite structure over the signature consisting of a binary
relation symbol~$E$, interpreted as a symmetric and irreflexive edge relation.

First-order logic over a signature $\sigma$ is defined in 
the usual way.
We usually write~$\bar x$ for tuples $(x_1,\ldots, x_k)$
of variables, $\bar a$ for tuples $(a_1,\ldots, a_\ell)$ of elements and leave it to the context to determine the length of the tuple.
For a formula $\varphi(\bar x;\bar y)$, a structure $\mathbf M$ and a tuple $\bar a$ with $|\bar a|=|\bar y|$, we define
$$\varphi(\mathbf M,\bar a)=\{\bar b\in M^{|\bar x|}\colon \mathbf M\models \varphi(\bar b;\bar a)\}.$$

We call a set $X$ in $\mathbf M$ \emph{definable (with parameters $\bar a$)} if there exists $\phi(\bar x,\bar y)$ such that $X=\varphi(\mathbf M,\bar a)$. 
An \ndef{interpretation} of a structure $\mathbf M$ in a structure $\mathbf N$ is a surjective map from a subset of some (finite) power of the domain of $\mathbf N$ to the domain of $\mathbf M$, such that the preimage of a definable set is a definable set.
In the special case 
where the mapping identifies the domain of $\mathbf M$ with a subset of the domain of $\mathbf N$, and
 considering classes instead of structures, 
we get the notion of simple interpretations.
\subsubsection{Interpretations}

A \ndef{simple interpretation} $\mathsf I$ of $\tau$-structures in $\sigma$-structures consists  of a $\sigma$-formula $\nu(x)$ and a $\sigma$-formula $\rho_R(\bar x)$ for each  $R\in \tau$ (with arity~$|\bar x|$). For a $\sigma$-structure $\mathbf M$, 
$\mathsf I(\mathbf M)$ is the $\tau$-structure $\mathbf N$ with domain $\nu(\mathbf M)$, such that $R(\mathbf N)=\{\bar a\subseteq\mathbf N\colon \mathbf M\models\rho_R(\bar a)\}$ for every $R\in\tau$.
For a class $\Mm$ of $\sigma$-structures we define $\mathsf I(\Mm)=\bigcup_{\struc M\in \Mm} \mathsf I(\struc M)$. 
We say that a class $\Nn$ can be interpreted in a class~$\Mm$ if there exists an interpretation~$\mathsf I$ such that $\Nn\subseteq \mathsf I(\Mm)$. 

\begin{ex}
\label{ex:reduct}
Let $\sigma$ be a relational signature and let $\tau\subseteq \sigma$. The $\tau$-\ndef{reduct} is the simple interpretation defined by $\nu(x):=\top$ and $\rho_R(\bar x):=R(\bar x)$ for every $R\in\tau$. For the sake of simplicity, we denote by
$\mathbf M^\tau$ the $\tau$-reduct of a structure $\mathbf M$.
\end{ex}

Following \zcref{ex:reduct}, reducts are denoted using a superscript. For instance, 
for a class~$\mathscr C$ of $\sigma$-structures, we denote by extension $\mathscr C^\tau$ the class of all the $\tau$-reducts of the structures in~$\mathscr C$.
Moreover, for a $\sigma$-structure $\mathbf M$ and $R\in\sigma$, it will be convenient to use the term of $R$-reduct instead of $\{R\}$-reduct, and to denote the $R$-reduct of~$\mathbf M$ by $\mathbf M^R$.
\subsubsection{Definition, stability, and dependence}
Close to the notion of interpretation is the notion of definition.
Let $\Dd$ be a class of finite $\sigma$-structures and let $d\in\mathbb N$. The class 
$\Dd$ is \emph{defined} in a class $\Cc$ (by formulas $\rho_R(\bar x^1;\dots;\bar x^k)$ for $R\in\sigma$ with arity $k$ and $|\bar x^i|=d$)  if, for every  $\mathbf M\in\Dd$ there exists $\mathbf N\in\Cc$ and an injective map $f:M\rightarrow N^d$ such that
$$\mathbf M\models R(a_1,\dots,a_k)\quad\iff\quad
\mathbf N\models \rho_R(f(a_1);\dots;f(a_k))$$ holds for every $R\in\sigma$ (with arity $k$).

A class $\Cc$ is \ndef{stable} if the class of all finite linear orders cannot be defined in $\Cc$; it is \ndef{dependent} if the class of all finite graphs cannot be defined in $\Cc$.

\subsubsection{Transductions, monadic stability, and monadic dependence}

For $k\in\mathbb N$, the \ndef{copy operation} $\mathsf C_k$ maps a $\sigma$-structure $\mathbf M$ to the $\sigma\cup\{E\}$-structure~$\mathsf C_k(\mathbf M)$ obtained by taking $k$ disjoint copies of $\mathbf M$ and making the clones of an element of $\mathbf M$ adjacent in $E$. (Note that $\mathsf C_1$ is the identity mapping.)

A unary predicate is also called a \emph{color}. 
A \ndef{monadic expansion} or \emph{coloring} of a \mbox{$\sigma$-structure~$\struc M$} is a \mbox{$\sigma^+$-structure~$\struc M^+$}, where $\sigma^+$
is obtained from $\sigma$ by adding unary relations, such 
that $\struc M$ is the $\sigma$-reduct of $\struc M^+$.
For a set $\Sigma$ of colors, the \ndef{coloring operation} $\Gamma_\Sigma$ maps a $\sigma$-structure $\struc M$ to the set~$\Gamma_\Sigma(\struc M)$ of all its $\Sigma$-colorings.

A \ndef{transduction} $\mathsf T$ is a composition
of copy operations, monadic expansions, and simple interpretations.
Every transduction $\mathsf T$ is equivalent to the composition
$\mathsf I\circ\Gamma_\Sigma\circ \mathsf C_k$ of a copy operation $\mathsf C_k$, a
coloring operation~$\Gamma_\Sigma$, and a simple interpretation~$\mathsf I$ of $\tau$-structures in $\Sigma$-colored \mbox{$\sigma$-structures} \cite{SBE_TOCL}. 
Hence, for every $\sigma$-structure we have
$\mathsf T(\struc M)=\{\mathsf I(\struc M^+): \struc M^+\in\Gamma_\Sigma(\mathsf C_k(\struc M))\}$. 
(When we define a transduction, if we do not say anything about the copy operation, this means that it is not present, i.e.\ it is $\mathsf C_1$.)

For a class $\Mm$ of $\sigma$-structures we define $\mathsf T(\Mm)=\bigcup_{\struc M\in \Mm} \mathsf T(\struc M)$. 
We say that a class $\Nn$ can be transduced in a class~$\Mm$ if there exists a transduction $\mathsf T$ such that $\Nn\subseteq \mathsf T(\Mm)$. 

A class $\Cc$ is \ndef{monadically stable} if all its monadic expansions are stable, and it is \ndef{monadically dependent} if all its monadic expansions are dependent. 
As shown by Baldwin and Shelah~\cite{baldwin1985second}, a class $\Mm$ is monadically stable/dependent if and only if one cannot transduce all finite linear orders/all graphs from $\Mm$. 
The following result is a  consequence of the simple observation that interpretations and transductions compose. 

\begin{fact}[label=lem:preserve-by-interpret]
    When $\Mm$ is stable/dependent, then all its interpretations (simple or not) are stable/dependent, and when $\Mm$ is monadically stable/dependent, then all its transductions are monadically stable/dependent. 
\end{fact}

The simple interpretation\footnote{Strictly speaking, we should define an interpretation for each signature~$\sigma$, but we choose to denote all these interpretations $\Gaif$ to avoid unnecessarily overloading notations.} $\Gaif$ of graphs in $\sigma$-structures is defined by the formula $\nu(x):=\top$ and the formula $\rho_E(x,y)$ expressing that $x$ and $y$ are distinct and belong together to some tuple in $R(\mathbf M)$ for some $R\in\sigma$. 
The graph $\Gaif(\mathbf M)$ is the \ndef{Gaifman graph} of~$\mathbf M$. 
Note that $\Gaif(\Mm)$, being a simple interpretation, is  a transduction of~$\Mm$, and hence by \zcref{lem:preserve-by-interpret}, if $\Mm$ is monadically stable, then $\Gaif(\Mm)$ is monadically stable. 

The \ndef{incidence graph}\footnote{Note that, instead of the incidence graph, one can define the incidence multigraph, where an edge $(u,(R,\bar v))$ has  multiplicity $k$ if there exists $k$ distinct indices $1\leq i_1<\dots<i_k\leq |\bar v|$ such that $u=v_{i_j}$ for all $j\in[k]$. In our setting, it will be more convenient not to consider multiplicities.} $\Inc(\mathbf{M})$ of a $\sigma$-structure is the bipartite graph, whose parts are the universe $M$ of $\mathbf M$ and the pairs $(R,\bar v)$ with $R\in\sigma$ and $\bar v\in R(\mathbf M)$, where $u$ is adjacent to $(R,\bar v)$ if $u\in\bar v$. 
Note that the incidence graph of a graph is obtained by subdividing every edge exactly once.
The incidence graph of a structure $\mathbf M$ is an interpretation of $\mathbf M$ (but not a simple interpretation\footnote{This is the reason why we typeset it as $\Inc$ and not as $\mathsf{Inc}$ like simple interpretations and transductions.}), and hence, if $\Mm$ is stable, then $\Inc(\Mm)$ is stable. The converse however is not true in general. 
For example, the class of linear orders is not stable, but the class of its incidence graphs is the class of all $1$-subdivided cliques, which is stable (but not monadically stable). 

\pagebreak
\subsection{Graph Theory}
\label{sec:GT}

\subsubsection{Weakly sparse classes}
A class~$\mathscr C$ of graphs is \ndef[graph class]{weakly sparse} (or \ndef{biclique-free}) if there exists an integer $t$ such that no graph in $\mathscr C$ contains the balanced bipartite graph $K_{t,t}$ with parts of size $t$ as a subgraph.  
More generally, a class $\mathscr M$ or relational structures is \ndef[class of structures]{weakly sparse} if 
$\Gaif(\mathscr M)$ is weakly sparse.
Sometimes, it will be useful to distinguish the special binary adjacency relation $E$ and to consider the $E$-reduct~$\mathbf M^E$ of a structure $\mathbf M$; in this setting, we introduce the following notation.
\begin{nota}[note=\ndef{biclique number},store=nota:bomega]
\IfRestatingF{Let $\mathbf M$ be a $\sigma$-structure, and let $E\in\sigma$ be a distinguished binary relation. We define }
\[
\bomega(\mathbf M)=\max\,\{t\colon K_{t,t}\subseteq \mathbf M^E\}.
\]
\IfRestatingF{
By extension, if  $\mathscr C$ is a class of tree-ordered  graphs, we define
\[
\bomega(\mathscr C)=\sup\,\{\bomega(\mathbf M)\colon  \mathbf M\in\mathscr C\}.
\]}
\end{nota}

The next (standard) lemma is an easy variation of Ramsey's theorem and shows that (informally speaking) the union of two weakly sparse graphs is again weakly sparse. 

\begin{lem}
Let $G_1=(V,E_1)$ and $G_2=(V,E_2)$ be two graphs on the same vertex set, and let $G=(V,E_1\cup E_2)$.
If $K_{t_1,t_1}$ is not a subgraph of $G_1$ and $K_{t_2,t_2}$ is not a subgraph of $G_2$, then $K_{t_1+t_2-1,2^{t_1+t_2-1}\max(t_1,t_2)}$ is not a subgraph of $G$.
\end{lem}
\begin{proof}
Let $t=2^{t_1+t_2-1}\max(t_1,t_2)$. 
Assume $K_{t_1+t_2-1,t}$ is a subgraph of $G$ and let $\{a_1,\dots,a_{t_1+t_2-1}\}$ and $\{b_1,\dots,b_t\}$ be the vertex sets of such a biclique in $G$.
We consider the vertices $a_1,\dots,a_{t_1+t_2-1}$. 
Let $B_1=\{b_1,\dots,b_t\}$. For $i=1,\dots,t_1+t_2-1$,
we consider the edges incident to $a_i$ and $B_i$ and let $c_i$ be $1$ if there are more edges in~$E_1$ connecting $a_i$ and $B_i$ than edges in $E_2$, and $c_i=2$, otherwise. Then, $B_{i+1}$ is the subset of $B_i$ 
connected to $a_i$ by an edge in $E_{c_i}$. Note that $B_{t_1+t_2-1}$ contains at least $\max(t_1,t_2)$ vertices. Moreover, either there exists $t_1$ values $c_i$ equal to $1$, or $t_2$ values $c_i$ equal to $2$. In the first case, we exhibit a $K_{t_1,t_1}$ in $G_1$, while in the second case we get a $K_{t_2,t_2}$ in $G_2$.
\end{proof}
\begin{cor}
\label{cor:ws}
    Let $\sigma$ be a finite signature, whose elements $R_1,\dots,R_k$ are binary relation symbols, and let $\mathscr M$ be a class of $\sigma$-structures. If, for each $i\in [k]$, the $R_i$-reduct of~$\mathscr M$ is weakly sparse, then $\mathscr M$ is weakly sparse.
\end{cor}

Let $G$ be a graph and let $t$ be a non-negative integer. 
A \ndef{subdivision} of $G$ is a graph~$J$ obtained from $G$ by replacing the edges of $G$ by (internally vertex-disjoint) paths. A vertex of $J$ corresponding to a vertex of $G$ is a \ndef{principal vertex} of $J$, while an internal vertex of a path corresponding to an edge of $G$ is a \ndef{subdivision} vertex. The $t$-\ndef[$t$-subdivision]{subdivision} of a graph $G$ is the subdivision of $G$ obtained by replacing each edge of $G$ by a path with $t$ internal vertices; a $(\le t)$-\ndef[$(\le t)$-subdivision]{subdivision} of a graph $G$ is a subdivision of $G$ obtained by replacing each edge of $G$ by a path with at most $t$ internal vertices.

\subsubsection{Shallow minors}

Let $r$ be a positive integer. 
A graph $H$ is a \emph{depth-$r$ minor} of a graph $G$ if there is a map $M$ that assigns to every vertex $v\in V(H)$ a connected subgraph $M(v) \subseteq G$ of radius at most $r$ 
and to every edge $e\in E(H)$ an edge  $M(e)\in E(G)$ such that $M(u)$ and $M(v)$ are vertex disjoint for distinct vertices $u,v\in V(H)$, and
if $e = uv \in E(H)$, then $M(e) = u'v'\in E(G)$ for vertices $u'\in M(u)$ and $v'\in M(v)$.
We denote by $G\shm r$ the class of all the minors of $G$ at depth $r$. In particular, $G\shm 1$ denotes the class of all the minors of $G$ obtained by contracting a star forest.
By extension, if $\mathscr C$ is a class of graphs, we define
$\mathscr C\shm r=\bigcup_{G\in\mathscr C}G\shm r$.
A class $\mathscr C$ of graphs is \ndef[class of graphs]{nowhere dense} if, for every integer $r$, the class $\mathscr C\shm r$ is not the class of all graphs.
More generally, a class~$\mathscr M$ of relational structures is \ndef[class of structures]{nowhere dense} if 
$\Gaif(\mathscr M)$ is nowhere dense.

The following characterization of nowhere dense classes of graphs will be useful.
\begin{theorem}[\cite{DVORAK2018143}]
\label{thm:indND}
    A weakly sparse class of graphs $\mathscr C$ is nowhere dense if, and only if,  for every non-negative integer $t$,  there exists an integer $n_t$ such that no graph in $\mathscr C$ contains the $t$-subdivision of $K_{n_t}$ as an induced subgraph.
\end{theorem}

A class $\mathscr C$ of graphs has \ndef{bounded expansion} if , for every integer $t$, there exists an integer $d$ such that $G^{(t)}$ is a subgraph of no graph in $\mathscr C$ if $G$ has average degree at least $d$.

For an in-depth study of shallow minors, nowhere dense classes, and related notions, we refer the interested reader to \cite{Sparsity}.

\subsection{Order Theory}
A (strict) \ndef{partial order} $\prec$ is an irreflexive, asymmetric and transitive binary relation. A \ndef{poset} (or \ndef{partially ordered set)} is a set equipped with a partial order.
In this paper, all posets will be finite unless explicitely stated.
When a partial order is clear from the context, we denote by $\parallel$ the associated \ndef{incomparability} relation: $a\parallel b$ if neither $a\prec b$, nor $a=b$, nor $a\succ b$ holds. A subset of a poset is a \ndef{chain} (resp.\ an  \ndef{antichain}) if all its elements are pairwise comparable (resp.\ non-comparable).
A chain is \ndef{saturated} if  no element can be added between two of its elements without losing  the property of the set of being a chain.
An element~$b$ \emph{covers} an element $a$, and we note $a\prec:b$, if $a\prec b$ and there exists no $z$ with 
$a\prec z\prec b$. If $a\prec: b$, we say that the pair $(a,b)$ is a \ndef{cover}. 
The \ndef{cover graph} of a strictly partially ordered set $(V,\prec)$ is the graph with vertex set~$V$, where~$a$ and~$b$ are adjacent if either $(a,b)$ or $(b,a)$ is a cover. 
The \ndef{Hasse diagram} of a strictly partially ordered set $(V,\prec)$ is an upward 
drawing of the cover graph of $\prec$ in the plane.

A strict partial order $\prec$ is a \ndef{tree-order} if it has a (unique) minimum $r$ (the \ndef{root}) and, for every element $x\neq r$ the set of the elements smaller than $x$ forms a chain. (Note that the Hasse diagram of a tree-order is a tree.) The \ndef{infimum} (or \ndef{least common ancestor} $x\meet y$ of two elements $x$ and $y$ of a tree-order is the maximum element smaller than both $x$ and $y$. The \ndef{parent} of a non-root element $x$ is the (unique) element $\pi(x)$ such that $\pi(x)\prec:x$.

A \ndef{convex subset} of a poset $(V,\prec)$ is a subset $X$ of $V$ such that if $a\prec b$ and $a,b\in X$, then every $z\in V$ with $a\prec z\prec b$ also belongs to $X$. Note that a subset of a tree-order $\prec$ is convex if and only if it is the vertex set of a subtree of the cover graph of $\prec$. 
An \ndef{interval} of a finite poset is a convex subset of the form $\{x\colon a\preceq x\preceq b\}$. Note that on a finite tree-ordered set, intervals are the same as saturated chains.

In this paper, we shall use the notations $\prec$, $\prec:$, $\meet$ for the tree-orders and 
$<$ for the linear orders.

\subsection{More Model Theory}

\subsubsection{Atomic types and order types}
Let $\mathbf M$ be a relational structure with signature~$\sigma$.
Given a  tuple 
$\bar{v} = (v_1, \ldots, v_k)$ of elements of $\mathbf M$. 
The \ndef{atomic type} of  
$\bar{v}$ in $\mathbf M$, denoted as $\atp_{\mathbf M}(v_1, \ldots, v_k)$,  is the quantifier-free formula 
$\alpha(x_1, \ldots, x_k)$ defined as the conjunction of all literals 
$\beta(x_1, \ldots, x_k)$ (that is, formulas 
$x_i=x_j$, 
$R(x_{i_1}, \dots,x_{i_r})$ with $R\in\sigma$ with arity $r$ and $i_1,\dots,i_r\in[k]$)
 or their negations such that $\mathbf M \models \beta(\bar{v})$.
When clear from context, we write $\atp(\bar v)$ instead of $\atp_{\mathbf M}(\bar v)$. 

For elements $a,b$ of a linearly ordered set $(A, <) $, let $\otp(a, b) \in \{<, =, >\}$  indicate whether $a < b$, $a = b$, or $a > b$ holds. 
For $\ell \geq 1$  and an $\ell$-tuple of elements  $a_1, \ldots, a_\ell$  of a linearly ordered set $(A, <) $, we define the \ndef{order type} of  $(a_1, \ldots, a_\ell)$, denoted $\otp(a_1, \ldots, a_\ell)$, as the tuple  $(\otp(a_i, a_j))_{1 \leq i < j \leq \ell}$.

\subsubsection{Stability and dependence on weakly sparse classes of structures}
\label{sec:stdep}

\begin{theorem}[store=thm:Adler,note=\cite{Adler2013}]
	Let $\sigma$ be a finite binary relational signature.
	For a monotone class $\mathscr M$ of finite $\sigma$-structures, the following are equivalent:
	\begin{enumerate}
		\item $\mathscr M$ is dependent;
		\item $\mathscr M$ is monadically dependent;
		\item $\mathscr M$ is stable;
		\item $\mathscr M$ is monadically stable;
		\item $\mathscr M$ is nowhere dense.
	\end{enumerate}
\end{theorem}
We shall see how this theorem extends to weakly sparse  hereditary classes of general relational structures in \zcref{sec:Gstdep}.

\section{Tree-ordered weakly sparse  
graphs}
\label{sec:TOWS_G}
In this section, we introduce tree-ordered graphs and derived notions that will be used all along this paper.
\begin{ndefi}[Tree-ordered graphs]
\label{def:TO_G}
Let $\sigma$ be a finite relational signature.
A \ndef[graph]{tree-ordered} graph is a $\{\prec, E\}$-structure, where $\prec$ is a tree-order and $E$ is a symmetric (binary) adjacency relation.
\end{ndefi}

\begin{ndefi}[TOWS]
A class $\mathscr C$ of tree-ordered graphs is \ndef[class]{tree-ordered weakly sparse} (or \ndef[class]{TOWS}) if the class 
$\mathscr C^E$ (that is, the $E$-reduct of $\mathscr C$)
is weakly sparse. By extension, when a bound of the maximal order of a balanced complete bipartite is assumed to be implicitly fixed, we shall speak about \ndef[graph]{tree-ordered  weakly sparse} graphs (or \ndef[graph]{TOWS} graphs).
\end{ndefi}

\subsection{The tree-ordered incidence graph}

\begin{ndefi}
\label{def:TInc_G}
The \ndef[of a tree-ordered graph]{tree-ordered incidence graph} $\TInc(\mathbf M)$ of a tree-ordered graph $\mathbf M$ is 
the tree-ordered graph, whose $E$-reduct is the $1$-subdivision of $\mathbf M^E$, and whose 
$\prec$-reduct is the tree-order obtained from $\mathbf M^\prec$ by adding all subdivision vertices as an antichain (comparable with the root only).
\end{ndefi}
\begin{lem}
    \label{lem:inc2tows_G}
    There exists a transduction~$\mathsf T$ such that, for every TOWS graph~$\mathbf M$, $\mathbf M\in\mathsf T(\TInc(\mathbf M))$.
\end{lem}

\begin{proof}
    The transduction $\mathsf T$ uses a unary predicate $V$,
  and  is defined by $\nu(x):=V(x)$ and $\rho_\prec(x,y):=(x\prec y)$ and $\rho_E(x_1,x_2):=(x_1\neq x_2)\wedge\bigl(\exists t \  E(x_1,t)\wedge E(x_2,t)\bigr)$.
    
    Let $\mathbf M$ be a tree-ordered graph.
    The vertices of $\TInc(\mathbf M)$ that belong to the domain  $\mathbf M$ are marked by $V$.
    With this marking, it is easily checked that $\mathbf M\in\mathsf T(\TInc(\mathbf{M}))$.
\end{proof}

\subsection{The fundamental graph}
\label{sec:fund}
Every graph $G$ defines a binary matroid $M(G)$ (its \ndef{cycle matroid}), whose independent sets are the edge sets of the forests of $G$. To every bipartition $(Y,Z)$ of $M(G)$ (that is, of the edge set of $G$) into a basis $Y$ (that is, a spanning forest of $G$) and its complement is associated a \ndef{fundamental graph} (also called \ndef{interlacement graph} or \ndef{dependence graph}) $\Lambda_{G}(Y,Z)$, which is a bipartite graph with parts $Y$ and $Z$, where $e\in Y$ is adjacent to $f\in Z$ if~$e$ belongs to the unique cycle in $Y\cup\{f\}$. See, for instance,
\cite{welsh} for an comprehensive introduction to matroid theory and
\cite{krogdahl1977dependence,bouchet2001multimatroids} for more details on fundamental graphs.

The main properties of $\Lambda_{G}(Y,Z)$ are as follows:
\begin{itemize}
\item For $e\in Y$ and $f\in Z$, $e$ is adjacent to $f$ in $\Lambda_{G}(Y,Z)$ if and only if $Y-e+f$ is a basis of $M(G)$. In such a case, we have
\[\Lambda_{G}(Y-e+f,Z-f+e)=\Lambda_{G}(Y,Z)\wedge ef,
\]
where $\wedge$ denotes here the pivoting operation.
\item if $e\in Y$, then 
\[\Lambda_{G/e}(Y-e,Z)=\Lambda_{G}(Y,Z)-e.
\]
\item if $f\in Z$, then 
\[\Lambda_{G\setminus f}(Y,Z-f)=\Lambda_{G}(Y,Z)-f.
\]
\end{itemize}

In a TOWS graph, the cover graph of the tree-order may naturally play the role of a spanning tree, hence suggesting the next definition.

\begin{ndefi}
    The \ndef{fundamental graph} $\Lambda(\mathbf M)$ of a tree-ordered graph $\mathbf M$ is defined as $\Lambda_G(Y,Z)$, where
    $G$ is the multigraph obtained by adding to $\mathbf M^E$ the edges of the cover graph of $\mathbf M^\prec$, $Y$ is the set of the covers of $\mathbf M^\prec$, and $Z=E(\mathbf M)$. 
\end{ndefi}

\begin{lem}
\label{lem:lambda_G}
There exists a transduction 
$\mathsf T_{\Lambda}$ such that, for every TOWS graph~$\mathbf M$, $\Lambda(\mathbf M)\in \mathsf T_\Lambda(\TInc\mathbf M))$.
\end{lem}
\begin{proof}
    The transduction $\mathsf{T}_\Lambda$ uses unary predicates $Y$ and $Z$.
    We define the transduction $\mathsf T_\Lambda$ by $\nu(x):=\top$ and
$$\rho_E(x,y):=Y(x)\wedge Z(y)\wedge 
    \bigl(\exists v,v'\ E(v,y)\\
    \wedge E(v',y)\wedge (v\meet v' \prec x)\wedge ((x\preceq v)\vee(x\preceq v'))
    \bigr),$$
where $v\meet v'$ denotes the infimum of $v$ and $v'$ in the tree-order, which is first-order definable.

    Let $\mathbf M$ be a tree-ordered graph. 
    We mark $Y$ the vertices of $\Inc(\mathbf M)$ that belong to the domain of $\mathbf M$ and by $Z$ the subdivision vertices.
    Identifying a cover $(x,y)$ with its large vertex $y$, we easily check that the interpreted graph is $\Lambda(\mathbf M)$.
\end{proof}

\subsection{Tree-ordered minors and {\tim}s}
\label{sec:minor_G}

Let $\mathbf M$ be a tree-ordered structure.
An \ndef[tree-ordered graph]{elementary $\prec$-contraction} of a tree-ordered graph $\mathbf M$ is obtained by identifying a vertex $v$ with the element $u$ it covers. A \ndef[tree-ordered graph]{$\prec$-contraction} is a sequence of elementary $\prec$-contractions.

\begin{ndefi}[\tim]
An \ndef[[tree-ordered graph]{\tim} of a tree-ordered graph $\mathbf M$ is a $\prec$-contraction of an induced tree-ordered subgraph of~$\mathbf M$ (including the root). 
We denote by $\Cont(\mathbf M)$ the set of all the {\tim}s of $\mathbf M$, and, by extension, define $\Cont(\mathscr C)=\bigcup_{\mathbf M\in\mathscr C}\Cont(\mathbf M)$.
\end{ndefi}

\begin{fact}
    An {\tim} of an {\tim} of a tree-ordered graph $\mathbf M$ is an {\tim} of $\mathbf M$, that is:
    $$\Cont(\Cont(\mathbf M))=\Cont(\mathbf M).$$
\end{fact}
\begin{proof}
    It is enough to prove the property when the first {\tim} is obtained by an elementary $\prec$-contraction and the second by a deletion.
    Let $\mathbf M$ be a tree-ordered graph~$\mathbf M$, let $(u,v)$ be a cover of $\mathbf M^\prec$, let $w$ be a vertex of $\mathbf M$ different from $v$ (as the $\prec$-contraction identifies $v$ with $u$), and let $\mathbf N$ be the tree-ordered structure obtained from $\mathbf M$ by the $\prec$-contraction of $(u,v)$ followed by the deletion of $w$.
    
    If $w\neq u$, the two operations clearly commute; otherwise, $\mathbf N$ can be obtained from $\mathbf M$ by deleting both $u$ and $v$.
\end{proof}

An \ndef{elementary deletion} of a tree-ordered graph $\mathbf M$  is obtained by deleting an edge in $E(\mathbf M)$. A \ndef{deletion} is a sequence of elementary deletions.

\begin{ndefi}[Tree-ordered minor]
\label{def:TminorG}
A \ndef{tree-ordered minor} of a tree-ordered graph $\mathbf M$ is a tree-ordered graph obtained from an {\tim} of $\mathbf M$ by a sequence of elementary deletions. For a tree-ordered graph $\mathbf M$, we denote by $\Minor(\mathscr M)$ the class of all the tree-ordered minors of $\mathbf M$ and extend the notation to classes of tree-ordered graphs.
\end{ndefi}

It is easily checked that a tree-ordered minor of a tree-ordered minor of a tree-ordered graph $\mathbf M$ is a tree-ordered minor of $\mathbf M$.

Fundamental graphs are designed to encode tree-ordered minor operations in a purely structural way:
both $\prec$-contractions and deletions correspond to deleting vertices in the fundamental graph.
As a consequence, the minor order of a tree-ordered graphs can be studied as an induced-substructure order. 
This is captured by the following two facts.

\begin{fact}
    Let $\mathbf M$ be a tree-ordered graph. 
    Then,
    $\Minor(\mathbf M)$ (ordered by tree-ordered minor) is isomorphic to
    the poset of all induced subgraphs of $\Lambda(\mathbf M)$ (ordered by induced subgraph). 
\end{fact}

\begin{fact}
Let $\mathbf M$ be a tree-ordered graph,
let $G$ be the graph whose edge set is the union of $E(\mathbf M)$ and the set $Y$ of the covers of $\mathbf M$. 
Then, $H\in\Minor(\mathbf M)^E$ if and only if 
 $H$ is a minor of $G$ obtained by contracting some edges in $Y$ and deleting some edges not in $Y$.
\end{fact}
\section{Characterization by induced substructures}
\label{sec:dep_bin}

Dreier,  M\"ahlmann  and  Toru\'nczyk in~\cite{dreier2024flipbreakability} characterized monadically dependent classes of binary structures by forbidden combinatorial patterns, called \ndef{transformers}, which generalize the notion of $h$-subdivided bicliques
(definitions will be given in a moment). 

\begin{theorem}[store=thm:dependent,note={\cite{dreier2024flipbreakability}}]
	A class $\mathscr C$ of binary structures is  monadically dependent if and only if for every $h\in\mathbb N$ there is some $n$ such that 
	no structure $\mathbf M\in\mathscr C$ contains a minimal transformer of length $h$ and order $n$.
\end{theorem}

In the case of hereditary classes of graphs, this materializes in a
family of minimal forbidden induced subgraphs as proved in 
\cite{dreier2024flipbreakability}. Similarly, for tree-ordered weakly
sparse graph classes we derive an explicit family of substructure obstructions,
introducing \emph{clean twisters}, strengthening the notion of transformers.

The first main result of this section is the following characterization theorem.

\getkeytheorem{thm:dependentTOWS}

From this theorem, we will derive that every independent
class of TOWS graphs includes a class constructed from a finite tree-ordered structure, which we call a \emph{core}. This   result, \zcref{thm:lics_twists}, is the main result of this section. 

\subsection{Transformers and meshes in binary structures}
\label{sec:transformers}
Let us recall the needed definitions, introduce some new ones, and prove some preliminary lemmas.

\begin{defi}[mesh, submesh]
	Let $I,J$ be two non-empty totally ordered sets.

An  \emph{$(I,J)$-mesh} (or simply a \ndef{mesh} when $I$ and $J$ are clear from the context) in a binary structure $\mathbf M$ is an injective function $\mu:I\times J\rightarrow M$. 
We define $V(\mu):=\mu(I\times J)$. For a mesh $\mu$, we define $\mu^\top:J\times I\rightarrow M$ by $\mu^\top(j,i)=\mu(i,j)$. 

A \ndef{submesh} of $\mu$ is the restriction   $\mu|_{I'\times J'}$ of $\mu$ to $I'\times J'$, where $I'\subseteq I$ and $J'\subseteq J$.
\end{defi}

\begin{defi}[vertical and horizontal mesh]
Let $\mu$ be an $(I,J)$-mesh in a binary structure $\mathbf M$. Then $\mu$ is \ndef{vertical} in $\mathbf M$ if  there is a function $a:I\rightarrow M$ such that
\begin{itemize}
    \item $\atp(\mu(i,j),a(i'))$ depends only on $\otp(i,i')$ for all $i,i'\in I$ and $j\in J$, and
    \item $\atp(\mu(i,j),a(i'))$ is not the same for all $i,i'\in I$ and $j\in J$.
\end{itemize}
Any function $a$ satisfying the above is a \ndef{vertical guard} of $\mu$.

We say that $\mu$ is \ndef{horizontal} if $\mu^\top$ is vertical.
A \ndef{horizontal guard} of $\mu$ is a vertical guard of $\mu^\top$.
\end{defi}

\begin{defi}[pairs of meshes]
Let $\mu,\mu': I\times J\rightarrow M$ be meshes in a binary structure $\mathbf M$. We say that~$(\mu,\mu')$ is
\begin{itemize}
	\item \ndef[pair of meshes]{regular} if $\atp(\mu(i,j),\mu'(i',j'))$ depends only on 
	$\otp(i,i')$ and $\otp(j,j')$;
	\item \ndef[pair of meshes]{homogeneous} if $\atp(\mu(i,j),\mu'(i',j'))$ is the same for all $i,j,i',j'$;
	\item \ndef{conducting} if 
    either $|I|=|J|\leq 3$ or 
    $(\mu,\mu')$ is regular but not homogeneous.
\end{itemize}
\end{defi}

\begin{defi}[store=def:min_transfo, note=minimal transformer]
	Let $h$ be an integer. A \emph{minimal $(I,J)$-transformer} (or simply a \ndef{minimal transformer} when $I$ and $J$ are clear from the context) of length $h$ is a sequence
	$(\mu_1,\dots,\mu_h)$ of $(I,J)$-meshes, with $|I|, |J|\geq 4$, such that
	\begin{itemize}
		\item $\mu_s$ is vertical if and only if $s=1$,
		\item $\mu_s$ is horizontal if and only if $s=h$,
		\item $\mu_s\neq \mu_t$ if $s\neq t$,
        \item $(\mu_s, \mu_t)$ is regular for all $s, t \in [h]$ (including for $s=t$),
		\item $(\mu_s,\mu_t)$ is conducting if $|s-t|=1$,
		\item $(\mu_s,\mu_t)$ is homogeneous if $|s-t|>1$.
	\end{itemize}
	A \emph{minimal transformer of order $n$} is a minimal $(I,J)$-transformer with $|I|=|J|=n$. (See \zcref{fig:transformer} for an example.)
\end{defi}

\begin{figure}[h!t]
    \centering
    \includegraphics[width=\linewidth]{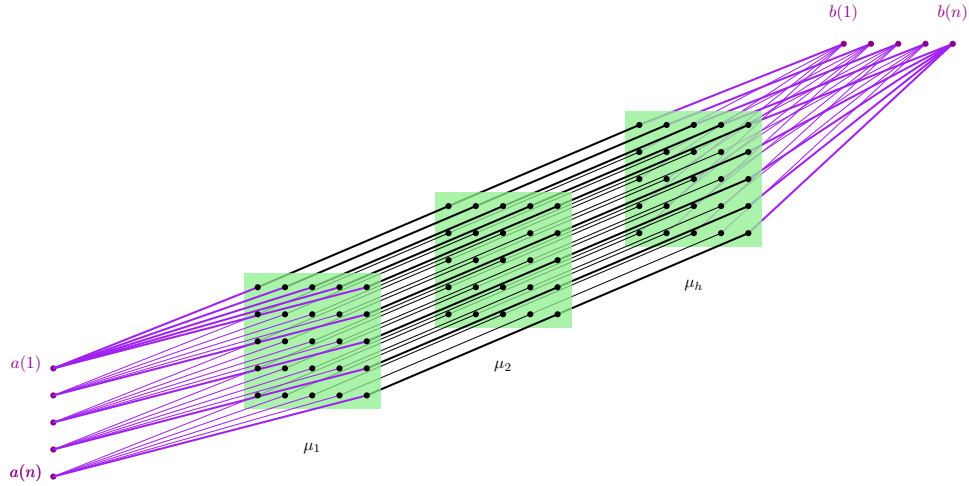}
    \caption{A minimal transformer of order $n$ (here, a subdivision of $K_{n,n}$).}
    \label{fig:transformer}
\end{figure}

The following fact shows that the condition $\mu_s\neq \mu_t$ is equivalent to the conditions that $V(\mu_s)$ and $V(\mu_t)$ are disjoint.
\begin{fact}
\label{fact:disjoint}
Let $\mu,\mu'$ be $(I,J)$-meshes with $|I|,|J|\geq 4$. 
	Assume $(\mu,\mu')$ is regular. Then  either $\mu=\mu'$ or $V(\mu)\cap V(\mu')=\emptyset$.
\end{fact}
\begin{proof}
	If $\mu(i,j)=\mu'(i,j)$ for some $i,j$, then $\mu(i,j)=\mu'(i,j)$ for all $i,j$, by regularity, and hence $\mu=\mu'$. 
	
	Otherwise, if $\mu(i,j)=\mu'(i',j')$ for some $i,j\neq i',j'$, say $i<i'$, then $\mu(i,j)=\mu(i'',j')$ for all $i<i''$ (we may assume that $i'$ is not maximal, as otherwise we can choose a different pair as $|I|,|J|\geq 4$). This is not possible as $\mu$ is injective. 
	Hence, in this case $V(\mu)\cap V(\mu')=\emptyset$.
\end{proof}

Note that if $(\mu_1,\dots,\mu_h)$ is a minimal $(I,J)$-transformer, then its \ndef{transpose} $(\mu_h^\top,\dots,\mu_1^\top)$ is a minimal $(J,I)$-transformer.

\medskip

In our study of tree-ordered weakly sparse graphs it will be convenient to introduce some variations of the above definitions.

\begin{ndefi}
	A pair $(\mu,\mu')$ of $(I,J)$-meshes in a binary structure $\mathbf M$  is
	\begin{itemize}
		\item 	 \ndef{quasi-homogeneous} if $(\mu,\mu')$ is regular and 
		$\atp(\mu(i,j),\mu'(i',j'))$ is the same for all $i,i'\in I$ and $j,j'\in J$ with $(i,j)\neq(i',j')$,
		\item	\ndef{matching} if $(\mu,\mu')$ is quasi-homogeneous but not homogeneous.
	\end{itemize}
\end{ndefi}
\begin{ndefi}	
	An $(I,J)$-mesh $\mu$ in a binary structure $\mathbf M$ is
\begin{itemize}
	\item \ndef[mesh]{regular} if $(\mu,\mu)$ is regular,
	\item \ndef[mesh]{homogeneous} if $(\mu,\mu)$ is quasi-homogeneous,
    \item \ndef{pseudo-vertical} if it is regular and the restriction of $\mu$ to\\ $I\times(J\setminus\{\min J,\max J\})$ is vertical;
    \item \ndef{inner-vertical} if it is pseudo-vertical, with vertical guard $a$ such that $a(i) := \mu(i,\min J)$ or $a(i) := \mu(i, \max J)$;
    \item \ndef{pseudo-horizontal} if $\mu^\top$ is pseudo-vertical.
   	\item \ndef{inner-horizontal} if $\mu^\top$ is inner-vertical;
\end{itemize}
\end{ndefi}

\begin{ndefi}
	Let $\mu$ be a mesh in a binary structure $\mathbf M$ and let $X$ be a subset of $M$. We say that the pair $(\mu,X)$ is \ndef[pair of a mesh and a set]{homogeneous} if 
	$\atp(u,v)$ is constant for all pairs $(u,v)\in V(\mu)\times X$ with $u\neq v$. We say that the pair $(X,\mu)$ is \emph{homogeneous} if $(\mu,X)$ is homogeneous.
\end{ndefi}

The next lemma determines the three possible types of regular pairs of regular meshes that can occur in a binary structure, and somehow justifies the introduction of the notions of pseudo-vertical, pseudo-horizontal, and quasi-homogeneous meshes.
   
\begin{lem}
\label{lem:quasi_hom}
Let $(\mu,\mu')$ be a regular pair of regular $(I,J)$-meshes in a binary structure $\mathbf M$ with $|I|,|J|\geq 2$. Then  one of the following holds:
    \begin{enumerate}
        \item $\mu$ and $\mu'$ are both pseudo-vertical;
        \item $\mu$ and $\mu'$ are both pseudo-horizontal;
        \item $(\mu,\mu')$ is quasi-homogeneous.
    \end{enumerate}
\end{lem}
\begin{proof}
We first show that the study of a regular pair of meshes can be reduced to the study of  a pair of $2\times 2$-submeshes.
\begin{claim}
  \label{cl:restr}
    Let $(\mu,\mu')$ be a regular pair of $(I,J)$-meshes, let $I'\subseteq I$ and $J'\subseteq J$ be such that $|I'|=|J'|=2$. Then  the mapping $(i,j,i',j')\mapsto \atp(\mu(i,j),\mu'(i',j'))$ (with $i,i'\in I$ and $j,j'\in J$) is fully determined by its restriction to $i,i'\in I'$ and $j,j'\in J'$.

    In particular, $(\mu,\mu')$ is homogeneous (resp.\ quasi-homogeneous) if and only if $(\mu|_{I'\times J'},\mu'|_{I'\times J'})$ is homogeneous  (resp.\ quasi-homogeneous).  
\end{claim}
\begin{clproof}
    The restriction to $I'\times J'$  covers all the possible pairs of order types $(\otp(i,i'),\otp(j,j'))$.
\end{clproof}

Now assume for contradiction that the conclusion of the lemma does not hold. First, we prove $\mu$  is pseudo-vertical or pseudo-horizontal. Since $(\mu,\mu')$ and  $\mu$ are regular, this will follow from the next claim.

\begin{claim}
\label{cl:quasi-hom0}
	Let $(\mu,\mu')$ be a regular pair of $(I,J)$-meshes in a binary structure~$\mathbf M$ with $|I|,|J|\geq 2$.

    Assume $\mu$ is regular.
	Then  either  $\mu$ is  pseudo-vertical or pseudo-horizontal, 	or else $(\mu,\mu')$ is  quasi-homogeneous.
\end{claim}
\begin{clproof}
Since  $(\mu,\mu')$ and $\mu$ are regular, the following four cases are easy to check:
	\begin{enumerate}[label=\emph{(\roman*)}]
		\item If  $\atp(\mu(i,j),\mu'(i',\min J))$ is not constant for $i,i'\in I$ and $j\in J\setminus\{\min J\}$, then $\mu$ is pseudo-vertical (witnessed by $a(i')=\mu'(i',\min J)$);
		\item if  $\atp(\mu(i,j),\mu'(i',\max J))$ is not constant for $i,i'\in I$ and $j\in J\setminus\{\max J\}$, then $\mu$ is pseudo-vertical (witnessed by $a(i')=\mu'(i',\max J)$);
		\item if $\atp(\mu(i,j),\mu'(\min I,j'))$ is not constant for $i\in I\setminus\{\min I\}$ and $j,j'\in J$, then $\mu$ is pseudo-horizontal 
		(witnessed by $b(j')=\mu'(\min I,j')$);
		\item if $\atp(\mu(i,j),\mu'(\max I,j'))$ is not constant for $i\in I\setminus\{\max I\}$ and $j,j'\in J$, then $\mu$ is pseudo-horizontal 
		(witnessed by $b(j')=\mu'(\max I,j')$).
	\end{enumerate}
   Now, assume that none of the above cases occurs. 
   Then   
   	\begin{mycases}
		\item\label{it:t1} 
      $\atp(\mu(i,\max J),\mu'(i',\min J))$ is the same  for  $i,i'\in I$ (by $\neg (i)$);
		\item\label{it:t2}    $\atp(\mu(i,\min J),\mu'(i',\max J))$ is the same  for  $i,i'\in I$ (by $\neg (ii)$);
		\item\label{it:t3} $\atp(\mu(\max I,j),\mu'(\min I,j'))$ is the same  for  $j,j'\in J$ (by $\neg (iii)$);
		\item\label{it:t4}  $\atp(\mu(\min I,j),\mu'(\max I,j'))$ is the same  for  $j,j'\in J$ (by $\neg (iv)$).
	\end{mycases}

Considering $\atp(\mu(\min I,\max J),\mu'(\max I, \min J))$, we get that the atomic types in \zcref{it:t1,it:t4}  are the same. 
Similarly, the atomic types in \zcref{it:t2,it:t3} are the same (consider $\atp(\mu(\max I,\min J),\mu'(\min I, \max J))$) and those in \zcref{it:t2,it:t4} are also the same (consider $\atp(\mu(\min I,\min J),\mu'(\max I, \max J))$).

We deduce that
$\atp(\mu(i,j),\mu'(i',j'))$  is constant for  $i,i'\in \{\min I, \max I\}$ and  $j,j'\in \{\min J, \max J\}$ with $(i,j)\neq (i',j')$, and so
the restriction of $(\mu,\mu')$ to 
$\{\min I, \max I\}$ and  $ \{\min J, \max J\}$ is quasi-homogeneous.
Therefore, by \zcref{{cl:restr}},  $(\mu,\mu')$ is also quasi-homogeneous.
\end{clproof}
 
 Likewise, it follows from \zcref{cl:quasi-hom0} that $\mu'$  is pseudo-vertical or pseudo-horizontal.
 Thus, we may further assume, without loss of generality, that
$\mu$ is  pseudo-vertical but not  pseudo-horizontal and 
  $\mu'$ is pseudo-horizontal but not pseudo-vertical.
  Then  we have the following  four properties:
	\begin{enumerate}
		\item $\atp(\mu'(i',j'),\mu(i,\min J))$ depends only  on $\otp(i,i')$ if $j'\in J'$ (by regularity of $(\mu,\mu')$) and is constant for $i,i'\in I$ and $j'\in J\setminus\{\min J\}$ (for otherwise~$\mu'$ would be pseudo-vertical);
		\item similarly, $\atp(\mu'(i',j'),\mu(i,\max J))$ is constant for $i,i'\in I$ and $j'\in J\setminus\{\max J\}$;
		\item $\atp(\mu(i,j),\mu'(\min I,j'))$ is constant for $i\in I\setminus \{\min I\}$ and $j,j'\in J$ (for otherwise $\mu$ would be pseudo-horizontal);
		\item similarly, $\atp(\mu(i,j),\mu'(\max I,j'))$  is constant for $i\in I\setminus\{\max I\}$ and \mbox{$j,j'\in J$};
	\end{enumerate}

It easily follows from the above four properties that 
$\atp(\mu(i,j),\mu'(i',j'))$  is constant for $i,i'\in \{\min I, \max I\}$ and  $j,j'\in \{\min J, \max J\}$ with $(i,j)\neq (i',j')$.
Therefore, the restriction of $(\mu,\mu')$ to 
$\{\min I, \max I\}$ and  $ \{\min J, \max J\}$ is quasi-homogeneous.
Thus, by \zcref{cl:restr},  $(\mu,\mu')$ is quasi-homogeneous, contradicting our assumptions.
\end{proof}

\subsection{Transformers and meshes in TOWS graphs}

In the following lemma, we see that sparsity greatly reduces the the
complexity of transformers, since a conducting pair must necessarily
be a matching.
Recall the definition of the biclique number of a tree-ordered graph (\zcref{nota:bomega}): \getkeytheorem[body]{nota:bomega}

In order to avoid unnecessary complicated notations, we say that a mesh $\mu$ in a TOWS graph is a \emph{chain} (resp.\ an \emph{antichain}) if $V(\mu)$ is a chain (resp.\ an antichain) in the partial order $\prec$.

\pagebreak
\begin{lem}
\label{lem:inde}
	Let $(\mu,\mu')$ be a regular pair of $(I,J)$-meshes in a tree-ordered graph~$\mathbf M$.
	
	If $\min(|I|,|J|)\geq 2\bomega(\mathbf M)+2$, then 
	\[
	(i,j)\neq (i',j')\quad\Longrightarrow
	\neg E(\mu(i,j),\mu'(i',j')).
	\]
	
	Hence, if additionally both  $\mu$ and $\mu'$ are regular, then either $V(\mu)\cup V(\mu')$ is an independent set or $V(\mu)\cup V(\mu')$ induces a perfect matching $E(\mu(i,j),\mu'(i,j))$.
\end{lem}
\begin{proof}
	Assume there exists $(i,j)\neq (i',j')$ in $I\times J$ such that 
	$E(\mu(i,j),\mu'(i',j'))$. Let $t=\bomega(\mathbf M)+1$. As $\min(|I|,|J|)\geq 2t$, 
	 there exist distinct pairs $(i_1,j_1),\dots,(i_t,j_t)$, $(i_1',j_1'),\dots,(i_t',j_t')\in I\times J$ such that for every 
	$k\neq l\in [t]$ we have $\otp(i_k,i_l')=\otp(i,i')$ and $\otp(j_k,j_l')=\otp(j,j')$. By regularity, we deduce that the vertices $\mu(i_k,j_k)$ \mbox{($k\in [t]$)} and $\mu'(i_l',j_l')$ ($l\in [t]$) bi-induce a $K_{t,t}$ in $\mathbf M$, contradicting \mbox{$t=\bomega(\mathbf M)+1$}.
\end{proof}

\begin{lem}
    \label{lem:anti-qh}
    Let $(\mu,\mu')$ be a regular pair of  $(I,J)$-meshes in a tree-ordered graph~$\mathbf M$ with $|I|,|J|\geq \max(2\bomega(\mathbf M)+2,5)$.

    If both $\mu$ and $\mu'$ are antichains, then $(\mu,\mu')$ is quasi-homogeneous.
\end{lem}
\begin{proof}
    Let $i_1<i_2<\dots<i_5$ be elements of $I$ and $j_1<j_2<\dots<j_5$ be elements of~$J$.
    Assume for contradiction that $\mu(i_3,j_3)\succ \mu'(i',j')$, where $(i',j')\neq (i_3,j_3)$. Then  there exists
    $(i'',j'')\neq (i',j')$ such that 
    $\otp(i_3,i'')=\otp(i_3,i')$ and
    $\otp(j_3,j'')=\otp(j_3,j')$. Then  $\mu(i_3,j_3)\succ \mu'(i'',j'')$, thus
    $\mu'(i',j')$ and $\mu'(i'',j'')$ are comparable, contradicting the assumption that $\mu'$ is an antichain.
    
    Assume now for contradiction that $\mu(i_3,j_3)\prec \mu'(i',j')$, where $(i',j')\neq (i_3,j_3)$. 
    Then  there exists $(i'',j'')\neq (i_3,j_3)$ such that $\otp(i'',i')=\otp(i_3,i')$ and
    $\otp(j'',j')=\otp(j_3,j')$. Then  $\mu(i_3,j_3)$ is comparable with $\mu(i'',j'')$, contradicting the assumption that $\mu$ is an antichain.

    Thus, $(\mu,\mu')$ is quasi-homogeneous in $\mathbf M^\prec$.
    According to \zcref{lem:inde}, $(\mu,\mu')$ is quasi-homogeneous in $\mathbf M^E$, hence
    in $\mathbf M$.
    \end{proof}

\begin{table}[h!t]
	\begin{center}
		\begin{tabular}{|Mc|Mc||Mc|Mc|}
			\hlx{hv[1,3]}
			\multicolumn{2}{|m{.45\columnwidth}||}{\centering Lexicographic type}&\multicolumn{2}{m{.45\columnwidth}|}{\centering Anti-lexicographic type}\\
			\hlx{v[1,3]hv}
			<_{I,J}&(i<i')\vee (i=i')\wedge(j<j')&<_{J,I}&(j<j')\vee (j=j')\wedge(i<i')\\
			\hlx{vhv}
			<_{\bar I,J}&(i>i')\vee (i=i')\wedge(j<j')&<_{\bar J,I}&(j>j')\vee (j=j')\wedge(i<i')\\
			\hlx{vhv}
			<_{I,\bar J}&(i<i')\vee (i=i')\wedge(j>j')&<_{J,\bar I}&(j<j')\vee (j=j')\wedge(i>i')\\
			\hlx{vhv}
			<_{\bar I,\bar J}&(i>i')\vee (i=i')\wedge(j>j')&<_{\bar J,\bar I}&(j>j')\vee (j=j')\wedge(i>i')\\
			\hlx{vh}
			\vgap[0,1,2,3,4]{12pt}
			\hlx{hv[1,3]}
			\multicolumn{2}{|c||}{Vertical type}&\multicolumn{2}{c|}{Horizontal type}\\
			\hlx{v[1,3]hv}
			<_{J}&(i=i')\wedge (j<j')&<_{I}&(j=j')\wedge(i<i')\\
			\hlx{vhv}
			<_{\bar J}&(i=i')\wedge (j>j')&<_{\bar I}&(j=j')\wedge(i>i')\\
			\hlx{vh}
			\vgap[0,1,2,3,4]{6pt}
		\end{tabular}
	\end{center}
	\caption{Different types of partial orders defined on $I\times J$. The entries correspond to a partial order $<_*$ and to the formula defining $(i,j)<_*(i',j')$.}
	\label{tab:orders}
\end{table}

In the following, we may as well assume that $|I|, |J| \ge 2\bomega(\mathbf M)+2$, so that in particular a regular
$(I,J)$-mesh must induce an independent set in $\mathbf M^E$.

In the following series of lemmas, we use the characteristics of
tree orders to gain more information on the structure of regular pairs.

\begin{lem}
	\label{lem:reg}
	Let $\mu$ be a regular $(I,J)$-mesh in a tree-ordered graph $\mathbf M$, \mbox{$|I|,|J|\ge 3$}. 
	
	Then  either $\mu$ is homogeneous (and then $\mu$ is an antichain of $\mathbf M^\prec$), or the restriction of $\prec$ to $V(\mu)$ is one of the orders shown in \zcref{tab:orders}. Then  $\mu$ is inner-vertical (if the order belongs to the left column) or inner-horizontal (if the order belongs to the right column).
\end{lem}

\begin{proof}
If $\mu$ is homogeneous, then  $V(\mu)$ is an antichain, as $\prec$ is antisymmetric.
	
	Assume $\mu$ is not homogeneous. Then  there exist $a,a',b,b'$ such that $\mu(a,b)\prec\mu(a',b')$ (with $(a,b)\neq (a',b')$).
	We consider the set $\mathcal P$ of all pairs $(\otp(a,a'),\otp(b,b'))$ for which $\mu(a,b)\prec\mu(a',b')$.

	First assume, that up to reversal of the order on $I$ and/or $J$, we have ${{(<,<)}\in\mathcal P}$.
	In the rest of this proof, let us assume, up to isomorphism, that $I = [n]$ and $J=[m]$ for $n=|I|,
    m=|J|$, with the
    natural orderings.
    Then  we have $\mu(i,j)\prec\mu(n,m)$ for all $i,j\leq 2$. As a consequence, 
	all the $\mu(i,j)$ with $i,j\leq 2$ are pairwise comparable (as $\prec$ is a tree-order). 
	For every $(i,j),(i',j')\in I\times J$ there exists
	$(a,b),(a',b')\in[2]^2$  such that $\otp(a,a')=\otp(i,i')$ and
	$\otp(b,b')=\otp(j,j')$. Hence, $\mu(i,j)$ and $\mu(i',j')$ are compared in $\prec$ the same way as $\mu(a,b)$ and $\mu(a',b')$
    in $\prec$. In other words, $V(\mu)$ is a chain for $\prec$.
	
	Then  looking at the relations between the $\mu(i,j)$ with $i,j\in [2]$ and using the fact that $V(\mu)$ is a chain, we get the  possibilities shown in \zcref{fig:case1}, leading to different types of lexicographic or anti-lexicographic orders.
	\begin{figure}[h!t]
	\begin{center}
	\includegraphics[width=\columnwidth]{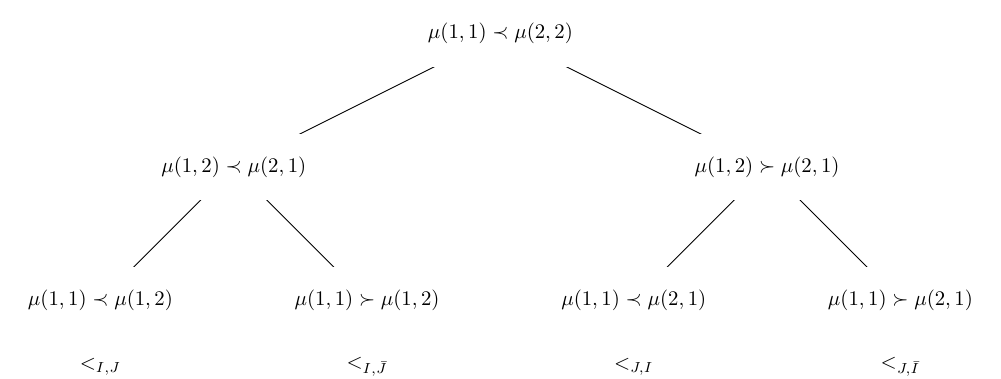}
\end{center}
	\caption{First case of the proof of \zcref{lem:reg}.}
	\label{fig:case1}
	\end{figure}
	
Otherwise, we can assume that every pair in $\mathcal P$ contains the equality, and we are left with an order of horizontal or vertical type (See \zcref{fig:case2}).
		\begin{figure}[h!t]
	\begin{center}
\includegraphics[width=.75\columnwidth]{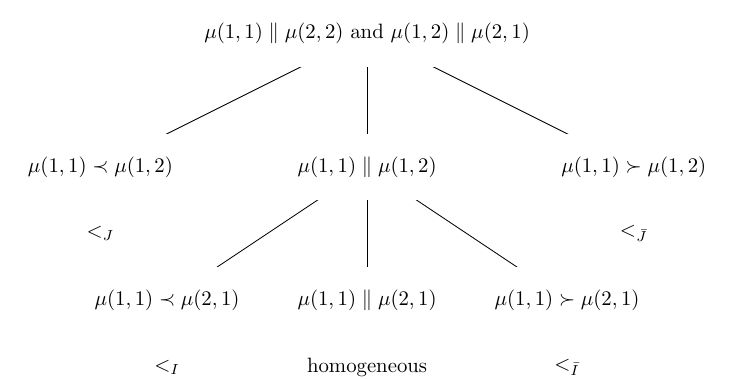}
	\end{center} 
		\caption{Second case of the proof of \zcref{lem:reg}.}
	\label{fig:case2}
	\end{figure}
	\end{proof}

\begin{lem}\label{lem:matching}
 For every tree-ordered graph $\mathbf M$, 
    every conducting pair $(\mu,\mu')$ of $(I,J)$-meshes in $\mathbf M$ 
     with $\min(|I|,|J|)\geq 2\bomega(\mathbf M)+2\geq 4$, $\mu$ inner-vertical, and $\mu'$ inner-horizontal is matching (more precisely: matching in $\mathbf M^E$ and homogeneous in~$\mathbf M^\prec$).
\end{lem}
\begin{proof}
    Assume for contradiction that $(\mu,\mu')$ is not homogeneous in $\mathbf M_m^\prec$. 
    Without loss of generality, there exists $i,i'$ in $I$ and $j,j'$ in $J$ such that $\mu(i,j)\prec \mu'(i',j')$.
    Let $i_1<i_2<i_3<i_4$ in $I$ and $j_1<j_2<j_3<j_4$ in $J$.
    While keeping $\otp(i,i')$ and $\otp(j,j')$ unchanged, we can assume $i=i_3$ and $j=j_2$. Hence, for every $x\in I$ and every $y\in J$ we have
    $$\mu(i_2,y)\prec \mu(i_3,j)\prec \mu'(i',j_2)\prec\mu'(x,j_3).$$
    
    As all the possible order types of $(i_2,y)$ and $(x,j_3)$ can be independently achieved, we deduce $V(\mu)\prec V(\mu')$, contradicting the hypothesis that $(\mu,\mu')$ is not homogeneous in $\mathbf M_m^\prec$. 

    As $(\mu,\mu')$ is conducting, if follows from \zcref{lem:inde} that $(\mu,\mu')$ is matching in~$\mathbf M^E$.
\end{proof}

\subsection{Twisters}
\label{sec:twisters}
We now introduce the notion of twister, which is a
technical definition allowing the iterative transformation of a transformer into clean twisters, which will be introduced in \zcref{sec:clean_twistsers}.

\begin{ndefi}[Twister] \label{def:twisters}
	Let $\mathbf M$ be a tree-ordered graph. 
	
	An  $(I,J)$-\ndef{twister} of length $h$ in $\mathbf M$ is  a sequence $(A,\mu_1,\mu_2,\dots,\mu_{h},B)$, where
	$\mu_1,\dots,\mu_h$ are $(I,J)$-meshes, and
	\begin{twprop}
    	\item\label{it:tw1} $\mu_1$ is vertical or inner-vertical;  
		\item\label{it:tw3} $\mu_h$ is  horizontal or inner-horizontal;
		\item\label{it:tw2} $\mu_s$ is an independent set in $\mathbf M^E$ and an antichain in $\mathbf M^\prec$ for every $1<s<h$;
		\item\label{it:tw4} $(\mu_s,\mu_t)$ is matching if $|s-t|=1$, with the possible following two exceptions, for which $(\mu_s,\mu_t)$ may be only conducting:
        \begin{itemize}
            \item $\{s,t\}=\{1,2\}$, $\mu_1$ is a chain, and $\mu_2$ is an antichain\footnote{\label{fn:1}Note that this condition follows from \zcref{it:tw2} when $h>2$.},
            \item or $\{s,t\}=\{h-1,h\}$, $\mu_h$ is a chain, and $\mu_{h-1}$ is an antichain\footnote{See  \zcref[nocap]{fn:1}.};
        \end{itemize}
		\item\label{it:tw5} $(\mu_s,\mu_t)$ is homogeneous if $|s-t|>1$;
		\item\label{it:tw6} $A$ is empty if $\mu_1$ is inner-vertical, and $A$ is the range of a vertical guard of $\mu_1$, otherwise;
		\item\label{it:tw7} $B$ is empty if $\mu_h$ is inner-horizontal, and $B$ is the range of a horizontal guard of $\mu_h$, otherwise;
		\item\label{it:tw8} the sets $A,V(\mu_1),\dots,V(\mu_h),B$ are disjoint;
		\item \label{it:tw9} the pairs $(A,\mu_s)$ are homogeneous if $s>1$;
		\item\label{it:tw10}  the pairs $(\mu_s,B)$ are homogeneous if $s<h$;
    \item \label{it:tw11} the pair $(A,B)$ is homogeneous,
		\item\label{it:tw12} each of $A$ and $B$ is independent (in $\mathbf M^E$), and either a chain or an antichain (in $\mathbf M^\prec$);
         \item\label{it:tw13} the root of $\mathbf M$ does not belong to  $A\cup B\cup\bigcup_{1\le s\le h} V(\mu_s)$;
    \item\label{it:tw14} only $\mu_1$ can be inner-vertical and only  $\mu_h$ can be inner-horizontal.
	\end{twprop}   
\end{ndefi}

Our aim is to extract large twisters in large transformers. Towards this end, we recall the following product version of Ramsey theorem.

\begin{lemma}[Ramsey theorem, product version]
	\label{lem:BCR}
	For every $k,\ell_1,\ell_2,m_1,m_2$, there exist  $n_1,n_2$ such that, for every coloring
	\[
	c:[n_1]^{\ell_1}\times [n_2]^{\ell_2}\rightarrow [k]
	\]
	there are sets $I_1\in\binom{[n_1]}{m_1}$ and $I_2\in\binom{[n_2]}{m_2}$ such that $c(\bar a,\bar b)$ depends only on $\otp(\bar a)$ and $\otp(\bar b)$, for all $\bar a\in I_1^{\ell_1}$ and $\bar b\in I_2^{\ell_2}$. \qed
\end{lemma}

We recall the definition of minimal transformers for convenience:

\getkeytheorem{def:min_transfo}

\begin{lem}
\label{lem:tr2tw}
	For every integer $h,m,n,w$ there exist integers $p,q$ such that if a tree-ordered graph $\mathbf M$ with $\bomega(\mathbf M)\leq w$ contains a minimal $(I,J)$-transformer of length $h$ with $|I|\geq p$ and $|J|\geq q$, then $\mathbf M$ contains an $(I',J')$-twister of length at most $h$ with $|I'|=n$ and $|J'|=m$.
\end{lem}

\begin{proof}
First, note that we can assume $\min(n,m)\geq \max(2w+2,5)$. 
As we shall obtain a twister from a transformer,
it will be useful to locally introduce a new definition encompassing these two notions.

Let $I$ and $J$ be two index sets, let $\mathbf M$ be a tree-ordered graph, and let $h,n$ be integers. An \emph{$(I,J)$-sequence} of $\mathbf M$ with length $h$ and order $n$ is a sequence $(A,\mu_1,\dots,\mu_h,B)$, where 
$\min(|I|,|J|)\geq n$, 
\begin{itemize}
    \item $(\mu_1,\dots,\mu_h)$ is an $(I,J)$-transformer,
    \item $\mu_1$ is inner-vertical or vertical (\zcref{it:tw1}),
    \item $\mu_h$ is inner-horizontal or horizontal (\zcref{it:tw3}),
    \item $(\mu_s,\mu_t)$ is homogeneous if $|s-t|>1$ (\zcref{it:tw5}),
    \item $A$ is either empty or the set of vertices of a vertical guard $a$ of $\mu_1$ (indexed by $I$),
    \item and $B$ is either empty or the set of vertices of a horizontal guard $b$ of $\mu_h$ (indexed by $J$).
\end{itemize}
The functions $a$ and $b$ are the \emph{indexing functions} of $A$ and $B$.

Moreover, if $I'\subseteq I$ and $J'\subseteq J$, the \ndef{restriction} of $(A,\mu_1,\dots,\mu_h,B)$ to $I'\times J'$ is the $(I',J')$-sequence defined by the restrictions of $\mu_1,\dots,\mu_h$ to $I'\times J'$, as well as the restrictions of $a$ to $I'$ and of $b$ to $J'$, if applicable. The restriction of a twister $\mathfrak T$ is called a \ndef{sub-twister} of $\mathfrak T$.

\pagebreak

The next claim is immediate from the definitions.

\begin{claim}\label{fact:restr2}
Let $(A,\mu_1,\dots,\mu_h,B)$ be an $(I,J)$-sequence, where $\mu_1,\dots,\mu_h$ are $(I,J)$-meshes in a tree-ordered graph $\mathbf M$ and $A$ and $B$ are subsets of vertices of $\mathbf M$ indexed by $I$ and $J$, respectively.

Let 
\[
    A'=\begin{cases}
        \emptyset&\text{if $\mu_1$ is inner-vertical,}\\
        A&\text{otherwise;}
        \end{cases}
        \quad\text{and}\quad
    B'=\begin{cases}
        \emptyset&\text{if $\mu_h$ is inner-horizontal,}\\
        B&\text{otherwise.}
    \end{cases}
\]
Then  $(A',\mu_1,\dots,\mu_h,B')$ is an $(I,J)$-sequence on $\mathbf M$ and
if some of \zcref{it:tw1,it:tw2,it:tw3,it:tw4,it:tw5,it:tw6,it:tw7,it:tw8,it:tw9,it:tw10,it:tw11,it:tw12,it:tw13,it:tw14} are satisfied by $(A,\mu_1,\dots,\mu_h,B)$, then they are also satisfied by the $(I,J)$-sequence
$(A',\mu_1,\dots,\mu_h,B')$. 
\hfill$\vartriangleleft$
\end{claim}

We now consider properties obtained by restricting $I$ and $J$ to some large subsets~$I'$ and $J'$.

\begin{claim}
\label{cl:regtw}
There exists a function $f_1:\mathbb N\times \mathbb N\rightarrow\mathbb N$ with the following property:

For all integers $h,n$ and every $(I,J)$-sequence $(A,\mu_1,\dots,\mu_h,B)$ in $\mathbf M$, with indexing functions $a$ and $b$ and $\min(|I|,|J|)\geq f_1(h,n)$,  
    there exists $I'\subseteq I$ and $J'\subseteq J$ with $\min(|I'|,|J'|)\geq n$, such that for every $1\leq s\leq h$, every $i,i'\in I'$, and every $j,j'\in J'$, we have
\begin{rprop}
    \item\label{it:r1} $\atp(\mu_s(i,j),a(i'))$ depends only on $\otp(i,i')$;
    \item\label{it:r2} $\atp(\mu_s(i,j),b(j'))$ depends only on $\otp(j,j')$;
    \item\label{it:r3} $\atp(a(i),a(i'))$ depends only on $\otp(i,i')$;
    \item\label{it:r4} $\atp(b(j),b(j'))$ depends only on $\otp(j,j')$;
    \item\label{it:r5} $\atp(a(i),b(j))$ is constant.
\end{rprop}    
\end{claim}
\begin{clproof}
        For $1\leq s\leq h$, define the coloring $\gamma_s^a$ of $I^2\times J$ by
\[
\gamma_s^a(i,i',j):=
\atp(\mu_s(i,j),a(i')).
\]

According to \zcref{lem:BCR} we can assume, by taking sufficient large sets $I$ and~$J$, that there exist arbitrarily large subsets $I'\subseteq I$ and $J'\subseteq J$ such 
 that $\gamma_s^a(i,i',j)$ depends only on $\otp(i,i')$ for $i,i'\in I'$ and $j\in J'$, that is \zcref{it:r1}. 
 By repeating a similar argument on appropriate mappings we obtain arbitrarily large sets $I'$ and~$J'$ satisfying all the required properties. The function $f_1$ is defined accordingly.
\end{clproof}

\begin{claim}
\label{cl:bigclean}
    For all integers $h,n$ and every $(I,J)$-sequence $(A,\mu_1,\dots,\mu_h,B)$ in~$\mathbf M$ with $\min(|I|,|J|)\geq f_1(h,n+1)$ and $2\bomega(\mathbf M)< n$,  
    there exists $I'\subseteq I$ and $J'\subseteq J$ with $\min(|I'|,|J'|)\geq n$, such that the restriction $(A',\mu_1',\dots,\mu_h',B')$ (with indexing functions still denoted $a$ and $b$) of  $(A,\mu_1,\dots,\mu_h,B)$ to $I'\times J'$ satisfies \zcref{it:r1,it:r2,it:r3,it:r4,it:r5} and \zcref{it:tw1,it:tw3,it:tw5,it:tw8,it:tw11,it:tw12,it:tw13}.
\end{claim}
\begin{clproof}
     According to \zcref{cl:regtw}, there exists subsets $I'\subseteq I$ and $J'\subseteq J$ with $\min(|I'|,|J'|)\geq f_1(h,\min(|I|,|J|))$ such that \zcref{it:r1,it:r2,it:r3,it:r4,it:r5} hold.
Note that \zcref{it:r5} is nothing but \zcref{it:tw11}.

Since $2\bomega(\mathbf M)< n$ and $\min(|I|,|J|)\geq n+1$, then 
 $\mathbf M^E[A'\cup B']$ is empty. By \zcref{it:r3,it:r4}, $\mathbf M^\prec[A']$ and $\mathbf M^\prec[B']$ are either chains or antichains.
 Thus, \zcref{it:tw12} holds.

 Then  the sets $A',V(\mu_1'),\dots,V(\mu_h'),B'$ are disjoint (by \zcref{fact:disjoint},
as equalities are part of the atomic type) and the pair $(A',B')$ is homogeneous (by \zcref{it:r5}). \zcref{it:tw8} follows.

\zcref{it:tw13} is also easily obtained by reducing $I'$ or $J'$ by at most $1$, hence leading to $\min(|I'|,|J'|)\geq n$.

Also note that \zcref{it:tw1,it:tw3,it:tw5} are satisfied by the $(I,J)$-sequence $(A,\mu_1,\dots,\mu_h,B)$ (by definition of an $(I,J)$-sequence).
\end{clproof}
\medskip

With these claims at hand, we are now ready to prove the following claim by induction on $h$, from which the lemma immediately follows.
\begin{claim}
\label{cl:seq2tw}
Let $g:\mathbb N\times\mathbb N\rightarrow\mathbb N$ be inductively defined by
\[
    g(n,h)= 
    \begin{cases}
        f_1(1,n+1)&\text{if $h=1$,}\\
        g(f_1(h,n+3),h-1)&\text{otherwise,}
    \end{cases}
\]
    where $f_1$ is the function introduced in \zcref{cl:bigclean}.
    Then  every tree-ordered graph~$\mathbf M$ and every integer $h$ and $n$ with
    $n\geq \max(2\bomega(\mathbf M)+2,5)$, the following holds:

    If $\mathbf M$ admits an $(I,J)$-sequence $(A,\mu_1,\dots,\mu_h,B)$ of length at most $h$ and order at least $g(h,n)$, then $\mathbf M$ contains a twister of length at most $h$ and order $n$.
\end{claim}
\begin{clproof}
According to \zcref{cl:bigclean}, we can find large subsets $I'\subseteq I$ and $J'\subset J$ such that the restriction of  $(A,\mu_1,\dots,\mu_h,B)$ to $I'\times J'$ satisfies \zcref{it:r1,it:r2,it:r3,it:r4,it:r5} and \zcref{it:tw1,it:tw3,it:tw5,it:tw8,it:tw11,it:tw12,it:tw13,it:tw14}.

\bigskip

Assume $h=1$.  As
\zcref{it:tw2,it:tw4,it:tw6,it:tw7,it:tw9,it:tw10,it:tw14}
are vacuously satisfied when $h=1$, the lemma holds in this case with 
$g(1,n)=f_1(1,n+1)$.
\medskip

Now assume $h>1$ and that we have proved that the claim holds up to $h-1$.
Assume some pair $(A',\mu_s')$, $s>1$ is not homogeneous. Then  $a$ is a vertical guard of $\mu_s'$ (by \zcref{it:r1}). Considering the shorter $(I',J')$-sequence $(a(I'),\mu_s',\dots,\mu_h',B')$ and induction, we get a twister of length at most $h$ and order $n$, as $g(n,h)\geq g(f_1(h,n+1),h-1)$.
Hence, we can assume \zcref{it:tw9}.
Similarly, using \zcref{it:r2}, we can assume that all pairs $(B',\mu_s')$, $s<h$ are homogeneous, hence
\zcref{it:tw10}.

If $\mu_1'$ is inner-vertical, we let $A':=\emptyset$; otherwise, we let $A':=a(I')$.
Similarly, if~$\mu_h'$ is inner-horizontal, we let $B':=\emptyset$; otherwise, we let $B':=b(J')$. Then  according to \zcref{fact:restr2} satisfied properties are preserved. Moreover, by constrcution, 
\zcref{it:tw6,it:tw7} hold.

Let $1<s<h$. 
According to \zcref{lem:reg}, $\mu_s$ is either homogeneous, inner-vertical, or inner-horizontal. In the latter cases, the $(I,J)$-sequence can be shortened and the induction hypothesis applies as $g(n,h)\geq g(f_1(h,n+1),h-1)$.
Hence, we can assume that $\mu_s'$ is homogeneous.
According to \zcref{lem:inde} (applied to the pair $(\mu_s',\mu_s')$), $V(\mu_s')$ is an independent set in $\mathbf M^E$ and, by antisymmetry of $\prec$, it is also an antichain. Hence, \zcref{it:tw2} holds. 

Let $1\leq s< h$ and $t=s+1$.
If $\mu_t'$ is pseudo-horizontal or pseudo-vertical, we can delete $2$ elements in $I$ and $J$, consider a shorter sequence, and apply induction as $g(n,h)\geq g(f_1(h,n+3),h-1)$. Hence, according to \zcref{lem:quasi_hom}, the pair $(\mu_s',\mu_t')$ is quasi-homogeneous, thus matching.
Hence, \zcref{it:tw4} holds.
If $\mu_1$ is inner-horizontal, then $(A,\mu_1)$ is a twister of $\mathbf M$ of order $n$ and  if $\mu_h$ is inner-vertical, then $(\mu_h,B)$ is a twister of $\mathbf M$ of order $n$. 
Otherwise, \zcref{it:tw14} holds.

With this last proof, we eventually get
that \zcref{it:tw1,it:tw2,it:tw3,it:tw4,it:tw5,it:tw6,it:tw7,it:tw8,it:tw9,it:tw10,it:tw11,it:tw12,it:tw13,it:tw14} are all satisfied.
\end{clproof}

Now, the lemma directly follows from \zcref{cl:seq2tw} applied to the $(I,J)$-sequence defined by the $(I,J)$-minimal transformer.
\end{proof}

\begin{lem}\label{lem:incomparable}
    Let $(A,\mu_1,\dots,\mu_h,B)$ be a twister in a TOWS graph $\mathbf M$.
    If $1< s< h$ and $1\le t\le h$ and $|s-t|>1$, then no element of $V(\mu_s)$ is smaller than an element of $V(\mu_t)$. 

    Consequently, if $1<s<t< h$ and $|s-t|>1$, then no element of $V(\mu_s)$ is comparable to an element of $V(\mu_t)$.
\end{lem}
\begin{proof}
    Assume for contradiction that some element of $V(\mu_s)$ is smaller than some element of $V(\mu_t)$. By homogeneity of the pair $(\mu_s,\mu_t)$ (\zcref{it:tw5}), we get that all the elements of $\mu_s$ are smaller (with respect to $\prec$) to all the elements of $\mu_t$. It follows that $\mu_s$ is a chain, contradicting \zcref{it:tw2}.
\end{proof}

We now make precise the properties of exceptional pairs that can arise in the definition of a twister.

\begin{lem}
\label{lem:except}
    Let $(A,\mu_1,\dots,\mu_h,B)$ be an $(I,J)$ twister in a tree-ordered graph $\mathbf M$, where $\min(|I|,|J|)\geq 4$.
    Assume $\mu_1$ is a chain and $\mu_2$ is an antichain.
    Then  for every $i,i'\in I$ and $j,j'\in J$, 
        \begin{align*}
        \text{either}
        \quad& \mu_1(i,j)\parallel\mu_2(i',j'),\\
        \text{or}
        \quad&(i<i')\Longrightarrow \mu_1(i,j)\prec\mu_2(i',j')
        \quad\text{ and }\quad
        (i>i')\Longrightarrow \mu_1(i,j)\parallel\mu_2(i',j'),\\
        \text{or}        \quad&\mu_1(i,j)\prec\mu_2(i',j').
    \end{align*}
\end{lem}
\begin{proof}
    Let $I_0=I\setminus\{\min I,\max I\}$ and $J_0=J\setminus\{\min J,\max J\}$. Without loss of generality we can assume that the restriction of $\prec$ to $V(\mu_1)$ is $<_{I,J}$.

    Assume for contradiction that 
    there exists $i,i'\in I_0$ and $j,j'\in J_0$ such that $\mu_2(i',j')\prec\mu_1(i,j)$.
    As $\mu_1(i,j)\prec\mu_1(\max I,\max J)$ we get, by regularity of the pair $(\mu_1,\mu_2)$ we deduce $\mu_2(i',j')\prec\mu_1(\max I,\max J)$ holds for every $i'\in I_0$ and $j'\in J_0$. Thus, all the $\mu_2(i',j')$ for $i'\in I_0$ and $j'\in J_0$ are comparable, contradicting the assumption that $\mu_2$ is an antichain.

    If there exists no $i,i'\in I_0$ and $j,j'\in J_0$ such that $\mu_1(i,j)\prec\mu_2(i',j')$, then for every $i,i'\in I$ and every $j,j'\in J$ we have $\mu_1(i,j)\parallel \mu_2(i',j')$ (by \zcref{cl:restr}).

    So, we can assume that there exists $i,i'\in I_0$ and $j,j'\in J_0$ such that $\mu_1(i,j)\prec\mu_2(i',j')$.
    Hence, for every $j''\in J$ we have
    $\mu_1(\min I,j'')\prec \mu_1(i,j)\prec\mu_2(i',j')$.
    By regularity of the pair $(\mu_1,\mu_2)$ we conclude that for every $i<i'$ in $I$ and every $j,j'\in J$ we have $\mu_1(i,j)\prec \mu_2'(i',j')$
    (for every $j,j'\in J$).

Assume there exists no $i>i'\in I_0$ and $j,j'\in J_0$ such that $\mu_1(i,j)\prec\mu_2(i',j')$. 
Then  $(i<i')\Longrightarrow \mu_1(i,j)\prec\mu_2(i',j')$ and $(i>i')\Longrightarrow \mu_1(i,j)\parallel\mu_2(i',j')$ (by \zcref{cl:restr}).

Otherwise, there exists $i>i'\in I_0$ and $j,j'\in J_0$ such that     $\mu_1(i,j)\prec\mu_2(i',j')$.  
    By regularity,
    $\mu_1(i,j)\prec\mu_2(\min I,j')$. Hence, for every $j''\in J$, we have $\mu_1(i',j'')\prec\mu_1(i,j)\prec\mu_2(\min I, j')$. Thus, by regularity, for every $i>i'\in I$ and every $j,j'\in J$, we have $\mu_1(i,j)\prec\mu_2(i',j')$. 
    Thus, $(i\neq i')\Longrightarrow \mu_1(i,j)\prec\mu_2(i',j')$.
    Hence, for every $i,i'\in I_0$ and every $j,j'\in J$, we have
    $\mu_1(i,j)\prec\mu_1(\min I,j)\prec \mu_2(i',j')$. Hence, according to \zcref{cl:restr}, $\mu_1(i,j)\prec\mu_2(i',j')$ for every $i,i'\in I$ and every $j,j'\in J$.
\end{proof}

\subsection{Clean twisters}
\label{sec:clean_twistsers}
The aim of this section is to refine the notion of twister to three clean  configurations.

\begin{ndefi}
    Let $(\mu,\mu')$ be a conducting pair of regular meshes in a tree-ordered graph $\mathbf M$, where $\mu$ is a chain and $\mu'$ is an antichain.
    Then  the pair $(\mu,\mu')$ is
\begin{itemize}
    \item \ndef{simply vertical} if 
    $\mu(i,j)\prec \mu'(i',j')$ if and only if $i\leq i'$ and $V(\mu)\cup V(\mu')$ is independent;
\item \ndef{simply horizontal} if 
    $\mu(i,j)\prec \mu'(i',j')$ if and only if $j\leq j'$ and \mbox{$V(\mu)\cup V(\mu')$} is independent.
\end{itemize}
\end{ndefi}

\pagebreak
\begin{ndefi} \label{def:preclean}
    A \ndef{pre-clean} $(I,J)$-twister in a tree-ordered graph $\mathbf M$ is a  $(I,J)$-twister $(A,\mu_1,\dots,\mu_h,B)$ in $\mathbf M$ of one of the following types:
\begin{mytype}
    \item $A$ and $B$ are non-empty antichains (in $\mathbf M^\prec$); \label{typ1}
    \item $A=B=\emptyset$, the restriction of $\prec$ to $V(\mu_1)$ (resp.\ $V(\mu_h)$) is $<_{J}$ (resp.~$<_{I}$);\label{typ2}
        \item $A=B=\emptyset$, and the restriction of $\prec$ to $V(\mu_1)$ (resp.\ $V(\mu_h)$) is $<_{I,J}$ (resp.\ a linear order of anti-lexicographic type). \label{typ3}
\end{mytype}
\end{ndefi}
Pre-clean twisters have some additional derived properties. We prove below some of them that will be useful in forthcoming proofs.
\begin{lem}\label{lem:preclean_prop}
    Let $(A,\mu_1,\dots,\mu_h,B)$ be a pre-clean $(I,J)$-twister in a tree-ordered graph $\mathbf M$. 
    \begin{itemize}
        \item If the twister has type 1, then  $A\cup B$, $A\cup V(\mu_s)$ (for $s>1$), and  $B\cup V(\mu_s)$ (for $s<h$) are antichains;
        \item If the twister has type 2, then vertices of $\mu_1$ and $\mu_2$ (resp.\ of $\mu_{h-1}$ and~$\mu_h$) are incomparable.
    \end{itemize}
\end{lem}
\begin{proof}
    Assume the twister has type 1. 
    According to \zcref{it:tw11}, the pair $(A,B)$ is homogeneous. Hence, if some element of $A$ is comparable with some element of $B$, then  one of $A$ and $B$ has all its elements smaller than the elements of the other, hence is a chain, contradicting the assumptions. Thus, $A\cup B$ is an antichain.
    According to \zcref{it:tw2}, $\mu_s$ is an antichain for all $1<s<h$. 
    As~$A$ is not empty, according to \zcref{it:tw6}, $V(\mu_1)$ is not inner-vertical. Moreover, $\mu_1$ is not inner-horizontal (by \zcref{it:tw14} if $h>1$ and by \zcref{it:tw7} if $h=1$ as $B\neq\emptyset$).
    Hence, according to \zcref{lem:reg}, $\mu_1$ is an antichain. Similarly, $\mu_h$ is an antichain. Let $1<s\leq h$. According to \zcref{it:tw9}, the pair $(A,\mu_s)$ is homogeneous. If the elements of $A$ would be comparable with the elements of~$\mu_s$, then at least one of $A$ and $\mu_s$ would be a chain (as $\prec$ is a tree-order). Hence, $A\cup V(\mu_s)$ is an antichain. Similarly, for every $1\leq s<h$, $B\cup V(\mu_s)$ is an antichain.

    Assume the twister has type 2. According to \zcref{it:tw4}, the pair $(\mu_1,\mu_2)$ is matching. 
    Assume for contradiction that there exists $i,i'\in I$ and $j,j'\in J$ with $\mu_1(i,j)\prec\mu_2(i',j')$. 
    By regularity, we can assume that $j\neq \min J$.
    Then  we have $\mu_1(i,\min J)\prec\mu_1(i,j)\prec\mu_2(i',j')$.
    As at least one of $(i,\min J)$ and $(i,j)$ is different from $(i',j')$, we deduce from the  quasi-homogeneity of $(\mu_1,\mu_2)$ that $V(\mu_1)$ is a chain, contradicting the assumption. 
    Similarly, if there exists $i,i'\in I$ and $j,j'\in J$ with $\mu_2(i,j)\prec\mu_1(i',j')$, we can assume that $j'$ is not $\max J$ and we deduce as above that $V(\mu_2)$ is a chain, contradicting our assumptions. Thus, the vertices of $\mu_1$ and $\mu_2$ are incomparable. 
    Alike, we prove that the vertices of $\mu_{h-1}$ and $\mu_h$ are incomparable. 
\end{proof}

\begin{ndefi}\label{def:clean}
    A \emph{clean}\index{clean twister} $(I,J)$-twister in a tree-ordered graph $\mathbf M$ is  a pre-clean $(I,J)$-twister such that no vertex of $\mu_h$ is smaller than a vertex of $\mu_1$, the root $r$ is homogeneous to $A,B$ and each mesh of the twister, with the following additional  properties if the twister has type 3:
    \begin{itemize}
        \item  $(\mu_1,\mu_2)$ is either  matching or simply vertical;
        \item  $(\mu_h,\mu_{h-1})$ is either matching or simply horizontal.
    \end{itemize}
\end{ndefi}

The rest of this section will be devoted to the proof that every twister can be turned into a clean one. 

The next two lemmas will be useful.
\begin{lem}
\label{lem:useful}
	Let $I,J,I_1,J_1$ be totally ordered sets with $|I|\geq |I_1|+|J_1|$ and $|J|\geq |I_1|\cdot|J_1|$. 

    Identify the linear sum $\hat I=I_1\oplus J_1$ with a subset of $I$ and the lexicographic product $\hat J=I_1\bullet J_1$ with a subset of $J$.
	
	Let $(\mu,\mu')$ be a regular pair of $(I,J)$-meshes of a tree-ordered graph $\mathbf M$, and let 
    \begin{align*}
    \hat\mu(i,j)&:=\mu(i,(i,j)) \\\text{and }\hat\mu'(i,j)&:=\mu'(i,(i,j))
    \end{align*}
    for $i\in I_1$ and $j\in J_1$.

    Then the pair $(\hat\mu,\hat\mu')$ is regular.
    
    Moreover, if the pair $(\mu,\mu')$ is matching, then the pair $(\hat\mu,\hat\mu')$ is matching.
\end{lem}
\begin{proof}
	Let $i,i'\in I_1$ and $j,j'\in J_1$.
	Then
	\[	\atp(\hat\mu(i,j),\hat\mu'(i',j'))=\atp(\mu(i,(i,j)),\mu'(i',(i',j')))
	\]
    
	depends only on $\otp(i,i')$ and $\otp(j,j')$, hence $(\hat\mu,\hat\mu')$ is regular.
	Moreover, if $(\mu,\mu')$ is matching, then 
    $\atp(\mu(i,(i,j)),\mu'(i',(i',j')))$
     takes different values if $(i,(i,j))=(i',(i',j'))$ and 
	if $(i,(i,j))\neq (i',(i',j'))$, i.e. if $(i,j)=(i',j')$ and if $(i,j)\neq(i',j')$.
    Thus, $(\hat\mu,\hat\mu')$ is matching.
\end{proof}
\begin{lem}
\label{lem:useful2}
	Let $I,J,I_1,J_1$ be totally ordered sets with $|I|,|J|\geq  |I_1|\cdot|J_1|$. 

    Identify the lexicographic products $I_1\bullet J_1$ with a subset $\hat I$ of $I$ and $I_1\bullet J_1$ or $J_1\bullet I_1$ with a subset of $\hat J$ of $J$.
	
	Let $(\mu,\mu')$ be a regular pair of $(I,J)$-meshes of a tree-ordered graph $\mathbf M$, and let 
    \begin{align*}
    \text{either }
    &\begin{cases}
        &\hat\mu(i,j):=\mu((i,j),(i,j))\\
        &\qquad\qquad\text{and}\\
        &\hat\mu'(i,j):=\mu'((i,j),(i,j)),
    \end{cases}\\
    \\
  \text{or }
    &\begin{cases}
     &\hat\mu(i,j):=\mu((i,j),(j,i))\\
        &\qquad\qquad\text{and}\\
        &\hat\mu'(i,j):=\mu'((i,j),(j,i)),
    \end{cases}
    \end{align*}
    for every $i\in I_1$ and $j\in J_1$.

    Then the pair $(\hat\mu,\hat\mu')$ is regular.
    
    Moreover, if the pair $(\mu,\mu')$ is matching, then the pair $(\hat\mu,\hat\mu')$ is matching.
\end{lem}
\begin{proof}
	Let $i,i'\in I_1$ and $j,j'\in J_1$.
	Then $\atp(\hat\mu(i,j),\hat\mu'(i',j'))$
	depends only on $\otp((i,j),(i',j'))$ and/or on $\otp((j,i),(j',i'))$, that is, on 
    $\otp(i,i')$ and $\otp(j,j')$, hence $(\hat\mu,\hat\mu')$ is regular.
	Moreover, if $(\mu,\mu')$ is matching, then it takes different values if $(i,j)=(i',j')$ and 
	if $(i,j)\neq (i',j')$.
    Thus, $(\hat\mu,\hat\mu')$ is matching.
\end{proof}

\begin{lem}
	\label{lem:killChain}
	Let $n$ be an integer, let $|I|,|J|\ge n^2$, 
	let $(A,\mu_1,\dots,\mu_h,B)$ be an $(I,J)$-twister of a tree-ordered graph $\mathbf M$, where  $A$ and $B$ are both non-empty chains and $n\geq 2\bomega(\mathbf M)+2$.

    Then  $\mathbf M$ admits a pre-clean twister of Type $3$,
    with length $h+2$ and order $n$.
\end{lem}
\begin{proof}
    By restricting $I$ and $J$ if necessary, we can assume $|I|=|J|=n^2$.
    We consider two ordered sets $I_1,J_1$ with $|I_1|=|J_1|=n$  and identify $I$ with $I_1\bullet J_1$ and~$J$ with $J_1\bullet I_1$. Define
\[
    \mu_s'(i,j)=\begin{cases}
        a(i,j)&\text{if $s=1$,}\\
        \mu_{s-1}((i,j),(j,i))&\text{if $2\leq s\leq h+1$,}\\
        b(j,i)&\text{if $s=h+2$.}
    \end{cases}
\]

We now prove that $(\emptyset,\mu_1',\dots,\mu_{h+2}',\emptyset)$ is a Type $3$ pre-clean twister.

By construction, the restriction of $\prec$
to $V(\mu_1')$ (\/resp.\ $V(\mu_{h+2}')$\/) is $<_{I_1,J_1}$ (\/resp.\ $<_{J_1,I_1}$\/). 
Hence,
\zcref{it:tw1,it:tw3,it:tw6,it:tw7} hold.
Moreover, \zcref{it:tw8,it:tw13} obviously hold and \zcref{it:tw9,it:tw10,it:tw11,it:tw12} are vacuously true.

\zcref{it:tw2} for $\mu_s$ with
$1<s<h$ implies the same property for $\mu_{s+1}'$. As $A\neq\emptyset$, $\mu_1$ is not inner-vertical (by \zcref{it:tw6}).
As $B\neq\emptyset$, $\mu_1$ is also not inner-horizontal (by \zcref{it:tw14,it:tw7}).
According to \zcref{lem:inde,lem:reg}, $\mu_1$ is an independent antichain. Similarly, $\mu_h$ is an independent antichain. Hence, \zcref{it:tw2} holds.
By construction, $\mu_1'$ is inner-vertical (thus not inner-horizontal) and $\mu_{h+2}'$ is inner-horizontal (thus not inner-vertical). Hence, \zcref{it:tw14} holds.

\zcref{it:tw5} follows from \zcref{it:tw5,it:tw9,it:tw10} of the twister $(A,\mu_1,\dots,\mu_h,B)$.

As $\mu_1$ and $\mu_h$ are antichains (see above), 
 \zcref{it:tw4} follows from \zcref{lem:useful2} and \zcref{it:tw4} of $(A,\mu_1,\dots,\mu_h,B)$ for $1<s<h+1$ and $t=s+1$.
As $\mu_1'$ is a chain and $\mu_2'$ is an antichain, 
\zcref{it:tw4} holds for $s=1$ and $t=2$; similarly, as $\mu_{h+2}'$ is a chain  and $\mu_{h+1}'$ is an antichain,
\zcref{it:tw4} holds for $s=h+1$ and $t=h+2$. Hence, the sequence $(\emptyset,\mu_1',\dots,\mu_{h+2}',\emptyset)$ satisfies \zcref{it:tw4}.

Altogether, we deduce that $(\emptyset,\mu_1',\dots,\mu_{h+2}',\emptyset)$ is a twister. That it is pre-clean and has Type $3$ is clear from the construction.
\end{proof}

\begin{lem}
	\label{lem:killB}
	Let $n$ be an integer, let $|I|\ge 2n^2$, $|J|\geq n^4$, and
	let $(A,\mu_1,\dots,\mu_h,B)$ be an $(I,J)$-twister of a tree-ordered graph $\mathbf M$, where $B$ is a non-empty antichain and $n\geq 2\bomega(\mathbf M)+2$.

    Then  $\mathbf M$ admits a pre-clean twister with length at most $2h+3$ and order $n$.
\end{lem}
\begin{proof}
If $A$ is a non-empty antichain, then $(A,\mu_1,\dots,\mu_h,B)$ is a pre-clean twister of Type $1$. Indeed, notice
that since ancestors in a tree order induce chains, if $(X,Y)$ is an homogenous pair where $X$ and $Y$ are 
antichains, then $X\cup Y$ is an antichain.
Thus, we can assume that $A$ is either empty or a non-empty chain (according to \zcref{it:tw12}).

By restricting $I$ and $J$ if necessary, we can assume $|I|=2n^2$ and $|J|=n^4$.
    We consider two disjoint ordered sets $I_1,J_1$ with $|I_1|=|J_1|=n^2$ and identify $I$ with $I_1\oplus J_1$ and $J$ with $I_1\bullet J_1$. 

	As $B$ is not empty, $\mu_h$ is an antichain and thus $B$ is the range of a horizontal guard $b$ of $\mu_h$. 
    If $A$ is non-empty, we 
    let $a:I\rightarrow M$ be a vertical guard of $\mu_1$ with range $A$.
	We define, for $i\in I_1$ and $j\in J_1$, 
	\begin{align*}
        A'&=\begin{cases}
            a(I_1)&\text{if $A\neq\emptyset$,}\\
            \emptyset&\text{otherwise;}
        \end{cases}\\
        B'&=\begin{cases}
            a(J_1)&\text{if $A\neq\emptyset$,}\\
            \emptyset&\text{otherwise;}
        \end{cases}\\
		\mu_s'(i,j)&=\begin{cases}
		\mu_s(i,(i,j))&\text{ if }1\leq s\leq h,\\
		b((i,j))&\text{if }s=h+1,\\
		\mu_{2h+2-s}(j,(i,j))&\text{ if }h+2\leq s\leq 2h+1.\\	
	\end{cases}	
	\end{align*}

    \zcref{it:tw2} follows from \zcref{it:tw2,it:tw12}  of $(A,\mu_1,\dots,\mu_h,B)$.

    If $A=\emptyset$, $\mu_1'$ is inner-vertical and $\mu_{2h+1}'$ is inner-horizontal, hence \zcref{it:tw1,it:tw3,it:tw6,it:tw7,it:tw14} hold. Moreover, \zcref{it:tw9,it:tw10,it:tw11,it:tw12} vacuously hold.
    Otherwise, $\mu_1'$ is vertical and guarded by $A'$ and
	$\mu_{2h+1}'$ is horizontal and guarded by $B'$. Hence, \zcref{it:tw1,it:tw3,it:tw6,it:tw7,it:tw14} hold. \zcref{it:tw9,it:tw10} follows from the \zcref{it:tw9,it:tw11} of $(A,\mu_1,\dots,\mu_h,B)$.
    \zcref{it:tw11,it:tw12} follows \zcref{it:tw12} of $(A,\mu_1,\dots,\mu_h,B)$.
    Hence, we have proved \zcref{it:tw1,it:tw3,it:tw6,it:tw7,it:tw9,it:tw10,it:tw11,it:tw12}.

    \zcref{it:tw5} follows from \zcref{it:tw5,it:tw10} of $(A,\mu_1,\dots,\mu_h,B)$.
	As $V(\mu_s')$ and $V(\mu_{2h+2-s}')$ are disjoint subsets of $V(\mu_s)$ (for $1\le s\le h$), and $V(\mu'_{h+1})\subseteq B$,
    \zcref{it:tw8,it:tw13} hold.

    For $1\leq s<h$ or $h+1<s\leq 2h$ and for $t=s+1$, \zcref{it:tw4} follows from \zcref{it:tw4} of $(A,\mu_1,\dots,\mu_h,B)$ and \zcref{lem:useful}.
    For $s=h$ and $t=h+1$, the pair $(\mu_s',\mu_t')$ is regular as $B$ is the range of a horizontal guard of $\mu_h$ (by \zcref{it:tw7} of $(A,\mu_1,\dots,\mu_h,B)$). As $\mu_h$  is not inner-horizontal, it is an antichain and hence, so is $\mu_s'$. As $B$ (hence $\mu_t'$) is also an antichain, according  to \zcref{lem:anti-qh}, $(\mu_s',\mu_t')$ is quasi-homogeneous. 
    As $B$ is a horizontal guard of $\mu_h$, $(\mu_s',\mu_t')$ is not homogeneous, hence matching. Similarly, $(\mu_{h+1}',\mu_{h+2}')$ is matching. Thus, \zcref{it:tw4} holds.
    
	Altogether, we proved that 
    $(A',\mu_1',\dots,\mu_{2h+1}',B')$ is a twister of order $n^2$.
    
    If $A$ is empty, this twister has Type $2$ or Type $3$.
    Otherwise, $A'$ and $B'$ are both non-empty chains. By \zcref{lem:killChain}, we deduce that $\mathbf M$ admits a pre-clean twister of Type $3$ with length $2h+3$ and order $n$.
\end{proof}

\begin{lem}
    \label{lem:killH}
    Let $n$ be an integer, let $|I|\ge 2n$, $|J|\ge n^2$, 
	let $(\emptyset,\mu_1,\dots,\mu_h,\emptyset)$ be an $(I,J)$-twister of a tree-ordered graph $\mathbf M$, where $\mu_1$ is a chain, $\mu_h$ is not a chain and $n\geq 2\bomega(\mathbf M)+2$.

    Then  $\mathbf M$ admits a pre-clean twister of Type $3$,
    with length $2h$ and order $n$.
\end{lem}
\begin{proof}
    By restricting $I$ and $J$ if necessary, we can assume $|I|=2n$ and $|J|=n^2$. Without loss of generality, we can assume that the restriction of $\prec$ to $V(\mu_1)$ is $<_{I,J}$.
    We consider two disjoint ordered sets $I_1,J_1$ with $|I_1|=|J_1|=n$ and identify $I$ with $I_1\oplus J_1$ and $J$ with $I_1\bullet J_1$. 
Define
\[
    \mu_s'(i,j)=\begin{cases}
        \mu_{s}(i,(i,j))&\text{if $1\leq s\leq h$,}\\
        \mu_{2h+1-s}(j,(i,j))&\text{if $h+1\leq s\leq 2h$.}
    \end{cases}
\]
We now prove that $(\emptyset,\mu_1',\dots,\mu_{2h}',\emptyset)$
is a pre-clean twister of Type $3$.

For $i\in I_1$ and $j\in J_1$,
we have $(i,(i,j))<_{I,J} (i',(i',j'))$
if $i<_{I_1} i'$ or $i=i'$ and $j<_{J_1} j'$, i.e. if $(i,j)<_{I_1,J_1}(i',j')$.
Hence, the restriction of $\prec$ to $V(\mu_1')$ is $<_{I_1,J_1}$.
Similarly, the restriction of $\prec$ to $V(\mu_{2h}')$ is $<_{J_1,I_1}$. Thus, \zcref{it:tw1,it:tw3,it:tw6,it:tw7} hold. 
Moreover, \zcref{it:tw9,it:tw10,it:tw11,it:tw12} vacuously hold.
As $V(\mu_s')$ and $V(\mu_{2h+1-s})$ are disjoint subsets of $V(\mu_s)$, \zcref{it:tw8} holds.
Assume $\mu_h'(i,j)$ is comparable with $\mu_h'(i',j')$. Then  $\mu_h(i,(i,j))$ is comparable with $\mu_h(i',(i',j'))$, thus $(i,j)=(i',j')$.
Hence, $\mu_h'$ is an antichain. 
Similarly, $\mu_{h+1}'$ is an antichain.
Both are independent sets according to \zcref{lem:inde}.
Thus, \zcref{it:tw2} and \zcref{it:tw13} follows from the same properties of $(\emptyset,\mu_1,\dots,\mu_h,\emptyset)$.
As $\mu_1'$ is inner-vertical and $\mu_{2h}'$ is inner-horizontal, \zcref{it:tw14} holds.

We now prove \zcref{it:tw5}.
Let $1\leq s,t\leq h$.
If $|s-t|>1$, then the pairs 
$(\mu_s',\mu_t'),(\mu_s',\mu_{2h+1-t}')$, $(\mu_{2h+1-s}',\mu_t')$, and $(\mu_{2h+1-s}',\mu_{2h+1-t}')$ are homogeneous as
$(\emptyset,\mu_1,\dots,\mu_h,\emptyset)$ has \zcref{it:tw5}.
If $1\leq s< h$, then
$(\mu_s',\mu_{2h+1-s}')$ is homogeneous as
$(\emptyset,\mu_1,\dots,\mu_h,\emptyset)$ has \zcref{it:tw2} and $\mu_1$ is a chain.
If $1<s\leq h$, then
 $(\mu_s',\mu_{2h-s}')$ and
  $(\mu_{s+1}',\mu_{2h+1-s}')$
 are homogeneous as $(\mu_s,\mu_{s+1})$
 is matching (by \zcref{it:tw4} of 
 $(\emptyset,\mu_1,\dots,\mu_h,\emptyset)$).
 That the pairs $(\mu_1',\mu_{2h-1}')$ and $(\mu_2',\mu_{2h}')$ are homogeneous follows from \zcref{lem:except}. Thus, \zcref{it:tw5} holds.

 We now prove \zcref{it:tw4}. For $1<s<h$ and $h<s<2h-1$, the pair $(\mu_s',\mu_{s+1}')$ is matching according to \zcref{it:tw4} of $(\emptyset,\mu_1,\dots,\mu_h,\emptyset)$. By construction, $\mu_h'(i,j)$ is comparable to $\mu_{h+1}'(i',j')$
if and only if $(i,j)=(i',j')$. Hence, $(\mu_h',\mu_{h+1}')$ is matching.
As $(\mu_1,\mu_2)$ is conducting, according to \zcref{lem:except}, we have that for every $i\neq i'\in I_1$ and $j\neq j'\in J_1$ we have
\begin{align*}
\atp(\mu_1(i,(i,j)),\mu_2(i',(i',j')))&\neq \atp(\mu_1(i',(i',j')),\mu_2(i,(i,j)))\\    
\text{and }\atp(\mu_1(j,(i,j)),\mu_2(j',(i',j')))&\neq \atp(\mu_1(j',(i',j')),\mu_2(j,(i,j))). 
\end{align*}
Hence, $(\mu_1',\mu_2')$ and $(\mu_{2h}',\mu_{2h-1}')$ are not homogeneous. According to \zcref{lem:useful},
these pairs are also regular, hence conducting. Thus, \zcref{it:tw4} holds.
\end{proof}

\begin{lem}
	\label{lem:killCC}
	Let $n$ be an integer, let $|I|,|J|\ge n^2$, 
	let $(A,\mu_1,\dots,\mu_h,\emptyset)$ be an $(I,J)$-twister of a tree-ordered graph $\mathbf M$, where both $A$ and $\mu_h$ are non-empty chains, and $n\geq 2\bomega(\mathbf M)+2$.

    Then  $\mathbf M$ admits a pre-clean twister of Type $3$ with length $h+1$ and order $n$.
\end{lem}
\begin{proof}
    By restricting $I$ and $J$ if necessary, we can assume $|I|=|J|=n^2$. 
    Without loss of generality, we can assume that 
    the restriction of $\prec$ to $V(\mu_h)$ is $<_{J,I}$.
    As $A$ is not empty, $A$ is the range of a vertical guard $a$ of $\mu_1$ (by \zcref{it:tw6}).
    We consider two disjoint ordered sets $I_1,J_1$ with $|I_1|=|J_1|=n$ and identify $I$ with $I_1\bullet J_1$ and $J$ with $J_1\bullet I_1$. 
    Define
    \[
\mu_s'(i,j)=\begin{cases}
    a((i,j))&\text{if }s=1;\\
    \mu_{s-1}((i,j),(j,i))&\text{if }1<s\leq h+1.
\end{cases}
\]

Now, we aim to prove that the sequence $(\emptyset,\mu_1',\dots,\mu_{h+1}',\emptyset)$ is a twister.

By construction, the restriction of $\prec$ to $V(\mu_1')$ is $<_{I_1,J_1}$. Hence, \zcref{it:tw1,it:tw6} hold.
Moreover, for $i,i'\in I_1$ and $j,j'\in J_1$ we have
\begin{align*}
    \mu_{h+1}'(i,j)\prec\mu_{h+1}'(i',j')&\iff \mu_h((i,j),(j,i))\prec\mu_h((i',j'),(j',i'))\\
    &\iff (j,i)<(j',i')\quad\text{(w.r.t. the order of $J=J_1\bullet I_1$)}\\
    &\iff (i,j)<_{J_1,I_1}(i',j')
\end{align*}
Hence, the restriction of $\prec$ to $V(\mu_{h+1}')$ is $<_{J_1,I_1}$. Thus, \zcref{it:tw3,it:tw7} hold.
\zcref{it:tw8,it:tw13} obviously follow from the same properties of $(A,\mu_1,\dots,\mu_h,\emptyset)$ and
\zcref{it:tw9,it:tw10,it:tw11,it:tw12} vacuously hold.
\zcref{it:tw5} follows from \zcref{it:tw5,it:tw9} of $(A,\mu_1,\dots,\mu_h,\emptyset)$. 

That $\mu_s'$ is independent and an antichain for $2<s<h+1$ follows from \zcref{it:tw2} of $(A,\mu_1,\dots,\mu_h,\emptyset)$. That $\mu_2'$ is independent follows from \zcref{lem:inde}. 

Assume $h>1$. That $\mu_1$ is an antichain follows from \zcref{it:tw6,it:tw14} of $(A,\mu_1,\dots,\mu_h,\emptyset)$, by the fact that $A$ is not empty.
Hence, \zcref{it:tw2} holds. 
Otherwise, \zcref{it:tw2} holds vacuously.

As $\mu_1'$ is internally vertical and $\mu_{h+1}'$ is internally horizontal, we deduce  \zcref{it:tw14}.

That $(\mu_s',\mu_{s+1}')$ is matching for $2\leq s\leq h$ follows from \zcref{lem:useful2}, \zcref{it:tw4} of  $(A,\mu_1,\dots,\mu_h,\emptyset)$, and the fact that $\mu_1$ is not a chain (as $h>1$ in this case).

For $i_1,i_1',i_2'\in I_1$ and $j_1,j_1',j_2'\in J_1$, 
$\atp(a((i_1,j_1)),\mu_1((i_1',j_1'),(j_2',i_2')))$ only depends on $\otp((i_1,j_1),(i_1',j_1'))$ and is not constant (by definition of a vertical guard).
It follows that $\atp(\mu_1'(i,j),\mu_2'(i',j'))=\atp(a((i,j)),\mu_1((i,j),(j,i))$
depends only on $\otp((i,j),(i',j'))$ (hence only on $\otp(i,i')$ and $\otp(j,j')$) and is not constant.
It follows that the pair $(\mu_1',\mu_2')$ is conducting.
If $h>1$, then $\mu_1'$ is a chain  and~$\mu_2'$ is an antichain, hence \zcref{it:tw4} holds. 
Otherwise, $\mu_1'$ is inner-vertical and~$\mu_2'$ is inner-horizontal. According to \zcref{lem:matching}, $(\mu_1',\mu_2')$ is matching, hence \zcref{it:tw4} holds.

By all the above, we get that $(\emptyset,\mu_1',\dots,\mu_{h+1}',\emptyset)$ is a pre-clean twister of Type~$3$.
\end{proof}

\begin{lem}
	\label{lem:killCS}
    Let $n$ be an integer, let $|I|\geq 2n^2$,$|J|\ge n^2$, 
	let $(A,\mu_1,\dots,\mu_h,\emptyset)$ be an $(I,J)$-twister of a tree-ordered graph $\mathbf M$, where $A$ is a non-empty chain, $\mu_h$ is not a chain, and $n\geq 2\bomega(\mathbf M)+2$.

    Then $\mathbf M$ admits a pre-clean twister of Type $3$ with length $2h+2$ and order $n$.
\end{lem}
\begin{proof}
     By restricting $I$ and $J$ if necessary, we can assume $|I|=2n^2$ and $|J|=n^2$. 
    Without loss of generality, we can assume that 
    the restriction of $\prec$ to $V(\mu_h)$ is $<_{I}$.
    As $A$ is not empty, $A$ is the range of a vertical guard $a$ of $\mu_1$ (by \zcref{it:tw6}).
    We consider two disjoint ordered sets $I_1,J_1$ with $|I_1|=|J_1|=n$ and identify $I$ with $I_1\bullet J_1\oplus J_1\bullet I_1$ and $J$ with $J_1\bullet I_1$. For $i\in I_1$ and $j\in J_1$,
    define  
    \[
\mu_s'(i,j)=\begin{cases}
    a((i,j))&\text{if }s=1;\\
    \mu_{s-1}((i,j),(j,i))&\text{if }1<s\leq h+1;\\
    \mu_{2h+2-s}((j,i),(j,i))&\text{if }h+1<s\leq 2h+1;\\
    a((j,i))&\text{if }s=2h+2.
\end{cases}
\]

We now prove that the sequence $(\emptyset,\mu_1',\dots,\mu_{2h+2}',\emptyset)$ is an $(I_1,J_1)$-twister.

As the restriction of $\prec$ to $V(\mu_1')$ (resp.\ $V(\mu_{2h+2}'$) is $<_{I_1,J_1}$ (resp.\ $<_{J_1,I_1}$), \zcref{it:tw1,it:tw6,it:tw3,it:tw7} hold.

\zcref{it:tw9,it:tw10,it:tw11,it:tw12} vacuously hold. Moreover, \zcref{it:tw8,it:tw13} obviously follow from the same properties of $(A,\mu_1,\dots,\mu_h,\emptyset)$.

That all the $\mu_s'$ are independent follows from \zcref{lem:inde}.
By \zcref{it:tw6} of $(A,\mu_1,\dots,\mu_h,\emptyset)$, $\mu_1$ is an antichain. Hence, by \zcref{it:tw2} of $(A,\mu_1,\dots,\mu_h,\emptyset)$, $\mu_s'$ is an antichain for every $1<s<h+1$ and every $h+2<s<2h+2$. 
For every $i,i'\in I_1$ and every $j,j'\in J_1$, we have (according to our assumption on $\mu_h$) that
\begin{align*}
    \mu_{h+1}'(i,j)\parallel\mu_{h+1}'(i',j')&\iff \mu_h((i,j),(j,i))\parallel\mu_h((i',j'),(j',i'))\\
    &\iff (j,i)\neq (j',i').
\end{align*}
Hence, $\mu_{h+1}'$ is an antichain.
Similarly, $\mu_{h+2}'$ is an antichain, thus \zcref{it:tw2} holds. As $\mu_1'$ is internally vertical and $\mu_{2h+2}'$ is internally horizontal, we deduce \zcref{it:tw14}.

We now prove \zcref{it:tw5}.
As all the pairs $(\mu_s,\mu_t)$ with $|s-t|>1$ are homogeneous (\zcref{it:tw5}), it follows that for every $1<s,t\leq h+1$ with $|s-t|>1$, the pairs $(\mu_s',\mu_t')$, $(\mu_s',\mu_{2h+2-t}')$, $(\mu_{2h_2-s}',\mu_t')$, and 
$(\mu_{2h+2-s}',\mu_{2h+2-t}')$ are homogeneous. That the pairs $(\mu_s',\mu_{2h+2-s}')$ (for $1<s<h+1$) are homogeneous follows from \zcref{it:tw2,it:tw6} of $(A,\mu_1,\dots,\mu_h,\emptyset)$. That $(\mu_1',\mu_{2h+2}')$ is homogeneous is immediate from the definition.
For $1<s,t<h+1$ and $|s-t|=1$, that the pair $(\mu_s',\mu_{2h+2-t}')$ is homogeneous follows from the fact that pair $(\mu_s,\mu_t)$ is matching (which follows from \zcref{it:tw4} of $(A,\mu_1,\dots,\mu_h,\emptyset)$ and the fact that none of $\mu_1$ and $\mu_h$ is a chain).
For $2<s<2h$, that the pairs $(\mu_1',\mu_s')$ and $(\mu_{2h+2}',\mu_s')$ are homogeneous follow from \zcref{it:tw9} of $(A,\mu_1,\dots,\mu_h,\emptyset)$.
For all $i,i'\in I_1$ and $j,j'\in J_1$ we have $(i,j)<_I (j',i')$. It follows that the pair $(\mu_1',\mu_{2h+1}')$ is homogeneous. Similarly, the pair $(\mu_2',\mu_{2h+2}')$ is homogeneous. Altogether, the sequence $(\emptyset,\mu_1',\dots,\mu_{2h+2}',\emptyset)$ satisfies \zcref{it:tw5}.

We now prove \zcref{it:tw4}. According to \zcref{lem:useful2} and \zcref{it:tw4} of 
$(A,\mu_1,\dots,\mu_h,\emptyset)$, the pair $(\mu_s',\mu_{s+1}')$ are matching if $1<s<h+1$ or $h+1<s<2h+1$.
For every $i,i'\in I_1$ and $j,j'\in J_1$ we have
\begin{align*}
    \mu_{h+1}'(i,j)\parallel\mu_{h+2}'(i',j')&\iff \mu_h((i,j),(j,i))\parallel\mu_h((j',i'),(j',i'))\\
    &\iff (j,i)\neq (j',i').
\end{align*}
Hence, the pair $(\mu_{h+1}',\mu_{h+2}')$ is matching.
Let $i\neq i'\in I_1$ and $j\neq j'\in J_1$.
As $a$ is a guard of $\mu_1$, we have
\[\atp(a((i,j),\mu_1((i',j'),(j',i'))\neq
\atp(a((i',j'),\mu_1((i,j),(j,i)).\]
Hence,
$\atp(\mu_1'(i,j),\mu_2'(i',j'))\neq \atp(\mu_1'(i',j')\mu_2'(i,j))$. As the pair $(\mu_1',\mu_2')$ is regular, it is conducting. A similar argument shows that the pair $(\mu_{2h+2}',\mu_{2h+1}')$ is conducting. Altogether, we proved \zcref{it:tw4}.

It follows from the above that the sequence $(\emptyset,\mu_1',\dots,\mu_{2h+2}',\emptyset)$ is a pre-clean twister of Type $3$.
\end{proof}

\begin{lem}
\label{lem:r}
    For all integers $n,h$ there exists an integer $m$ such that if a tree-ordered graph $\mathbf M$ contains a pre-clean twister $(A,\mu_1,\dots,\mu_h,B)$ of order $m$ and length $h$, then it contains a pre-clean twister $(A',\mu_1',\dots,\mu_h',B')$ of order $n$ and length $h$ with the property that no vertex of $\mu_h$ is smaller than a vertex of $\mu_1$ and~$\{r\}$ is homogeneous to $A',B'$, and each of the $\mu_i'$.
\end{lem}
\begin{proof}
First note that if some vertex of $\mu_h$ is smaller than some vertex of $\mu_1$, then $h>1$ and we can consider the transpose $(B^\top,\mu_h^\top,\dots,\mu_1^\top,A^\top)$ of our twister.
We assume that the meshes of the twister are index by $I$ and $J$. By the pigeonhole principle, we can find subset $I_0\subseteq I$ and $J_0\subseteq J$ of size at least $m/2$ such that 
$E(a(i),r)$ is constant for $i\in I_0$ and 
$E(b(j),r)$ is constant for $j\in J_0$.
Consider the coloring $\gamma:I_0\times J_0\rightarrow \{\bot,\top\}^h$ given by
$\gamma(i,j):=(E(\mu_1(i,j),r),\dots,E(\mu_h(i,j),r))$. By Ramsey's theorem, if $m$ is sufficiently large, there exists $I'\subseteq I_0$ and $J'\subset J_0$ such that
$\gamma$ is constant on $I'\times J'$. The restriction of $I$ to $I'$ and $J$ to $J'$ yields the claimed twister.
\end{proof}

\begin{lem}
\label{lem:last_clean}
    Let $(\emptyset,\mu_1,\dots,\mu_h,\emptyset)$ be a pre-clean $(I,J)$-twister of type $3$, length~$h$, with order $2n$ ($n\ge 4$) in a tree-ordered graph $\mathbf M$ such that $r$ is homogeneous to all the meshes of the twister.
    Then  $\mathbf M$ contains a clean $(I',J')$-twister $(\emptyset,\mu_1',\dots,\mu_h',\emptyset)$  of type $3$, length $h$, and order $n$.
\end{lem}
\begin{proof}
    Reducing $I$ and $J$, we can assume $I=J=[2n]$. Define
    \begin{align*}
    \mu_1'(i,j)&=\begin{cases}
        \mu_1(2i,2j-1)&\text{if $(\mu_1,\mu_2)$ is not matching,}\\
        \mu_1(2i-1,2j-1)&\text{otherwise;}
    \end{cases}\\
    \mu_s'(i,j)&=\mu_s(2i-1,2j-1)\qquad\text{if $1<s<h$;}\\
        \mu_h'(i,j)&=\begin{cases}
        \mu_h(2i-1,2j)&\text{if $(\mu_h,\mu_{h-1})$ is not matching,}\\
        \mu_h(2i-1,2j-1)&\text{otherwise.}
    \end{cases}
    \end{align*}
    According to \zcref{lem:except},
    this process changes, \emph{ceteris paribus}, a non-matching pair $(\mu_1,\mu_2)$ to a simply vertical pair and a non-matching pair $(\mu_h,\mu_{h-1})$ to a simply horizontal pair.
\end{proof} 

\begin{lem}
\label{lem:cleantw}
    For all integers $n,h$ there exists an integer $m$ such that if a tree-ordered graph $\mathbf M$ contains a twister of order at least $m$
    and $n\geq 2\bomega(\mathbf M)+2$, then $\mathbf M$ contains a clean twister of order $n$ and length at most $2h+3$.
\end{lem}
\begin{proof}
    Let $(A,\mu_1,\ldots,\mu_h,B)$ be a twister of $\mathbf M$ of order $m$ and length at most $h$.

We consider the following cases:

\begin{mycases}
    \item\label{c:c} $A$ and $B$ are both non-empty.
    \begin{mycases}
        \item\label{c:cc}  $A$ and $B$ are both chains. Then  the result follows from \zcref{lem:killChain}.
        \item\label{c:nec} at least one of $A$ and $B$ is an  antichain (without loss of generality, $B$ is an antichain). Then  the result follows from  \zcref{lem:killB}.
    \end{mycases}
    \pagebreak
    \item\label{c:e} At least one of $A$ and $B$ is empty (without loss of generality, $B=\emptyset$).
    \begin{mycases}
        \item\label{c:ee} $A$ is empty.
        \begin{mycases}
            \item\label{c:eecc} $\mu_1$ and $\mu_h$ are both chains. Then  $(A,\mu_1,\dots,\mu_h,B)$ is a pre-clean twister of Type $3$. \item\label{c:eencnc} none of $\mu_1$ and $\mu_h$ are chains. Then  $(A,\mu_1,\dots,\mu_h,B)$ is a pre-clean twister of Type $2$.       \item\label{c:eecnc} $\mu_1$ is a chain and $\mu_h$ is not a chain. Then  the result follows from   \zcref{lem:killH}.
        \end{mycases}
        \item\label{c:ea} $A$ is a non-empty antichain. Then   the result follows from   \zcref{lem:killB}.
        \item\label{c:ec} $A$ is a non-empty chain.
       \begin{mycases}
            \item\label{c:ecc} $\mu_h$ is a chain. Then  the result follows from \zcref{lem:killCC}.
            \item\label{c:ecnc} $\mu_h$ is not a chain. Then   the result follows from \zcref{lem:killCS}.
        \end{mycases}
    \end{mycases}
\end{mycases} 
From the above, we obtain a pre-clean twister with length at most $2h+3$ and with order at least $m^{1/4}$.
Applying \zcref{lem:r,lem:last_clean} and choosing $m$ sufficiently large, we eventually get a clean twister with length at most $2h+3$ and with order at least $n$.
\end{proof}

We deduce a new version of~\zcref{thm:dependent} adapted to classes of tree-ordered weakly sparse graphs. 
\begin{thm}[store=thm:dependentTOWS]
	A class $\mathscr C$ of TOWS graphs is  monadically dependent if and only if for every $h\in\mathbb N$ there is some $n$ such that 
	no tree-ordered graph $G\in\mathscr C$ contains a clean twister of length $h$ and order $n$.
\end{thm}
\begin{proof}
This is a direct consequence of \zcref{thm:dependent,lem:tr2tw,lem:cleantw}.
\end{proof}

\subsection{Unavoidable induced substructures: Twists}
\label{sec:twists}
We now explicit the induced substructures of TOWS graphs (somehow) induced by clean twisters, which will be more suitable for our analysis.

\begin{ndefi}
A \ndef{twist} is 
a TOWS-graph $\mathbf M$ with a distinguished clean twister~$\mathfrak T$ of order at least $2$ covering all the elements of $\mathbf M$ but its root. 
For the sake of simplicity, we call $\mathfrak T$ the twister of $\mathbf M$, and implicitly denote by $I$ and $J$ the ordered sets used to index the meshes of $\mathfrak T$.
\end{ndefi}

The following fact is immediate and follows from the preservation of the properties defining clean twisters when restricting to an induced substructure.

 \begin{fact}\label{fact:restr}
    Let $\mathbf M$ be a TOWS graph
    and let $\mathfrak T=(A,\mu_1,\dots,\mu_h,B)$  be 
    a clean twister.
Let  $\mathbf N$ be the substructure of $\mathbf M$ induced by the union of $\{r\}$ and the domain of $\mathfrak T$.
    Then $\mathfrak T$ is a clean twister of $\mathbf N$. Hence, $\mathbf N$ is a twist.
\end{fact}

\begin{lem}
\label{lem:cov_twist}
    Let $\mathbf M$ be a twist with twister
    $\mathfrak T=(A,\mu_1,\dots,\mu_h,B)$. 

    Let $1\leq s<s+1=t\le h$.    
    If $(\mu_s,\mu_t)$ is matching in $\mathbf M^\prec$, then $(\mu_s(i,j),\mu_t(i,j))$ or $(\mu_t(i,j),\mu_{s}(i,j))$  is a cover in $\mathbf M^\prec$ for all $(i,j)\in I\times J$.
\end{lem}
\begin{proof}
    Let $I$ and $J$ be the index sets of the twister $\mathfrak T$.
    First note that $\mu_s$ and $\mu_t$ are both antichains, as otherwise the pair $(\mu_s,\mu_t)$ would not be matching in $\mathbf M^\prec$.
    Assume for contradiction that there exists $i,j\in I\times J$ and a vertex $z$ such that $\mu_s(i,j)\prec z\prec \mu_t(i,j)$.
    First note that $z\neq r$ as it is not minimum.
    Then $z$ does not belong to $\mu_s$ or $\mu_t$ as they are antichains. 
    Assume for contradiction that $z\in V(\mu_i)$. 
    Consider the pair $(\mu_s,\mu_i)$. According to \zcref{lem:incomparable}, 
    we have $s=1$ (hence $t=2$). It follows that $i>2$, hence $(\mu_s,\mu_i)$ is homogeneous (by \zcref{it:tw5}), implying that all the elements of $\mu_s$ are smaller than all the elements of $\mu_i$ (including $z$), hence all the elements of $\mu_s$ are smaller than $\mu_t(i,j)$, contradicting the assumption that $(\mu_s,\mu_t)$ is matching in $\mathbf M^\prec$.
    
    The vertex $z$ cannot be in $A$ as, according to \zcref{lem:preclean_prop}, $A\cup V(\mu_t)$ is an antichain.
    Similarly, $z$ cannot be in $B$ as, according to \zcref{lem:preclean_prop}, $B\cup V(\mu_s)$ is an antichain.
    
    Altogether, we proved by contradiction (using the symmetry between $\mu_s$ and $\mu_t$) that 
    $(\mu_s(i,j),\mu_t(i,j))$ or $(\mu_t(i,j),\mu_{s}(i,j))$  is a cover in $\mathbf M^\prec$ for all $(i,j)\in I\times J$.
\end{proof}

We define the following types of matching pairs in twists, where we define $x \mathbin{\Ddot{\parallel}} y:=\neg\bigl(x\prec:y)\vee (x:\succ y)\bigr)$ (see \zcref{fig:matchings}).
\begin{align*}
 M^{\rightarrow}&: &\mu_s(i,j)&\prec:\mu_{s+1}(i,j)& \text{ and }\qquad &\neg E(\mu_s(i,j),\mu_{s+1}(i,j));\\
M^{\leftarrow}&: &\mu_{s}(i,j)&:\succ\mu_{s+1}(i,j)&\text{ and }\qquad&\neg E(\mu_s(i,j),\mu_{s+1}(i,j));\\
M^{\underline{\rightarrow}}&: &\mu_s(i,j)&\prec:\mu_{s+1}(i,j)& \text{ and }\qquad &\phantom{\neg}E(\mu_s(i,j),\mu_{s+1}(i,j));\\
M^{\underline{\leftarrow}}&: &\mu_{s}(i,j)&:\succ\mu_{s+1}(i,j)& \text{ and }\qquad &\phantom{\neg}E(\mu_s(i,j),\mu_{s+1}(i,j));\\
M^{\underline{\phantom{\rightarrow}}}&: 
&\mu_s(i,j)&\mathbin{\Ddot{\parallel}}\mu_{s+1}(i,j)& 
\text{ and }\qquad &\phantom{\neg}E(\mu_s(i,j),\mu_{s+1}(i,j));\\
M^{\underline{\phantom{\rightarrow}}}_0&: 
&\max V(\mu_1)&\prec:\mu_{2}(i,j)& 
\text{ and }\qquad &\phantom{\neg}E(\mu_1(i,j),\mu_{2}(i,j));\\
M^{\underline{\phantom{\rightarrow}}}_1&: 
&\max V(\mu_1)&\prec:\min V(\mu_{2})& 
\text{ and }\qquad &\phantom{\neg}E(\mu_1(i,j),\mu_{2}(i,j));\\
M^{\underline{\phantom{\rightarrow}}}_2&: 
&\mu_{h-1}(i,j)&:\succ \max V(\mu_h)&
\text{ and }\qquad &\phantom{\neg}E(\mu_{h-1}(i,j),\mu_{h}(i,j)).
\end{align*}

\begin{figure}[ht]
    \centering
    \includegraphics[width=0.75\linewidth]{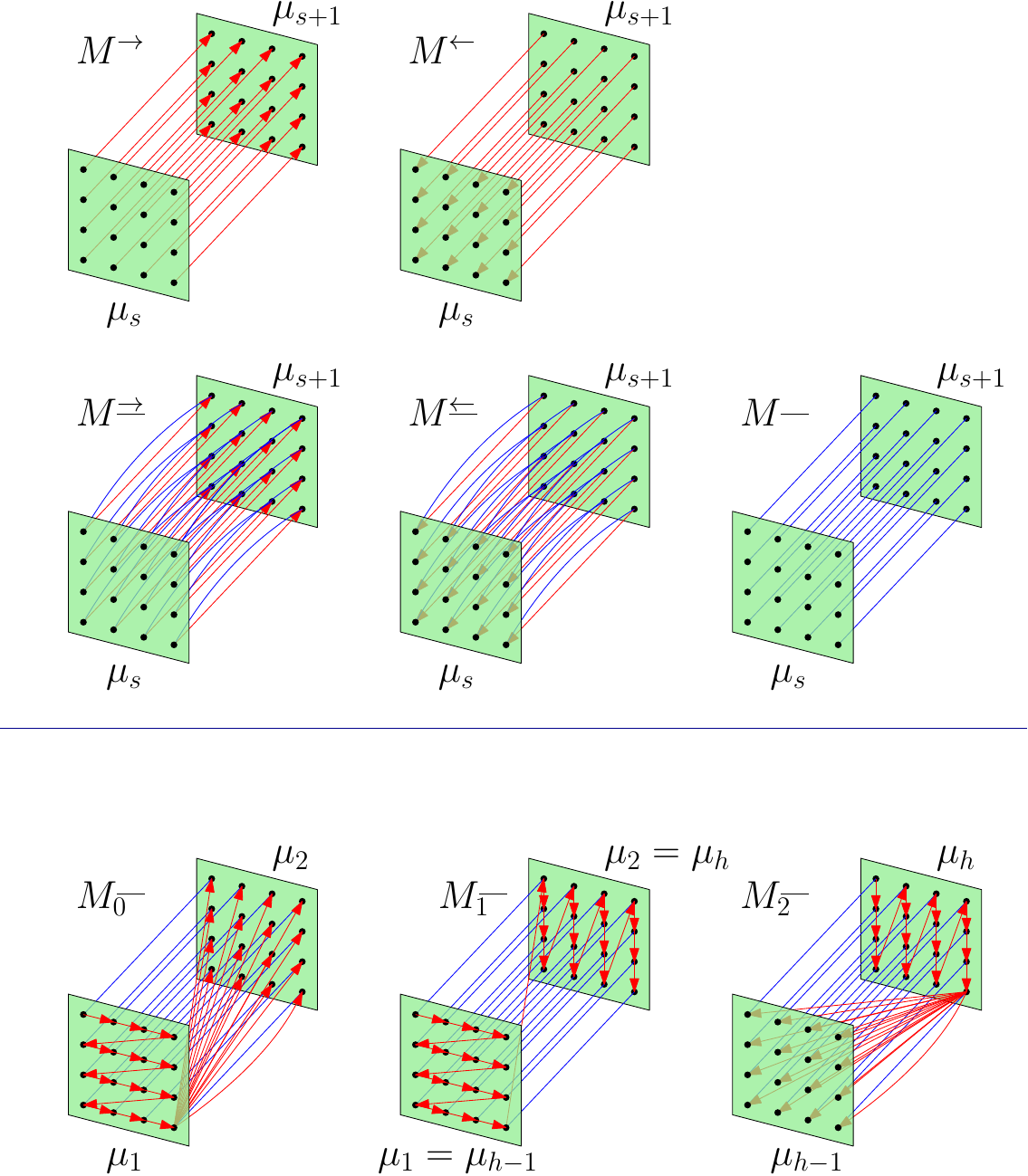}
    \caption{Matching pairs in twists. Red arrows correspond to covers in $\prec$, while blue edges correspond to $E$-relations.
    Three exceptional matching types (last line) may exist in type 3 twisters. For them, we also represented the cover relations inside the meshes.
    }
    \label{fig:matchings}
\end{figure}

The relevant notion for the relations between $\mu_1$ and its vertical guard (in a type 1 twister) or between $\mu_1$ and $\mu_2$ when the pair $(\mu_1,\mu_2)$ is not matching (in a type 3 twister) is the notion of a star.
A \ndef{star} is a type of relations between a set indexed by $I$ or $J$ and a mesh (indexed by $I\times J$) in a twister, which is of one of the following types, where $X=\{x(i)\colon i\in I\}$ stands for a set indexed by $I$ serving as vertical guard for $\mu_s$ and $Y=\{y(j)\colon j\in J\}$ is a set indexed by $J$ serving as a horizontal guard for $\mu_t$ (see \zcref{fig:stars}). As we shall see below, stars will only appear in clean twisters of type 1 
(with $x(i)=a(i)$ and $s=1$ or $y(j)=b(j)$ and $s=h$)
or  in clean twisters of type $3$
(with $x(i)=\mu_1(i,\max J)$ and $s=2$ or $y(j)=\mu_h(\max I,j)$ and $t=h-1$).

\begin{align*}
S^\rightarrow_v&: &x(i)&\prec: \mu_s(i,j)& \text{ and }\qquad &\neg E(x(i),\mu_s(i,j));\\
S^{\underline{\rightarrow}}_v&: &x(i)&\prec: \mu_s(i,j)& \text{ and }\qquad &\phantom{\neg}E(x(i),\mu_s(i,j));\\
S^{\underline{\phantom{\rightarrow}}}_v&: &x(i)&\parallel\mu_s(i,j)& \text{ and }\qquad &\phantom{\neg}E(x(i),\mu_s(i,j));\\
S^\leftarrow_h&: &\mu_t(i,j)&:\succ y(j)& \text{ and }\qquad &\neg E(\mu_t(i,j),y(j));\\
S^{\underline{\leftarrow}}_h&: &\mu_t(i,j)&:\succ y(j)& \text{ and }\qquad & \phantom{\neg}E(\mu_t(i,j),y(j));\\
S^{\underline{\phantom{\leftarrow}}}_h&: &\mu_t(i,j)&\parallel y(j)& \text{ and }\qquad &\phantom{\neg}E(\mu_t(i,j),y(j)).
\end{align*}

\begin{figure}[ht]
    \centering
    \includegraphics[width=0.75\linewidth]{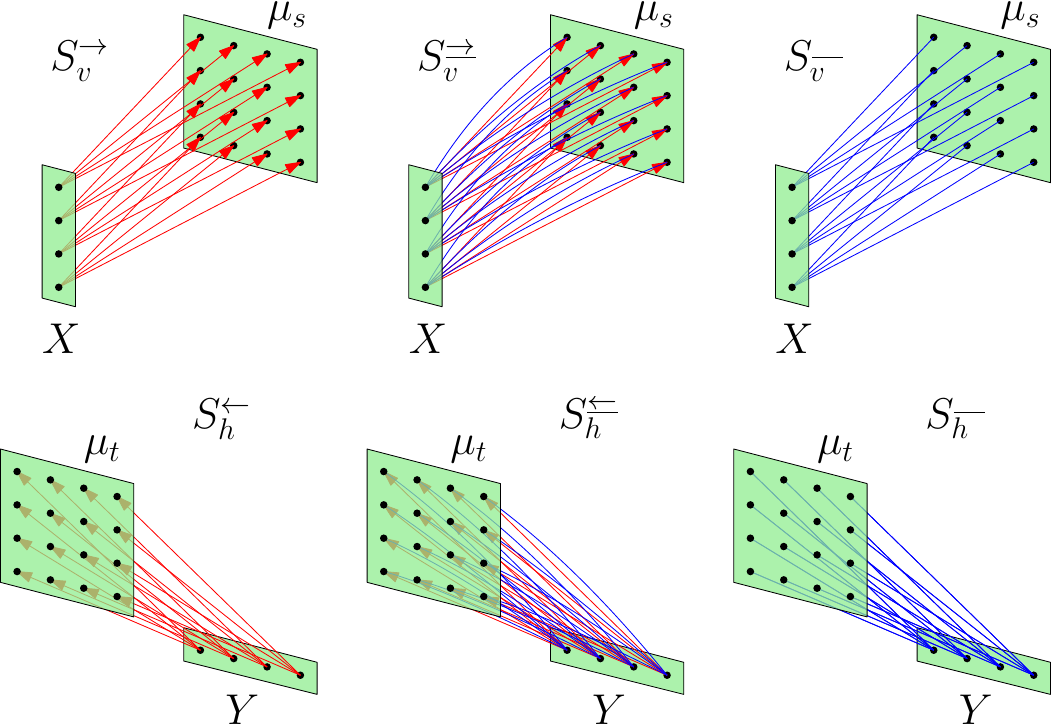}
    \caption{Different types of stars in a twist of type $1$ or $3$.}
    \label{fig:stars}
\end{figure}

\begin{lem}
\label{lem:twist1}
    Let $\mathbf M$ be a twist of type $1$
    with twister $\mathfrak T=(A,\mu_1,\dots,\mu_h,B)$ and order $n\geq 2\bomega(\mathbf M)+2$. 
    
    Then  $(A,\mu_1)$ is a star of one of the types $S^{{\rightarrow}}_v, S^{\underline{\rightarrow}}_v, S^{\underline{\phantom{\rightarrow}}}_v$ 
    and $(B,\mu_h)$ is a star of one of the types
    $S^{{\leftarrow}}_h, S^{\underline{\leftarrow}}_h, S^{\underline{\phantom{\leftarrow}}}_h$ 
    (See \zcref{fig:stars}).
\end{lem}
\begin{proof}
    As $A$ is a vertical guard of $\mathfrak T$, 
    $\otp(a(i),\mu_1(i',j'))$ depends only on 
    $\otp(i,i')$. 
    As $A$ is an antichain,
    no two distinct vertices of $A$ can be smaller than a same vertex of $\mu_1$.
    According to \zcref{lem:preclean_prop}, $\mu_1$ is also an antichain and thus no vertices of $\mu_1$ can be smaller than some vertex of $A$.
    Moreover, as the order of $\mathbf M$ is greater than $2\bomega(\mathbf M)+2$ we also cannot have $E(a(i),\mu_1(i',j'))$ for some $i\neq i'$. 

    Now assume that $A$ is a vertical guard of 
    $\mu_1$ in $\mathbf M^\prec$. In this case, we have $a(i)\prec \mu_1(i',j')$ if and only if $i=i'$. Assume for contradiction that there exists $z$ such that
    $a(i)\prec z\prec \mu_1(i,j)$.
    As $A$ and $\mu_1$ are antichains, $z\notin A\cup V(\mu_1)$. Also, $z\neq r$ as $r$ is minimum.
  Since $\mathbf M$ is a twist, we have  $z\in B$ or $z\in V(\mu_s)$ for $s>1$, hence a contradiction to \zcref{lem:preclean_prop}.
    Thus, the pair $(A,\mu_1)$ is a star of one of the types $S^{{\rightarrow}}_v, S^{\underline{\rightarrow}}_v, S^{\underline{\phantom{\rightarrow}}}_v$.
    The proof for the pair $(B,\mu_h)$ is alike.
\end{proof}

 Note that the union of the cover graph and the edge-relation of a clean twister of type $1$ with order $n$ and length $h$ is the $h$-subdivision of $K_{n,n}$, see
 \zcref{sfig:typ1} for an example.

 \begin{figure}[h!t]
	\begin{center}
	\includegraphics[width=.85\columnwidth]{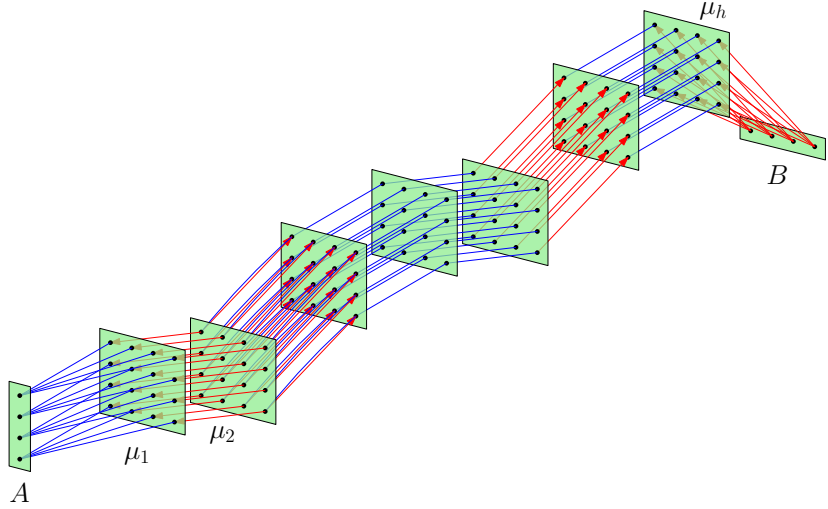}
\end{center}
\caption{The twister of a twist of type 1. 
In this example,  $(A,\mu_1)$ is a star of type $S^{\underline{\phantom{\rightarrow}}}_v$ and 
$(\mu_h,B)$ is a star of type $S^{{\leftarrow}}_h$.
\label{sfig:typ1}}
\end{figure}

\begin{lem}
\label{lem:twist2}
    Let $\mathbf M$ be a twist of type $2$, 
    with twister $\mathfrak T=(\emptyset,\mu_1,\dots,\mu_h,\emptyset)$ and  order $n\geq 2\bomega(\mathbf M)+2$. 
    
    Then  $(\mu_1,\mu_2)$ and
 $(\mu_{h-1},\mu_h)$ are matchings of type $M^{\underline{\phantom{\rightarrow}}}$.  
\end{lem}
\begin{proof}
By \zcref{lem:preclean_prop}, the vertices of $\mu_1$ and $\mu_2$ are incomparable. Hence, according to \zcref{it:tw4}, $(\mu_1,\mu_2)$ is matching in $\mathbf M^E$. Thus $(\mu_1,\mu_2)$ is a matching of of type $M^{\underline{\phantom{\rightarrow}}}$. The case of $(\mu_{h-1},\mu_h)$ is proved in the same way.
\end{proof}

\begin{lem}\label{lem:twist2_mu1}
Let $\mathbf M$ be a twist of type 2, with
twister $\mathfrak T = (\emptyset, \mu_1,
\ldots, \mu_h, \emptyset)$ and order
$n\ge 2\bomega(\mathbf M) + 2$.

Then for every $i\in I$ and $j<_J: j' \in J$,
$\mu_1(i,j) \prec\mu_1(i,j')$ is a cover
in $\mathbf M$.
Similarly, for every $j\in J$ and $i<:i'\in I$, $\mu_h(i,j)\prec\mu_h(i',j)$ is a cover
in $\mathbf M$.
\end{lem}

\begin{proof}
    Let $i\in I$ and $j<_J: j' \in J$.
    By \zcref{lem:twist2}, $(\mu_1, \mu_2)$
    is a matching of type $M^{\underline{\phantom{\rightarrow}}}$,
    so no vertex of $V(\mu_2)$ can be
    ordered between $\mu_1(i,j)$ and
    $\mu_1(i,j')$ in~$\prec$.
    Moreover, $\mu_1$ is homogeneous with
    $\mu_s$ for $s>2$ by \zcref{it:tw5}, so
    no vertex of such a mesh can be compared
    differently with $\mu_1(i,j)$ and
    $\mu_1(i,j')$. Hence, $\mu_1(i,j)\prec \mu_1(i,j')$ is a cover.

    The proof for $\mu_h$ goes \emph{mutatis mutandis}.
\end{proof}

\begin{figure}[h!t]
	\begin{center}
\includegraphics[width=.85\columnwidth]{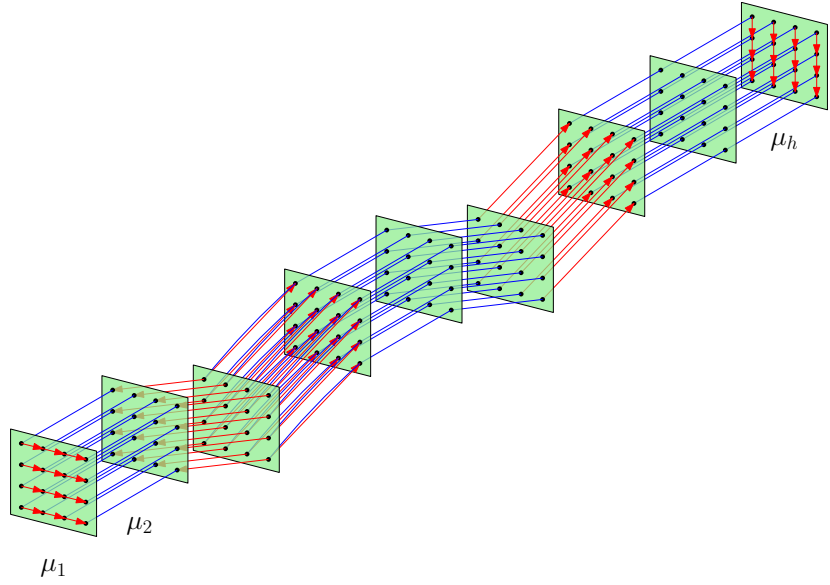}
\end{center}
\caption{The twister of a twist of type 2. The pairs $(\mu_1,\mu_2)$ and $(\mu_{h-1},\mu_h)$ are both matching of type $M^{\underline{\phantom{\rightarrow}}}$.}
\label{sfig:type2}
\end{figure}

\begin{lem}
\label{lem:twist3}
    Let $\mathbf M$ be a twist of type $3$ 
    with twister  $\mathfrak T=(\emptyset,\mu_1,\dots,\mu_h,\emptyset)$ and order $n\geq 2\bomega(\mathbf M)+2$. 

    Then
    \begin{itemize}
        \item either $(\mu_1,\mu_2)$ is a matching of type $M^{\underline{\phantom{\rightarrow}}}, M^{\underline{\phantom{\rightarrow}}}_0$, or
$M^{\underline{\phantom{\rightarrow}}}_1$, 
        \item or $(\mu_1(\cdot,\max J),\mu_2)$ is a star of type $S^{{\rightarrow}}_v$.
    \end{itemize}

    Let $i_{\text{max}}=\max I$ if the order type of the restriction of $\prec$ to $V(\mu_h)$ is $<_{I,J}$ or $<_{I,\bar J}$ and $i_{\text{max}}=\min I$, otherwise. Then
    
    \begin{itemize}
\item   either $(\mu_{h-1},\mu_h)$ is a matching of type $M^{\underline{\phantom{\rightarrow}}}$, $M^{\underline{\phantom{\rightarrow}}}_1$, or
$M^{\underline{\phantom{\rightarrow}}}_2$,
\item or $(\mu_{h-1},\mu_h(i_{\text{max}},\cdot))$ is a star of type $S^{{\leftarrow}}_h$.
\end{itemize}
\end{lem}
\begin{proof}
Assume $(\mu_1,\mu_2)$ is matching. Then,  either $V(\mu_1)$ and $V(\mu_2)$ are incomparable and $(\mu_1,\mu_2)$ has type 
$M^{\underline{\phantom{\rightarrow}}}$, 
or $V(\mu_1)$ is smaller than $V(\mu_2)$. In the latter case, 
if $h>2$, by regularity, either no pair $(\max V(\mu_1),\mu_2(i,j))$ is a cover (type 
$M^{\underline{\phantom{\rightarrow}}}$) or all of them are covers (type $M^{\underline{\phantom{\rightarrow}}}_0$). If $h=2$, then  $(\max V(\mu_1),\min V(\mu_2))$ is a cover and $(\mu_1,\mu_2)$ has type 
$M^{\underline{\phantom{\rightarrow}}}_1$.

Otherwise, $(\mu_1,\mu_2)$ is simply vertical (and thus $h>2$), meaning that $V(\mu_1)\cup V(\mu_2)$ is independent and 
\begin{equation}
\label{eq:1}
\mu_1(i,j)\prec \mu_2(i',j')\text{ if and only if } i\leq i'.
\end{equation}

In particular, $\mu_1(i,\max J)\prec\mu_2(i,j)$ for every $i\in I$ and $j\in J$.
Assume it is not a cover, that is, that there exists a vertex $z$ such that $\mu_1(i,\max J)\prec z\prec \mu_2(i,j)$.
Clearly, $z$ does not belong to $\mu_1$ 
(as $\mu_1(i,\max J)$ is, according to \zcref{eq:1}, the maximum element of $\mu_1$ that is smaller than $\mu_2(i,j)$) and $z$ does not belong to $\mu_2$ as it is an antichain. 

Assume for contradiction that $z=\mu_s(\alpha,\beta)$ for some $s>2$, $\alpha\in I$ and $\beta\in J$.
As $(\mu_1,\mu_s)$ is homogeneous, we deduce $\mu_1(i',j')\prec z$ (hence $\mu_1(i',j')\prec \mu_2(i,j)$) for all $i'\in I$ and $j'\in J$. Then  it follows from \zcref{eq:1} that $i=\max I$. 
So, we have $\mu_1(\max I,\max J)\prec \mu_s(\alpha,\beta)\prec\mu_2(\max I,j)$.
By regularity of the pair $(\mu_s,\mu_2)$, there exists some $\alpha'\leq \alpha$ and $i''<\max I$ such that $\mu_s(\alpha',\beta)\prec\mu_2(i'',j)$. Hence,
$\mu_1(\max I,\max J)\prec \mu_2(i'',j)$ though $\max I>i''$, contradicting \zcref{eq:1}. 

Thus, $\mu_1(i,\max J)\prec:\mu_2(i,j)$ for every $i\in I$ and $j\in J$ and thus $(\mu_1(\cdot,\max J),\mu_2)$ is a star of type $S^{{\rightarrow}}_v$.

The case of $(\mu_{h-1},\mu_h)$ is proved in a similar way (where $i_\text{max}$ plays the role of $\max J$). 
\end{proof}

\begin{figure}[h!t]
\begin{center}
\includegraphics[width=.85\columnwidth]{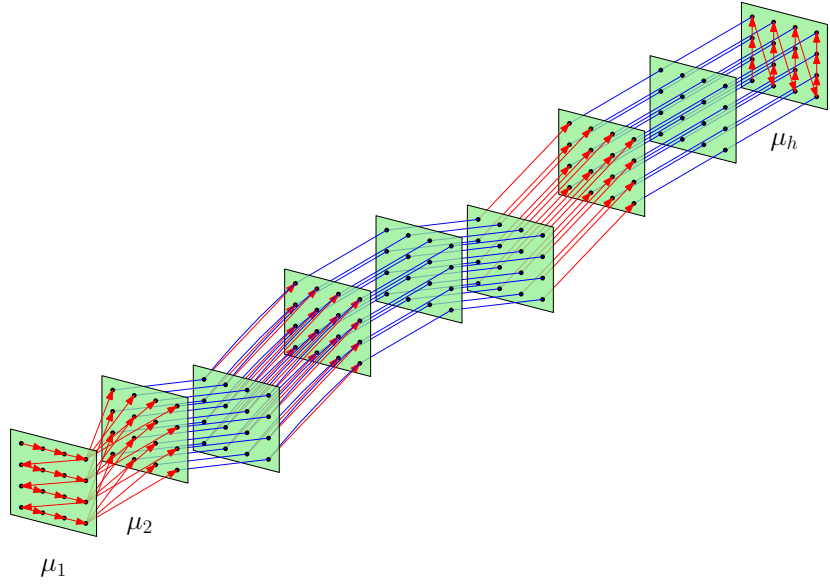}
\end{center}
\caption{The twister of a twist of type 3. In this example, the pair 
$(\mu_1(\cdot,\max J),\mu_2)$ is a star of type $S^{{\rightarrow}}_v$ and the pair $(\mu_{h-1},\mu_h)$ is a matching of type $M^{\underline{\phantom{\rightarrow}}}$.}  
\label{sfig:type3}
\end{figure}

According to the above lemmas, we can give a full descriptions of cores. 

A \ndef{core} $\tau$ with length $h$ is a tree-ordered graph with root $r$, with three possible types.

\begin{ndefi}[Core of type 1]
A \ndef[type 1]{core} $\tau$ of length $h$ with type $1$
has its domain partitioned into sets $\{r\}$, 
\begin{align*}
    A&=\{a(i)\colon i\in [2]\},\\
    V(\mu_s)&=\{\mu_s(i,j)\colon i,j\in [2]\} \qquad (1\le s\le h),\\
    \text{and }B&=\{b(j)\colon j\in [2]\}.
\end{align*}

The sets $A, V(\mu_1),\dots, V(\mu_h), B$ are the \ndef[of a core]{layers} of $\tau$ and are independent, and each is an antichain. Each layer is either included in the $E$-neighborhood of $r$ or does not intersect it. 

No two non-consecutive layers contain comparable vertices.
With this constraint, and the fact the order is a tree-order, 
the core is built by choosing to which kind  each pair of successive layers belongs to (see \zcref{fig:core1}).

The restriction of $\tau$ to the pair $(A,\mu_1)$ is a star of one of the types $S^{{\rightarrow}}_v, S^{\underline{\rightarrow}}_v, S^{\underline{\phantom{\rightarrow}}}_v$.

For $1\le s<h$, the restriction of $\tau$ to the pair $(\mu_s,\mu_{s+1})$
has type $M^{\rightarrow}, M^{\leftarrow}, M^{\underline{\rightarrow}}, M^{\underline{\leftarrow}}$, or $M^{\underline{\phantom{\rightarrow}}}$. 

The restriction of $\tau$ to the pair $(B,\mu_h)$ is a star of one of the types    $S^{{\leftarrow}}_h, S^{\underline{\leftarrow}}_h, S^{\underline{\phantom{\leftarrow}}}_h$. 
\end{ndefi}

\begin{figure}[ht]
    \centering
    \includegraphics[height=.27\textheight]{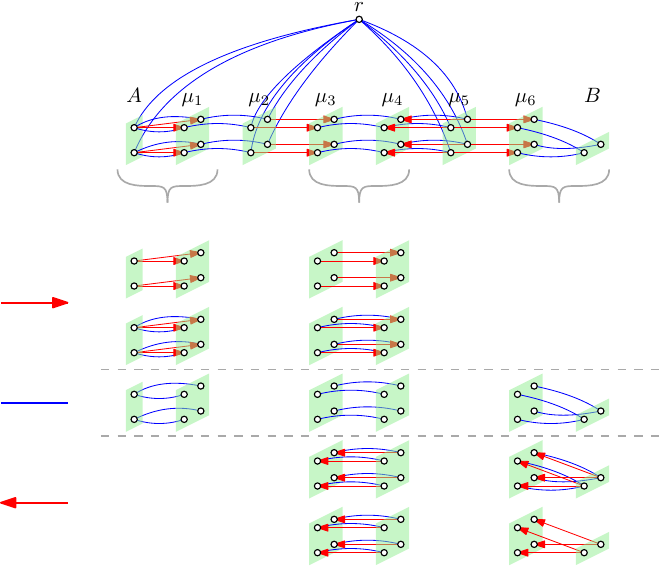}
    \caption{A type 1 core of length $6$. (For sake of readability, we did not draw the covers involving the root.)}
    \label{fig:core1}
\end{figure}

\begin{ndefi}[Core of type 2]
A \ndef[type 2]{core} $\tau$ of length $h$ with type $2$
has its domain partitioned into sets $\{r\}$ and 
    $V(\mu_s)=\{\mu_s(i,j)\colon i,j\in [2]\}$ (for $1\le s\le h$).

The sets $V(\mu_1),\dots$, and  $V(\mu_h)$ are the \emph{layers} of $\tau$ and are independent. The restriction of $\prec$ to $V(\mu_1)$ is $<_I$, $V(\mu_s)$ is an antichain for $1<s<h$, and the restriction of $\prec$ to $V(\mu_h)$ is $<_J$.
Each layer is either included in the $E$-neighborhood of $r$ or does not intersect it.

No two non-consecutive layers contain comparable vertices.
With this constraint, and the fact the order is a tree-order, 
the core is built by choosing to which kind  each pair of successive layers belongs to (see \zcref{fig:core2}).

The pairs $(\mu_1,\mu_2)$ and $(\mu_{h-1},\mu_h)$ are matchings of type $M^{\underline{\phantom{\rightarrow}}}$.

For $1<s<h-1$, the restriction of $\tau$ to the pair $(\mu_s,\mu_{s+1})$
has type $M^{\rightarrow}, M^{\leftarrow}, M^{\underline{\rightarrow}}, M^{\underline{\leftarrow}}$, or $M^{\underline{\phantom{\rightarrow}}}$. 
\end{ndefi}

\begin{figure}[ht]
    \centering
    \includegraphics[height=.27\textheight]{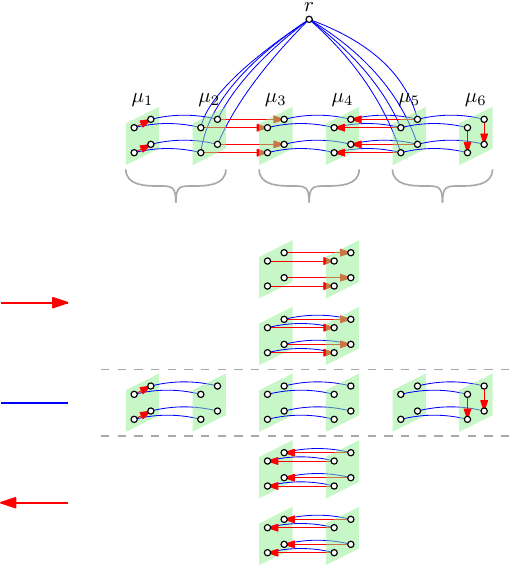}
    \caption{A type 2 core of length $6$. (For sake of readability, we did not draw the covers involving the root.)}
    \label{fig:core2}
\end{figure}

\begin{ndefi}[Core of type 3]
A \ndef[type 3]{core} $\tau$ of length $h$ with type $3$
has its domain partitioned into sets $\{r\}$ and 
    $V(\mu_s)=\{\mu_s(i,j)\colon i,j\in [2]\}$ (for $1\le s\le h$).

The sets $V(\mu_1),\dots$, and  $V(\mu_h)$ are the \emph{layers} of $\tau$ and are independent. The restriction of $\prec$ to $V(\mu_1)$ is $<_{IJ}$, $V(\mu_s)$ is an antichain for $1<s<h$, and the restriction of $\prec$ to $V(\mu_h)$ is a linear order of anti-lexicographic type.
Each layer is either included in the $E$-neighborhood of $r$ or does not intersect it. 

No two non-consecutive layers within $V(\mu_2),\dots,V(\mu_{h-1})$ contain comparable vertices.
Each layer $V(\mu_s)$ with $s>1$ may be greater than $V(\mu_1)$, and each layer~$V(\mu_s)$ with $1<s<h$ may be greater than $V(\mu_h)$.
With these constraints, and the fact the order is a tree-order, 
the core is built by choosing to which kind  each pair of successive layers belongs to (see \zcref{fig:core3}).

Either $(\mu_1,\mu_2)$ is a matching of type $M^{\underline{\phantom{\rightarrow}}}$ (with possibly $V(\mu_1)$ smaller than~$V(\mu_2)$), or 
there is no $E$-relation between $V(\mu_1)$ and $V(\mu_2)$ and we have $\mu_1(i,j)\prec \mu_2(i',j')$ if $i'\geq i$.

For $1<s<h-1$, the restriction of $\tau$ to the pair $(\mu_s,\mu_{s+1})$
has type $M^{\rightarrow}, M^{\leftarrow}, M^{\underline{\rightarrow}}, M^{\underline{\leftarrow}}$, or $M^{\underline{\phantom{\rightarrow}}}$. 

Either $(\mu_{h-1},\mu_h)$ is a matching of type $M^{\underline{\phantom{\rightarrow}}}$ (with possibly~$V(\mu_h)$ smaller than~$V(\mu_{h-1})$), or 
there is no $E$-relation between $V(\mu_{h-1})$ and $V(\mu_h)$ and we have $\mu_h(i,j)\prec \mu_{h-1}(i',j')$ if $j'\geq j$ (or if $j'\leq j$, depending on the type of anti-lexicographic order on $V(\mu_h)$).
\end{ndefi}

\begin{figure}[ht]
    \centering
    \includegraphics[height=.27\textheight]{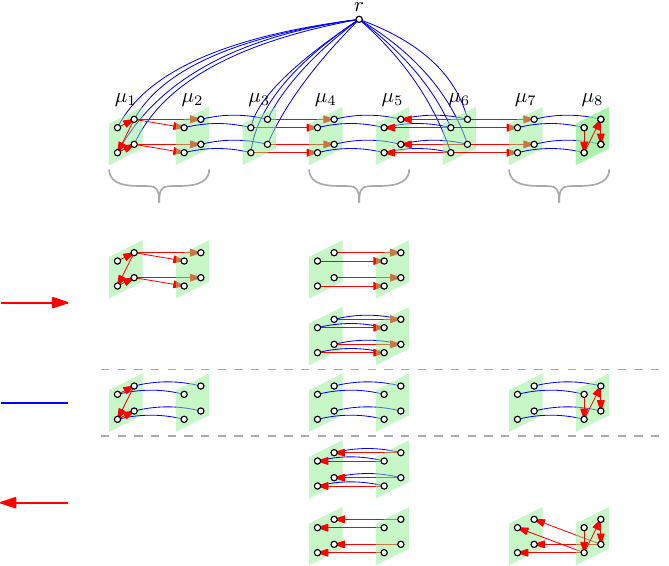}
    \caption{A type 3 core of length $8$. (For sake of readability, we did not draw some covers involving the root, $\max \mu_1$, and $\max \mu_8$.)}
    \label{fig:core3}
\end{figure}

\begin{ndefi}\label{def:core}
The \ndef[of a twist]{core} of a twist $\mathbf M$ is the substructure of $\mathbf M$ induced by a subtwister of $\mathbf T$ of order $2$ and the root $r$. (Note that the core of a twist is a twist.) 
\end{ndefi}

\begin{lem}\label{lem:twist}
    For every core $\tau$ and every integer $n\geq 2$, there exists a unique twist $\tau[n]$ with core $\tau$. Moreover, given $\mathbf N$, the twist $\tau[n]$ can be constructed in time $\mathcal{O}(n^4)$.
\end{lem}
\begin{proof}
Let $\mathbf M$ be a twist with twister $(A,\mu_1,\dots,\mu_h,B)$.
As in the proof of \zcref{cl:restr}, it follows from the regularity of  the pairs $(A,\mu_s)$, $(\mu_s,\mu_t)$, and   $(\mu_s,B)$ and the homogeneity of $r$ with $A,B$ and each of the $\mu_s$ that $\mathbf M$ is completely determined by its order and its core.

The structure $\tau[n]$ has $\mathcal O(n^2)$ vertices, and is easily constructed in time $\mathcal O(n^4)$ using regularity.
\end{proof}

\label{p:lics_twists}
\begin{thm}[store*=thm:lics_twists,restate-keys={note=see \zcpageref[nocap]{p:lics_twists}}]
	A hereditary class $\mathscr C$ of TOWS graphs is  independent if and only if there exists a core $\tau$ such that 
    $\{\tau[n]\colon n\in\mathbb N\}\subseteq\mathscr C$.
\end{thm}
\begin{proof}
    Assume that $\mathscr C$ is independent. Then, according to \zcref{thm:dependentTOWS}, there exists some integer $h$ such that for every integer $n$ some TOWS graph $G_n\in\mathscr C$ contains a clean twister of length $h$ and order~$n$.
    As $\mathscr C$ is hereditary, we deduce (by \zcref{fact:restr}) that for every integer $n$, $\mathscr C$ contains a twist of length $h$ and order~$n$.
    By \zcref{lem:cov_twist,lem:twist1,lem:twist2,lem:twist2_mu1,lem:twist3}, if $n\geq 2\bomega(\mathscr C)+2$, then 
    every subtwist of order $2$ of a twist of order $n$ in $\mathscr C$ is a core. As there are only finitely many cores with length $h$, according to \zcref{lem:twist},  there exists a core $\tau$ with length $h$ such that $\{\tau[n]\colon n\in\mathbb N\}\subseteq \mathscr C$.

    The reverse implication follows directly from the fact that $\{\tau[n]\colon n\in\mathbb N\}$ contains arbitrarily large clean twisters of length $h$, hence is independent (by \zcref{thm:dependentTOWS}).
\end{proof}

\section{Characterization by (induced) tree-ordered minors}
\label{sec:Iminor}
In this section, we prove the following characterization theorem for tree-ordered weakly sparse graph classes:
\getkeytheorem[body]{thm:mainG}
\subsection{Tree-order starification}
\label{sec:star}
In this section we introduce starifications of tree-ordered graphs and prove that nowhere denseness of the class of all the starifications of a  class $\mathscr C$ of tree-ordered weakly sparse graphs is equivalent to monadic dependence of $\mathscr C$.

We first consider tree-orders.
A \ndef{starification} $Q$ of a tree-order $(V,\prec)$ is a star forest obtained by partitioning the cover graph of the tree order into subtrees, then replacing each part by a star whose center is the root of the subtree.
The graph $Q$ can be obtained by means of a transduction, by coloring by a color~$A$ the minimum of each part (see \zcref{fig:starification}). Then the adjacency formula is
\begin{align*}
\rho(x,y)&:=\rho_0(x,y)\vee\rho_0(y,x),
\intertext{where}
\rho_0(x,y)&:=A(x)\wedge  (x\prec y)\wedge \forall z\ \bigl(
x\prec z\preceq y\ \Rightarrow\ \neg A(z)\bigr).
\end{align*}
\begin{figure}[ht]
    \centering
    \includegraphics[width=\columnwidth]{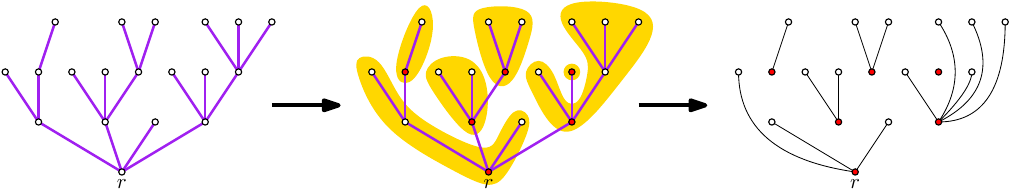}
    \caption{Starification of a tree-order. The minimum vertex of each part of a partition of the cover graph into subtrees is marked. Then  each part is turned into a star whose center is the minimum of the part.}
    \label{fig:starification}
\end{figure}
A \ndef{starification} of a tree-ordered graph $\mathbf M$  is a graph with same domain, whose edge set is the union of the edge set of $\mathbf M$ and the edge set of a starification of $\mathbf M^\prec$ \zcref{fig:starification2}).
It is obtained as a transduction $\mathsf{St}$ of $\mathbf M$, where the adjacency formula $\rho$ is slightly modified, compared to the above:
$\rho(x,y):=\rho_0(x,y)\vee\rho_0(y,x)\vee E(x,y)$.

\begin{figure}[ht]
    \centering
    \includegraphics[width=\columnwidth]{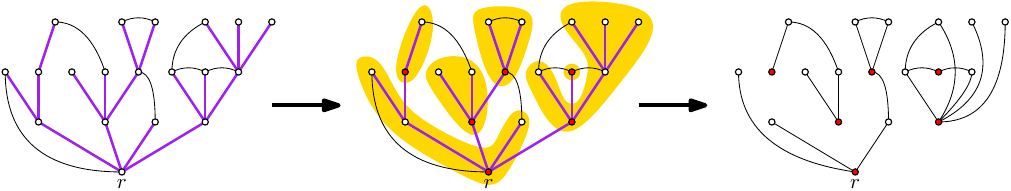}
    \caption{Starification of a tree-ordered graph.}
    \label{fig:starification2}
\end{figure}
\pagebreak

\begin{lem}
\label{lem:StND}
    Let $\mathscr C$ be a  monadically  dependent    class of tree-ordered weakly sparse graphs. 
    Then $\mathsf{St}(\mathscr C)$ is nowhere dense.
\end{lem}
\begin{proof}
    That $\mathsf{St}(\mathscr C)$ is monadically dependent follows from the fact that $\mathsf{St}(\mathscr C)$ is a transduction of $\mathscr C$. 
    According to \zcref{cor:ws}, that $\mathsf{St}(\mathscr C)$ is weakly sparse follows from the fact that the edge set of the graphs in $\mathsf{St}(\mathscr C)$ is obtained as the union of the edge sets of two weakly sparse graphs, namely the $E$-reducts of the tree-ordered graphs in $\mathscr C$ and a star-forest obtained from the tree-order. 
    Being both weakly sparse and monadically dependent,  $\mathsf{St}(\mathscr C)$ is nowhere dense \cite{msrw}.
\end{proof}
\subsection{{\Tim}s and tree-ordered minors  of TOWS graphs}

We denote by $\Mon(\mathscr M)$ the set of all the structures obtained by $\sigma$-deletions on structures in $\mathscr M$. 
The next fact is direct from the definitions.
\begin{fact}
    \label{fact:cont_min}
    Let $\mathscr M$ be a class of tree-ordered $\sigma$-structures. 
    Then 
\[
\Minor(\mathscr M)=\Mon(\Cont(\mathscr M))=\Cont(\Mon(\mathscr M)).
\]
\end{fact}

\begin{ndefi}[$\mathsf{Shrink}$]
\label{def:shrink_G}
We define the interpretation $\mathsf{Shrink}$ \index{Shrink (of a tree-ordered graph)} of tree-ordered graphs in tree-ordered $\{E,\prec,D,V\}$-structures, where $D$ and $V$ are unary relations, by 

\begin{align*}    
    \nu(x)&:=V(x)\wedge\neg D(x)\vee (x=r)\\
    \rho_E(x,y)&:=\exists x',y'\ \biggl(E(x',y')\wedge (x'\preceq x)\wedge(y'\preceq y)\wedge (\neg D(x')\vee (x'=r)\\
    &\qquad
     \wedge(\neg D(y')\vee (y'=r))
    \wedge\bigl(\forall z\ ((x'\prec z\preceq x)\vee(y'\prec z\preceq y)) \rightarrow\neg \nu(z)\bigr)
\end{align*}
\end{ndefi}

We stress the following easy fact, which justifies to consider that $\Cont$ is a transduction.
\begin{fact}
\label{fact:cont_trans}
    The transduction $\mathsf T$ defined as $\mathsf T=\mathsf{Shrink}\circ \Lambda$ (where $\Lambda$ is the monadic expansion adding predicates $D$ and $V$) is
    such that, for every class $\mathscr C$ of tree-ordered graphs we have
    $\Cont(\mathscr C)=\mathsf T(\mathscr C)$.
\end{fact}
\begin{proof}
    Let $\mathbf M$ be a tree-ordered graph.
    For every expansion $\mathbf M^+$ of $\mathbf M$ by unary relations $V$ and $D$, $\mathsf{Shrink}(\mathbf M^+)\in\Cont(\mathbf M)$.
    Conversely, for every $\mathbf N\in \Cont(\mathbf M)$, by marking $D$ the  deleted vertices and $V$ the minimum vertex of each contracted part, we check that $\mathbf N\in\mathsf T(\mathbf M)$.
\end{proof}

\begin{lem}
\label{lem:St2Min}
    Let $\mathbf M$ be a tree-ordered graph.
    Then  
    \[\Cont(\mathbf M)^E\subseteq \Minor(\mathbf M)^E\subseteq (\mathsf{St}\circ\mathsf{Her}(\mathbf M))\shm 1,\]
    where $\mathsf{Her}$ is the hereditary closure transduction.
\end{lem}
\begin{proof}
Let $\mathbf M$ be a tree-ordered graph and let $G\in \Minor(\mathbf M)^E$. The graph $G$ is a subgraph of the $E$-reduct $\mathbf N^E$ of an {\tim}   $\mathbf N$ of $\mathbf M$. 
By definition, $\mathbf N$ is obtained from an induced substructure $\mathbf M'$ of $\mathbf M$ by a sequence of elementary $\prec$-contractions. Hence, $\mathbf N$
has  vertex set $\mathscr P$, where $\mathscr P$ is a partition of the domain of $\mathbf M'$ into subtrees. Let $S$ be the starification of 
$\mathbf M'$ induced by this partition. Then $\mathbf N^E$ is obtained from $S$ by contracting a star forest, and the subgraph $G$ of~$\mathbf N^E$ also belongs to $\mathsf{St}(\mathbf M')\shm 1$.
\end{proof}

\begin{lem}
\label{lem:Min}
    Let $\mathscr C$ be a   monadically dependent class of tree-ordered weakly sparse graphs. Then $\Minor(\mathscr C)^E$ (\/hence $\Cont(\mathscr C)^E$\/) is nowhere dense.
\end{lem}
\begin{proof}
    According to \zcref{lem:StND}, $\mathsf{St}\circ\mathsf{Her}(\mathscr C)$ is nowhere dense. As nowhere denseness is preserved by shallow minors, the property follows from \zcref{lem:St2Min}. 
\end{proof}

\subsection{Racks and groundings}

In this Section, we prove that the monadic independence of a hereditary class of TOWS graphs is witnessed by the existence of 
specific {\tim}s, which we call \ndef{racks} and \ndef{groundings}.
In essence,
a rack or a grounding is a tree-ordered graph $\mathbf M$ obtained by subdividing a bipartite graph~$G$ and turning this subdivision into a tree-ordered graphs in a simple, highly regular way.
In particular,  the subdivision of $G$ 
can easily be retrieved from $\mathbf M$, by taking the $E$-reduct of $\mathbf M$, then deleting the root.
\medskip

First, we define the (simpler) notion of a grounding. In short, the grounding of a bipartite graph is obtained from a subdivision by adding a root and making  this root adjacent to some layers of the subdivision.

\begin{ndefi}[Grounding]
Let $\mathfrak G=(h,N_r)$, where
$h$ is a non-negative integer and $N_r\subseteq \{0,\dots,h+1\}$.

Let $G=(A,B,F)$ be a bipartite graph, let $G^{(h)}$ be the $h$-subdivision of $G$ and $L_i$ be the set of the subdivision vertices of $G^{(h)}$ at distance $i$ from $A$.

We define the \ndef{$\mathfrak G$-grounding} of $G$ as the tree-ordered graph $\mathbf M$ obtained from the disjoint union of $G^{(h)}$ and the singleton $\{r\}$ by defining the tree-order $x\prec y:=(x=r)\wedge(y\neq r)$ and
 the neighborhood (in $E$) of $r$ as $\bigcup_{i\in N_r}L_i$.
 
The graph $G$ is the \ndef{reduction} of $\mathbf M$. We also say that $\mathbf M$ is a \emph{$\mathfrak G$-grounding} with reduction $G$ and length $h$.
\end{ndefi}

The tree-order of a $\mathfrak G$-grounding is somewhat trivial. Thus, we cannot hope
to always find such simple structures in a monadically independent class of TOWS graphs as we only proceed to {\tim}s.
Indeed, any {\tim} of a class of ordered graphs (which is a special case of tree-ordered graphs) is a class of ordered graphs. 

Contrary to the previous case, we now consider a bipartite graph whose parts are independently linearly ordered.
Formally, an \ndef{ordered bipartite graph} is a triple
$(G,<_A,<_B)$, where $G$ is a bipartite graph, $<_A$ is a linear order on the first part of~$G$, and $<_B$ is a linear order on the second part of $G$.
Then  
 the rack of the ordered bipartite graph $(G,<_A,<_B)$ is basically obtained from a subdivision of $G$ by adding a root, making  this root adjacent to some layers of the subdivision, and assigning as a predecessor of a layer either the root, or the maximum of the first part of the graph, or the maximum of the second part of the graph (See \zcref{fig:rack}).

\begin{ndefi}[Rack]
Let $\mathfrak R=(h,N_r,C_A,C_B)$, where
$h$ is a non-negative integer, $N_r\subseteq \{0,\dots,h+1\}$, $C_A\subseteq [h+1]$, $C_B\subseteq [h]$, and $C_A\cap C_B=\emptyset$.

Let $(G,<_A,<_B)$ be an ordered bipartite graph with parts $A$ and $B$, let $G^{(h)}$ be the $h$-subdivision of $G$ and $L_i$ be the set of the subdivision vertices of $G^{(h)}$ at distance~$i$ from $A$. Let $\pi_A$ and $\pi_B$ be the predecessor functions of $<_A$ and $<_B$.

The \ndef{$\mathfrak R$-rack} of $(G,<_A,<_B)$ is the tree-ordered graph $\mathbf M$ obtained from the disjoint union of $G^{(h)}$ and the singleton $\{r\}$ as follows:

Let 
$L_0=A, L_0'=\{\min A\}$, $L_i'=L_i$ (for $i\in [h]$), $L_{h+1}=B$, and $L_{h+1}'=\{\min B\}$.
Then  the $E$-neighborhood of $r$ is $\bigcup_{i\in N_r}L_i$ and the tree-order $\prec$ is defined from its  predecessor function $\pi$, which is defined by

\[
\pi(x)=\begin{cases}
    \pi_A(x)&\text{if }x\in A\setminus L_0',\\
    \pi_B(x)&\text{if }x\in B\setminus L_{h+1}',\\
    \max A&\text{if }(\exists i\in C_A)\ x\in L_i',\\
    \max B&\text{if }(\exists i\in C_B)\ x\in L_i',\\
    r&\text{otherwise.}
\end{cases}
\]
The  graph $G$ is the \ndef{reduction} of $\mathbf M$. We also say that $\mathbf M$ is an \emph{$\mathfrak R$-rack} with reduction $G$ and length $h$.
\end{ndefi}

\begin{figure}[ht]
    \centering
    \includegraphics[width=.75\linewidth]{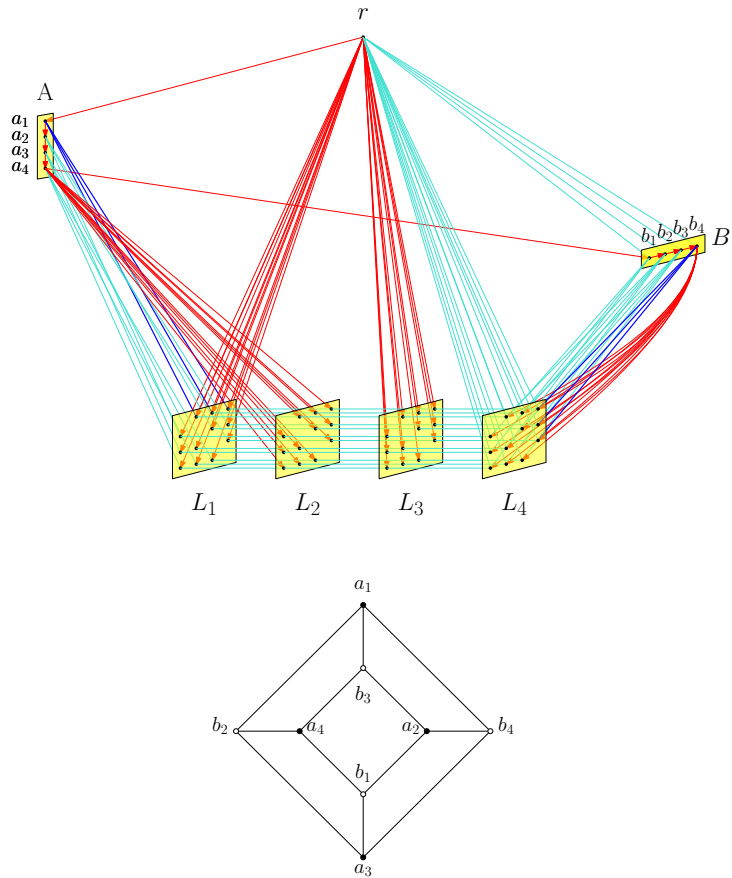}
    \caption{An $\mathfrak R$-rack  with $h=4$, $N_r=\{4,5\}, C_A=\{2,5\}$, and $C_B=\{4\}$ and its reduction.}
    \label{fig:rack}
\end{figure}

A class of groundings is a
\ndef{full class of groundings} if it is the class of the $\mathfrak G$-groundings of all bipartite graphs, for some $\mathfrak G$; Likewise, a class of racks is a \ndef{full class of racks} is it is the class of the $\mathfrak R$-racks of all ordered bipartite graphs, for some $\mathfrak R$.
Our aim is to prove that not only a full class of groundings or racks can be obtained as {\tim}s from
any hereditary independent class of TOWS graphs, but that this may be done in a definable and effective way
(after having selected a well-behaved subclass).

\begin{ndefi}
A  \ndef{definable \tim} is the composition of an interpretation~$\mathsf\Lambda$ performing a monadic expansion 
(defining $D$ and $V$) and the interpretation $\mathsf{Shrink}$.
Let $\mathscr C,\mathscr D$ be classes of tree-ordered graphs.   We say that $\mathscr D$ is a \ndef{definable \tim} of $\mathscr C$ if $\mathscr D=\mathsf I(\mathscr C)$, for some definable {\tim} $\mathsf I$.
\end{ndefi}

According to \zcref{thm:lics_twists}, 
for every independent hereditary class of TOWS graphs $\mathscr C$
there exists a core $\tau$, such that
$\{\tau[n]\colon n\in\mathbb N\}\subseteq \mathscr C$.
\medskip

Thus, we consider a fixed core $\tau$ with length $h$ and  define the mapping ${\rm Host}_\tau$ and the interpretation $\mathsf{\Lambda}_\tau$ such that, for every bipartite graph $G$, 
${\rm Host}_\tau(G)$ belongs to the hereditary closure of $\{\tau[n]\colon n\in\mathbb N\}$, is polynomial time computable from $G$, and $\mathsf{Shrink}(\mathsf{\Lambda}_\tau({\rm Host}_\tau(G)))$ is either a grounding or a rack with reduction $G$  and length $h'\leq h$.

The interpretation $\mathsf \Lambda_\tau$ is a monadic expansion defining unary predicates $V$ and~$D$, by formulas $\rho_V(x)$ and $\rho_D(x)$. 
We start by introducing the following definable unary predicates, that will be
used in the definition of $\rho_V$:
 \begin{itemize}
        \item $S$ is the set of \ndef{small} vertices, i.e.
        having at most 2 successors in $\prec$;
        \item $M$ is the set of the successors of $r$.
    \end{itemize}

We consider different cases, depending on $\tau$. In the following, we assume
that $\tau[n]$ has guards $A = \{a(1), \ldots, a(n)\}$ and $B = \{b(1),\ldots, b(n)\}$, is indexed by $I=J=[n]$, and
contains meshes $\mu_1, \ldots, \mu_h$.

\subsubsection{$\tau$ has type $1$}
In this case, $\mathsf{\Lambda}_\tau$ is defined by
\begin{align*}
\rho_D(x)&:=\bot,\\
\rho_V(x)&:=M(x),
\end{align*}
and, for a bipartite graph $G=(X,Y,F)$, 
with $X=\{x_i\colon i\in[n_X]\}$ and $Y=\{y_i\colon i\in [n_Y]\}$,
we define $n=\max(n_X,n_Y)$ and 
${\rm Host}_\tau(G)$ is the substructure of $\tau[n]$ induced by
\begin{itemize}
    \item the root $r$,
    \item the vertices $a(i)$ for $i\in [n_X]$,
    \item the vertices $\mu_s(i,j)$ for $\{x_i,y_j\}\in F$, $s\in [h]$,
    \item the vertices $b(j)$ for $j\in [n_Y]$.
\end{itemize}
\begin{lem}
\label{lem:tau1}
Let $\tau$ be a core of type $1$ and length $h$. 
Then  there exists $h'\le h$ and $N_r\subseteq\{0,\dots,h'+1\}$ such that 
for every bipartite graph $G=(X,Y,F)$, 
$\mathsf{Shrink}(\mathsf{\Lambda}_\tau({\rm Host}_\tau(G)))$ is 
the $(h',N_r)$-grounding of $G$.
\end{lem}
\begin{proof}
    Let $\mathfrak T=(A,\mu_1,\dots,\mu_h,B)$ be the twister of the twist $\tau[n]$ including ${\rm Host}_\tau(G)$ as an induced substructure.
    We denote by $D$ the domain of ${\rm Host}_\tau(G)$.

     Note that the vertices of $A$ are successors of $r$ (i.e. cover $r$), hence $A\cap D\subseteq M$ and, similarly, $B\cap D\subseteq M$.

    By \zcref{lem:cov_twist}, for each
    $s\in [h]$, vertices of $\mu_s$ are
    either in $M$ or cover a vertex in $M$.
    As each of $A\cap D$, $B\cap D$, and $V(\mu_1)\cap D,\dots,V(\mu_h)\cap D$ is homogeneous with~$\{r\}$, it follows
    by construction of ${\rm Host}_\tau(G)$ that  $\mathsf{Shrink}(\mathsf{\Lambda}_\tau({\rm Host}(G)))$ will result in the $(h',N_r)$-grounding of $G$ with  $h'\leq h$ and $N_r$ determined by~$\tau$.
\end{proof}

\subsubsection{$\tau$ has type $2$}
In this case also, $\mathsf{\Lambda}_\tau$ is defined by
\begin{align*}
\rho_D(x)&:=\bot,\\
\rho_V(x)&:=M(x).
\end{align*}
and, for a  bipartite graph $G$ with parts $X=\{x_i\colon i\in[n_A]\}$,  $Y=\{y_i\colon i\in [n_B]\}$, and edge set $F$,
we define $n=\max(n_X,n_Y)+1$ and 
${\rm Host}_\tau(G)$ is the substructure of $\tau[n]$ induced by
\begin{itemize}
    \item the root $r$,
    \item the vertices $\mu_1(i+1,1)$ for $i\in [n_X]$,
    \item the vertices $\mu_s(i+1,j+1)$ for $\{x_i,y_j\}\in F$ and $1\le s\le h$,
    \item the vertices $\mu_h(1,j+1)$ for $j\in [n_Y]$.
\end{itemize}

\begin{lem}
\label{lem:tau2}
Let $\tau$ be a core of type $2$ and length $h$. 
Then  there exists $h'\le h$ and $N_r\subseteq\{0,\dots,h'+1\}$ such that 
for every bipartite graph $G=(X,Y,F)$, 
$\mathsf{Shrink}(\mathsf{\Lambda}_\tau({\rm Host}_\tau(G)))$ is 
the $(h',N_r)$-grounding of $G$.
\end{lem}

\begin{proof}
Let $\mathfrak T = (\emptyset, \mu_1, \ldots, \mu_h, \emptyset)$ be the twist
of $\tau[n]$.
By definition of \zcref{typ2}, the
restriction of $\prec$ to $\mu_1$ and $\mu_h$
is a union of chains, and by \zcref{lem:twist2_mu1}, such that $\mu_1(i,j)\prec: \mu_1(i,j+1)$ (resp.
$\mu_h(i,j)\prec:\mu_h(i+1,j)$).
As only $\mu_1(i,1)$ and $\mu_h(1,j)$ are
in $M$ in those chains, and the others marks
correspond to the ones computed for type 1,
we can consider that the interpretation  
$\mathsf{Shrink}$ applied to $\mathsf{\Lambda}_\tau({\rm Host}_\tau(G))$ first contracts each chain of $\mu_1$ and $\mu_h$ to its minimal argument, yielding
${\rm Host}_{\tau'}(G)$ for some $\tau'$ of
type 1 and order $h-2$, then does the same contractions to
${\rm Host}_{\tau'}(G)$ as those defined by the monadic expansion
$\mathsf{\Lambda}_{\tau'}({\rm Host}_{\tau'}(G))$.
Hence, it follows from \zcref{lem:tau1} that $\mathsf{Shrink}(\mathsf{\Lambda}_\tau({\rm Host}_\tau(G)))$ is the $(h,N_r)$-grounding of $G$.
\end{proof}

\subsubsection{$\tau$ has type $3$}

We call a vertex $x$ regular if $x$ has a big predecessor and the sum of the number of successors of $x$ and the number of $E$-neighbors of $x$ distinct from $r$ is two. (Note that regular vertices are definable.)

We define an interpretation $\mathsf \Lambda_\tau^0$ by $\rho_D(x):=\bot$ and $\rho_V(x)$, where $\rho_V(x)$ is the formula asserting that either $x$ is regular, or $x$ has no $E$-neighbor distinct from $r$ and its predecessor
is either big or with at least one $E$-neighbor distinct from $r$.

Let $(G,<_X,<_Y)$ be a  ordered bipartite graph, with parts $X=\{x_i\colon i\in [n_X]\}$ ordered by $<_X$ and $Y=\{y_j\colon j\in [n_Y]\}$ ordered by $<_Y$, and edge set $F$. Define $n=\max(n_X,n_Y)+1$.
We assume that $G$ has minimum degree at least $2$.
${\rm Host}^0_\tau(G)$ is the substructure of $\tau[n]$ induced by
\begin{itemize}
    \item the root $r$,
    \item the vertices $\mu_1(i+1,1)$ for $i\in [n_X]$,
    \item the vertices $\mu_s(i+1,j+1)$ for $\{x_i,y_j\}\in F$ and $1\le s\le h$,
    \item the vertices $\mu_h(1,j+1)$ for $j\in [n_Y]$.
\end{itemize}
    
\begin{lem}
\label{lem:tau3}
Let $\tau$ be a core of type $3$ and length $h$. 
Then  there exists $h'\le h$, $N_r\subseteq\{0,\dots,h'+1\}$, $C_A\subseteq [h'+1]$, and $C_B\subseteq [h']$, such that $C_A\cap C_B=\emptyset$ and,
for every  ordered bipartite graph $(G,<_X,<_Y)$ with minimum degree at least $2$, 
$\mathsf{Shrink}(\mathsf{\Lambda}^0_\tau({\rm Host}^0_\tau(G)))$ is 
the $(h',N_r,C_A,C_B)$-rack of $G$.
\end{lem}
\begin{proof}
First note that a vertex in $V(\mu_s)$ (for $1<s<h$) is regular if and only if
its predecessor is neither in $V(\mu_{s-1})$ nor in $V(\mu_{s+1})$.

Assume $(\mu_1,\mu_2)$ is matching. The only vertices of $V(\mu_1)$ without $E$-neighbors distinct from $r$ are the vertices $\mu_1(i+1,1)$ for $i\in [n_X]$.
Moreover, the predecessors of these vertices are either $r$ (which is big), or a vertex of the form $\mu_1(i,j)$, which has an $E$-neighbor distinct from $r$.
Hence, the vertices of $V(\mu_1)$ satisfying $\rho_V$ are exactly the vertices of the form $\mu_1(i+1,1)$ with $i\in [n_X]$.

The case of $(\mu_{h-1},\mu_h)$ is similar, except that 
the predecessor of $\mu_h(1,2)$ can either be $r$ (which is big) or 
$\max V(\mu_1)$. In this last case, either  if $(\mu_1,\mu_2)$ is simply vertical and, as the minimum degree of $G$ is at least $2$, $V(\mu_1)$ is big, or if $(\mu_1,\mu_2)$ is matching and the vertices of some $V(\mu_s)$ have $\max V(\mu_1)$ as their predecessors and $V(\mu_1)$ is big, or $(\mu_1,\mu_2)$ is matching and $\max V(\mu_1)$ has $\min V(\mu_h)$ as unique successor. However, in this case $V(\mu_1)$ has an $E$-neighbor distinct from $r$. Hence, the vertices of $V(\mu_h)$ satisfying $\rho_V$ are exactly the vertices of the form $\mu_h(1,j+1)$ with $j\in [n_Y]$.

Assume $(\mu_1,\mu_2)$ is simply vertical. Then  no vertex of $V(\mu_1)$ has an $E$-neighbor distinct from $r$. The big vertices of $V(\mu_1)\cup\{r\}$ are the predecessors of the vertices $\mu_1(i+1,1)$ for $i\in [n_X]$ and $\max V(\mu_1)$. Hence, the vertices of $V(\mu_1)$ satisfying $\rho_V$ are exactly the vertices of the form $\mu_1(i+1,1)$ with $i\in [n_X]$. Similarly (with the same explanation for the special case of $\min V(\mu_h)$ as above), the vertices of $V(\mu_h)$ satisfying $\rho_V$ are exactly the vertices of the form $\mu_h(1,j+1)$ with $j\in [n_Y]$.

From all of these properties, we deduce that $\mathsf{Shrink}(\mathsf{\Lambda}_\tau({\rm Host}(G)))$ is the $(h',N_r,C_A,C_B)$-rack of $G$ with  $h'\leq h$, $N_r$, $C_A$, and $C_B$ determined by~$\tau$.
\end{proof}

\begin{figure}[ht]
    \begin{minipage}{.25\linewidth}
    \includegraphics[width=\linewidth]{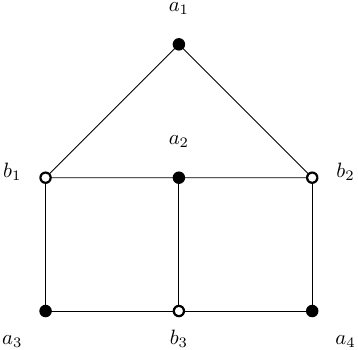}
    \end{minipage}
    \hfill
    \begin{minipage}{.7\linewidth}
    \includegraphics[width=\linewidth]{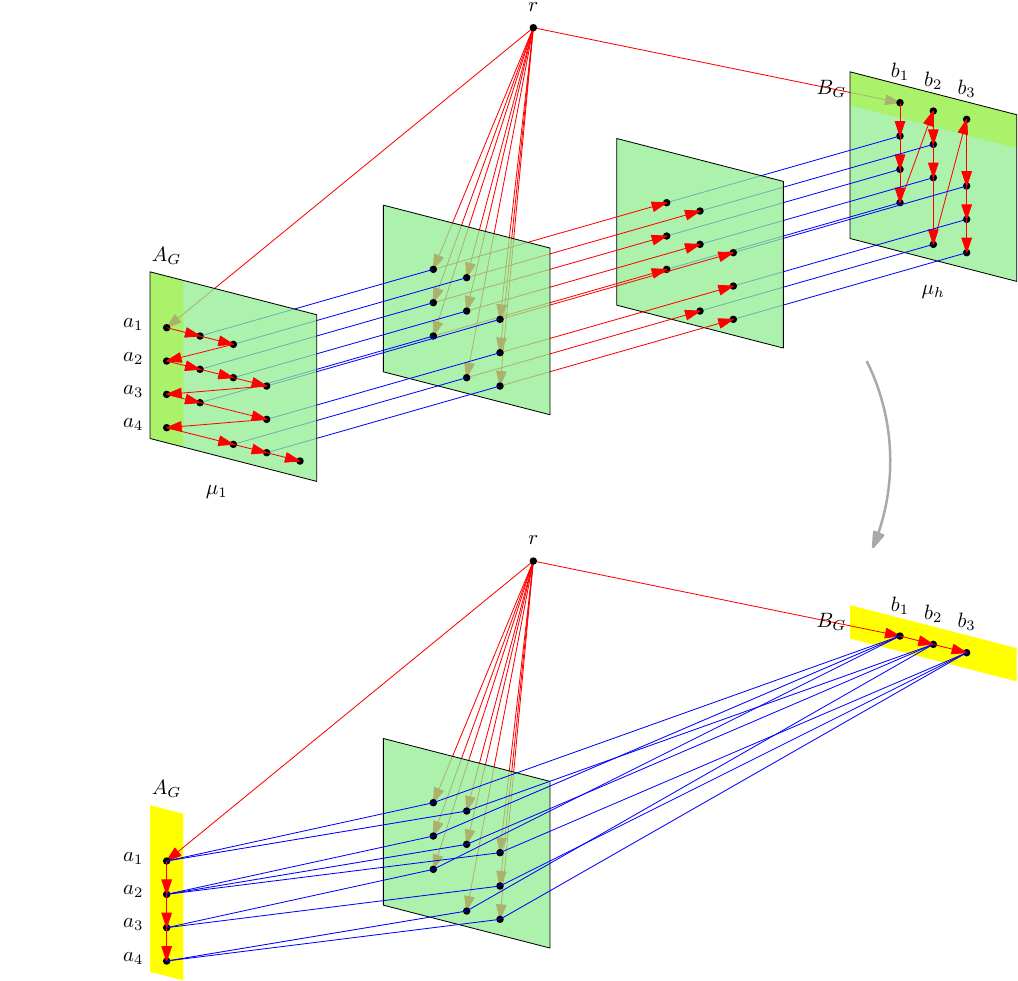}
    \end{minipage}
    \caption{An ordered bipartite graph $(G,<_A,<_B)$ with $a_1<_A a_2<_A a_3<_A a_4$ and $b_1<_B b_2<_B b_3<_B b_4$)  and its host ${\rm Host}^0(G)$, in the case where $\tau$ has type $3$, with both $(\mu_1,\mu_2)$ and $(\mu_{h-1},\mu_h)$ matchings of type $M^{\underline{\phantom{\rightarrow}}}$.}
    \label{fig:host3_0}
\end{figure}

\begin{thm}
\label{thm:racks}
    Let $\mathscr C$ be an independent hereditary class of TOWS graphs, and let $\tau$ be a core such that 
    $\tau[n]\in\mathscr C$ for all $n\in\mathbb N$.
    
    Then  there exists a polynomial time computable ${\rm Host}_\tau$ from graphs to TOWS graphs, and a definable {\tim} $\mathsf I_\tau$ such that for every graph~$G$, ${\rm Host}_\tau(G)\in \mathscr C$, and
\begin{itemize}
    \item if $\tau$ has type $1$ or $2$, then $\mathsf I_\tau({\rm Host}_\tau(G))$ is the $\mathfrak G$-grounding of $G$ (for some $\mathfrak G$ depending only on $\tau$);
    \item if $\tau$ has type $3$, then $\mathsf I_\tau({\rm Host}_\tau(G))$ is the $\mathfrak R$-rack of $G$ (for some $\mathfrak R$ depending only on $\tau$).
\end{itemize}
\end{thm}

\begin{proof}
    We reduce to twists of arbitrary large
    size with a fixed core $\tau$ by \zcref{thm:lics_twists}.

    For cores of type $1$ or $2$, this directly follows from \zcref{lem:tau1,lem:tau2}.
    
    Let $\tau$ be a core of type $3$ and let $(G,<_X,<_Y)$ be a  ordered bipartite graph. Let $\mathfrak T=(\mu_1,\dots,\mu_h)$ be the twister of  $\tau$.
    We define $(G^+,<_{X^+},<_{Y^+})$ from $(G,<_X,<_Y)$. 
    \begin{itemize}
        \item 
    The graph $G^+$ has parts $X^+=X\cup\{a_1,a_2\}$ and $Y^+=Y\cup\{b_1,b_2\}$, each~$a_i$ is adjacent to all the the vertices in $Y^+$, each $b_j$ is adjacent to all the the vertices in $X^+$, $G^+[X\cup Y]=G$.
    \item
    The order $<_{X^+}$ has restriction $<_Y$ on $X$, 
    $a_1<_{X^+} a_2$, and the missing cover is defined as follows:
    if $\max V(\mu_1)\prec\min V(\mu_h)$ in $\tau$, then
    $a_2<_{X^+}\min X$; otherwise, $\max X<_{X^+} a_1$.
\item
    The order $<_{Y^+}$ has restriction $<_Y$ on $Y$ and  $\max Y<_{Y^+}b_1<_{Y^+} b_2$.
    \end{itemize}
    
    The vertices $a_1,a_2,b_1,b_2$ are easily definable in $\mathsf{Shrink}(\mathsf{\Lambda}^0_\tau({\rm Hosts}^0_\tau(G^+)))$, hence definable in ${\rm Host}^0_\tau(G^+)$, and so is the set $Z$ of the vertices of ${\rm Host}^0_\tau(G^+)$ contracted to $a_1,a_2,b_1,b_2$ or to one of the subdivision vertices on a path linking one of these vertices to another vertex.
    The interpretation $\mathsf \Lambda_\tau$ is defined by the same formula~$\phi_V$ as $\mathsf\Lambda_\tau^0$, but with the formula $\phi_D(x)$ expressing that $x$ belongs to $Z$. 
    We further define ${\rm Host}_\tau(G)={\rm Host}_\tau^0(G^+)$.
    It is easily checked that
    $\mathsf{Shrink}(\mathsf{\Lambda}_\tau({\rm Hosts}_\tau(G)))$ is the $(h,N_r,C_A,C_B)$-rack of $G$.
\end{proof}

\begin{cor}\label{cor:racks}
    Let $\mathscr C$ be an independent hereditary class of TOWS graphs. Then  there exists a subclass $\mathscr C'\subseteq \mathscr C$ such that
    \begin{itemize}
        \item either the class of all $\mathfrak G$-groundings (for some $\mathfrak G$) is a definable {\tim} of $\mathscr C'$, 
        \item or the class of all $\mathfrak R$-racks (for some $\mathfrak R$) is a definable {\tim} of $\mathscr C'$.
    \end{itemize}
\end{cor}

\begin{cor}
\label{cor:interpret}
    Let $\mathscr C$ be an independent hereditary class of TOWS graphs.
    Then, there exists an integer $h$ such that the class of all bipartite graphs is an effective interpretation of a subclass of $\mathscr C$.
\end{cor}
\begin{proof}
    This follows from \zcref{thm:racks} by applying to the $\mathfrak G$-groundings or $\mathfrak R$-racks the interpretation computing the $E$-reduct and removing the root $r$.
\end{proof}

\subsection{Tree-ordered minors of TOWS graphs}

We now considered the class $\Minor(\mathscr C)$ of all tree-ordered minors of the structures in a class $\mathscr C$ of TOWS graphs and prove that this class can serve as a bridge between the notion of monadic dependence of classes of sparse graphs and the notion of nowhere-denseness of a class of graphs.

 \begin{thm}[store=thm:mainG]
Let $\mathscr C$ be a class of tree-ordered weakly sparse graphs. Then  the following are equivalent:
\begin{enumerate}
    \item\label{enum:G-mdependent} $\mathscr C$ is monadically dependent;
    \item\label{enum:G-IM-mdependent} $\Cont(\mathscr C)$ is monadically dependent;
    \item\label{enum:G-IM-ND} $\Cont(\mathscr C)^E$ is nowhere dense;
    \item\label{enum:G-M-mdependent} $\Minor(\mathscr C)$ is monadically dependent;
    \item\label{enum:G-M-ND} $\Minor(\mathscr C)^E$ is nowhere dense.
\end{enumerate}
\end{thm}
\begin{proof}
(1)$\iff$(2) follows from the fact that $\Cont(\mathscr C)$ is a transduction of $\mathscr C$ (\zcref{fact:cont_trans}) and $\mathscr C\subseteq \Cont(\mathscr C)$;
(1)$\Rightarrow$(3) follows from \zcref{lem:Min}, and the converse implication (3)$\Rightarrow$(1) follows from \zcref{thm:racks}.

As $\Minor(\mathscr C)=\Cont(\Mon(\mathscr C))$ (\zcref{fact:cont_min}), we deduce from $(1)\Leftrightarrow(2)\Leftrightarrow(3)$ the equivalence of 
\begin{enumerate}[label=\emph{(\arabic*')}]
    \item $\Mon(\mathscr C)$ is monadically dependent;
    \item $\Minor(\mathscr C)$ ($=\Cont(\Mon(\mathscr C))$) is monadically dependent (that is, (4));
    \item $\Minor(\mathscr C)^E$ ($=\Cont(\Mon(\mathscr C))^E$) is nowhere dense (that is, (5)).
\end{enumerate}
As $\Minor(\mathscr C)^E$ is also the monotone closure of $\Cont(\mathscr C)^E$, the conditions (3) and (5) are equivalent.
The theorem follows.
\end{proof}
\subsection{Fundamental graphs of TOWS graphs}
\label{sec:fund_dep}
In this section, we prove 
a relation between the monadic dependence of a class of TOWS graphs and the monadic dependence of the class of its generalized fundamental graph. This result will be one of the bricks of the proof of \zcref{thm:mainS}.
\begin{lem}
\label{lem:lambdaG}
Let $\mathscr C$ be a class  of TOWS graphs.

If the class $\Lambda(\mathscr C)$ is monadically dependent,
then $\mathscr C$ is monadically dependent.
\end{lem}
\begin{proof}
    Assume for contradiction that $\Lambda(\mathscr C)$ is monadically dependent but $\mathscr C$ is not monadically dependent. First note that $\Lambda(\Minor(\mathscr C))$ is nothing but the hereditary closure of $\Lambda(\mathscr C)$.
    Hence, a contradiction will arise if we prove that $\Lambda(\Minor(\mathscr C))$ is not monadically dependent.
    
    According to \zcref{thm:racks},
    there is a contraction $\mathscr D$ of a subclass of $\mathscr C$, which is either the class of all $\mathfrak R$-racks for some $\mathfrak R$ or the class of all $\mathfrak G$-groundings for some~$\mathfrak G$.

    Let $\mathfrak G=(h,N_r)$, and assume that 
    $\Cont(\mathscr C)$ contains all $\mathfrak G$-groundings.
    By deleting the edges incident to the root, we get that
    $\Minor(\mathscr C)$ contains all $(h,\emptyset)$-groundings.
    If $\mathbf M$ is the $(h,\emptyset)$-grounding of $K_{t,t}$, then $\Lambda(\mathbf M)$ is the $(2h+1)$-subdivision of $K_{t,t}$. Hence, $\Lambda(\Minor(\mathscr C))$ is not monadically dependent.

    Otherwise, there is some $\mathfrak R=(h,N_r,C_A,C_B)$ such that  $\Cont(\mathscr C)$ contains all $\mathfrak R$-racks, hence $\Minor(\mathscr C)$ contains all $(h,\emptyset,C_A,C_B)$-racks.
    Let $n\in\mathbb N$.  $\Minor(\mathscr C)$ contains 
    the $(h,\emptyset,C_A,C_B)$-rack $\mathbf N$ of the ordered $K_{n,n}$, with set of vertices $A=\{u_1,\dots,u_n\}$, $B=\{v_1,\dots,v_n\}$, and $L_s=\{x_s(i,j)\colon i,j\in[n]\}$. 
    Identifying the tree edge corresponding to the cover $(u,v)$ to the vertex $v$, let 
    \begin{align*}
    V_A&=A=\{u_1\dots,u_n\},\\
    V_B&=B=\{v_1\dots,v_n\},\\
    X_s&=L_s=\{x_s(i,j)\colon i,j\in [n]\}&\text{(for $s\in [h]$)},\\
    Y(A,X_1)&=\{w_0(i,j)\colon i,j\in [n]\},\\
    Y(X_s,X_{s+1})&=\{w_s(i,j)\colon i,j\in [n]\}&\text{(for $s\in [h-1]$)},\\
    Y(X_{h},B)&=\{w_{h}(i,j)\colon i,j\in [n]\}.
    \end{align*}

    Then   $\Lambda_0$ be the bipartite graph with parts $A\cup B\cup\bigcup_s X_s$ and $Y(A,X_1)\cup Y(X_{h'},B)\cup \bigcup_s Y(X_s,X_{s+1})$ and edges
\begin{align*}
&\{u_i',w_0(i,j)\}&\text{(for $1\le i'\le i\le n$ and $j\in [n]$)},\\ 
&\{w_s(i,j),x_s(i,j)\}&\text{(for $s\in[h-1]$ and $i,j\in [n]$)},\\ 
&\{w_s(i,j),x_{s+1}(i,j)\}&\text{(for $s\in[h-1]$ and $i,j\in [n]$)},\\ 
&\{w_{h}(i,j),v_{j'}\}&\text{(for $1\le j'\le j\le n$ and $i\in [n]$)}.
\end{align*}
Note that $\Lambda_0$ corresponds to the generalized fundamental graph of $\mathbf N$ if $C_A=C_B=\emptyset$. We now show how to deduce $\Lambda(\mathbf N)$ when this is not the case. Let 
\[V_B'=\begin{cases}
V_A\cup V_B&\text{if $h+1\in C_A$;}\\
V_B&\text{otherwise.}
\end{cases}
\]
Then  $\Lambda(\mathbf N)$ is isomorphic to the graph obtained from $\Lambda_0$ by
successively
\begin{itemize}
    \item flipping the pair $(V_A,Y(A,X_1))$ if $1\in C_A$,
    \item flipping the pair $(V_A,Y(X_s,X_{s+1})$ if $s\in C_A$,
    \item flipping the pair $(V_A,Y(X_s,X_{s+1})$ if $s+1\in C_A$,
    \item flipping the pair $(V_A,Y(X_{h},B))$ if $h\in C_A$, 
    \item flipping the pair $(V_B',Y(A,X_1))$ if $1\in C_B$,
    \item flipping the pair $(V_B',Y(X_s,X_{s+1})$ if $s\in C_B$,
    \item flipping the pair $(V_B',Y(X_s,X_{s+1})$ if $s+1\in C_B$,
    \item flipping the pair $(V_B',Y(X_{h},B))$ if $h\in C_B$. 
\end{itemize}

    As $\Lambda(\mathbf N)$ is an induced subgraph of $\Lambda(\mathbf M)$, it follows from the characterization of~\cite{dreier2024flipbreakability} that $\Lambda(\mathscr C)$ is not monadically dependent, contradicting our assumption.
\end{proof}
\begin{figure}[ht]
    \centering
    \includegraphics[width=0.5\linewidth]{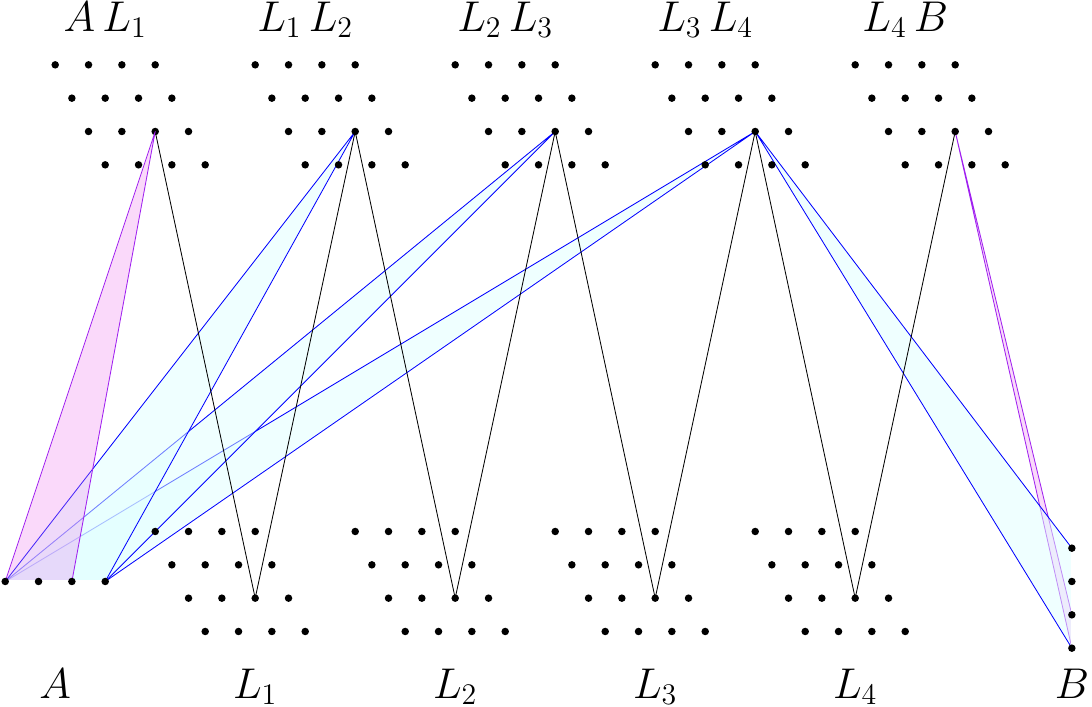}
    \caption{The fundamental graph of the $(4,\emptyset,\{2,5\},\{4\})$-rack of $K_{4,4}$.}
    \label{fig:placeholder}
\end{figure}

\newpage
\part{Monadic dependence of tree-ordered weakly sparse structures}
\addtocontents{toc}{\vspace{2pt}}
\section{Stability and dependence on weakly sparse classes of structures}
\label{sec:Gstdep}

We recall the Theorem of Adler and Adler characterizing monotone dependent classes of binary structures.

\getkeytheorem{thm:Adler}

This theorem has recently been  (almost fully) extended to relational structures by Braunfeld and Laskowski.
\begin{theorem}[\cite{braunfeld2022existential}]
	\label{thm:Adler2}
	Let $\sigma$ be a finite signature.
	For a monotone class $\mathscr M$ of finite $\sigma$-structures, the following are equivalent:
	\begin{enumerate}
		\item $\mathscr M$ is dependent;
		\item $\mathscr M$ is monadically dependent;
		\item $\mathscr M$ is stable;
		\item $\mathscr M$ is monadically stable.
	\end{enumerate}
\end{theorem}

It follows from \zcref{thm:indND} that, in the case of hereditary classes of graphs, one can replace the assumption that $\mathscr C$ is monotone by the assumption that it is weakly sparse. 
Furthermore, Braunfeld and Laskowski proved the following. 

\begin{theorem}[\cite{braunfeld2022existential}]
	\label{thm:LB}
	Let $\sigma$ be a finite signature.
	Let $\mathscr M$ be a hereditary class of finite $\sigma$-structures. Then 
	\begin{enumerate}
		\item $\mathscr M$ is dependent if and only if $\mathscr M$ is monadically dependent;
		\item $\mathscr M$ is stable if and only if $\mathscr M$ is monadically stable. 
	\end{enumerate}
\end{theorem}

\zcref{thm:LB}, together with the next remark, 
allows us to completely  generalize \zcref{thm:Adler2} to hereditary weakly sparse classes of relational structures.

\begin{remark}[{\cite[Proposition 5.7]{Sparsity}}]
\label{rem:GI}
     Let  $\mathscr M$ be a class of $\sigma$-structures. 
     Then,  $\Gaif(\mathscr M)$ is nowhere dense if and only if $\Inc(\mathscr M)$ is nowhere dense. 
\end{remark}

(Note that the original statement and proof of \cite[Proposition 5.7]{Sparsity} considered the multigraph version of the incidence graph,  but the same proof works for the simple graph version as well.)

\begin{thm}[store=cor:NDG]
			Let $\mathscr M$ be a hereditary weakly sparse class of $\sigma$-structures. Then, the following are equivalent:
			\begin{enumerate}
				\item $\mathscr M$ is monadically stable;
				\item $\mathscr M$ is  stable;
				\item $\mathscr M$ is monadically dependent;
				\item $\mathscr M$ is dependent;
				\item the monotone closure of $\mathscr M$ is monadically stable;
				\item $\mathscr M$ is nowhere dense (i.e.\ $\Gaif(\mathscr M)$ is nowhere dense);
				\item $\Inc(\mathscr M)$ is nowhere dense.
			\end{enumerate}    
\end{thm}
\begin{proof}
	We have $(1)\Leftrightarrow (2)$ and $(3)\Leftrightarrow (4)$ by \zcref{thm:LB} and $(2)\Rightarrow (4)$ by definition. 
    
    Now, $(3)$ implies that $\Gaif(\mathscr M)$ is monadically dependent, as it is a simple interpretation of $\mathscr M$ (see comment after~\zcref{lem:preserve-by-interpret}). Assume for contradiction that $\Gaif(\mathscr M)$ is monadically dependent but not nowhere dense. Then, as $\mathscr M$ is weakly sparse, according to \zcref{thm:indND} there exists a positive integer $t\geq 1$ such that for every integer  $n$, the $t$-subdivision of~$K_{n}$ is an induced subgraph of some $G\in\Gaif(\mathscr M)$.  
    This contradicts the hypothesis that $\Gaif(\mathscr M)$ is monadically dependent: from induced $t$-subdivisions of arbitrarily large cliques one can transduce all finite graphs by choosing the branch vertices and encoding edges via the subdivided paths.
    Hence, $(3)\Rightarrow (6)$. 
    By \zcref{rem:GI} $(6)\Leftrightarrow (7)$.

    Assume $(6)$ and let $\mathscr M'$ be the monotone closure of $\mathscr M$. 
    Then also $\Gaif(\mathscr M')$ (which includes $\Gaif(\mathscr M)$ and is included in the monotone closure of $\Gaif(\mathscr M)$) is nowhere dense.
    According to \zcref{rem:GI}, this implies that $\Inc(\mathscr M')$ is nowhere dense. 
    Applying the property a second time, we get that $\Inc(\Inc(\mathscr M'))$ is nowhere dense.
    Then also the monotone closure of $\Inc(\Inc(\Mm))$ is nowhere dense, which implies that it is monadically stable by \zcref{thm:Adler}. 
    As $\Mm$ and $\mathscr M'$ are a transduction of $\Inc(\Inc(\mathscr M'))$ we conclude that $\Mm$ and $\Mm'$ are monadically stable, hence $(6)\Rightarrow (1)$ and $(6)\Rightarrow (5)$. 

    Now assume $(5)$ and again denote by $\Mm'$ the monotone closure of $\Mm$. Then $\Gaif(\Mm')$, as a transduction of $\Mm'$, is monadically stable. 
    By \zcref{thm:Adler}, $\Gaif(\Mm')$ is nowhere dense, and hence $\Gaif(\Mm)$ is nowhere dense. 
    This means $\Mm$ is nowhere dense and hence $(5)\Rightarrow (6)$. 
\end{proof}

\section{Tree-ordered weakly sparse  
structures}
\label{sec:TOWS}
In this section, we introduce tree-ordered structures and derived notions that will be used all along this paper.
\begin{ndefi}[Tree-ordered $\sigma$-structure]
\label{def:TO}
Let $\sigma$ be a finite relational signature.
A \ndef{tree-ordered} $\sigma$-structure is a $\sigma\cup\{\prec\}$-structure, where $\prec$ is a tree-order. 
\end{ndefi}

Recall that the $\sigma$-reduct of a structure $\mathbf M$ is denoted by $\mathbf M^\sigma$ and that, for a class~$\mathscr M$, we define $\mathscr M^\sigma=\{\mathbf M^\sigma\colon \mathbf M\in\mathscr M\}$.

\begin{ndefi}[TOWS]
A class $\mathscr M$ of tree-ordered $\sigma$-structures is \ndef[class]{tree-ordered weakly sparse} (or \ndef[class]{TOWS}) if the class 
$\Gaif(\mathscr M^\sigma)$
is weakly sparse. By extension, when a bound of the maximal order of a balanced complete bipartite is assumed to be implicitly fixed, we shall speak about \ndef[$\sigma$-structure]{tree-ordered  weakly sparse} $\sigma$-structures (or \ndef[$\sigma$-structure]{TOWS} $\sigma$-structures).
\end{ndefi}

\subsection{Basic constructions}
In this section we introduce some variants of classical constructions (Gaifman graph, incidence graph) in the context of tree-ordered structures. The idea is to ignore the tree-order when performing these constructions.

\begin{ndefi}[Tree-ordered Gaifman graph]
\label{def:TGaif}
The \ndef[of a tree-ordered structure]{tree-ordered Gaifman graph} $\TGaif(\mathbf M)$ of a tree-ordered $\sigma$-structure $\mathbf M$ is the tree-ordered graph $\mathbf G$ with the same domain as $\mathbf M$, such that $\mathbf G^\prec=\mathbf M^\prec$ and 
$\mathbf G^\sigma=\Gaif(\mathbf M^\sigma)$. 
\end{ndefi}

\begin{ndefi}[Tree-ordered incidence graph]
\label{def:TInc}
The \ndef[of a tree-ordered structure]{tree-ordered incidence  graph} $\TInc(\mathbf M)$ of $\mathbf M$ is 
the tree-ordered graph, whose $E$-reduct is $\Inc(\mathbf M^\sigma)$, and whose 
$\prec$-reduct is the tree-order obtained from $\mathbf M^\prec$ by adding all the other elements in  $\Inc(\mathbf M^\sigma)$ as an antichain (comparable with the root only).
\end{ndefi}

\begin{lem}
    \label{lem:inc2tows}
    Let $\sigma$ be a finite relational signature. There exists a transduction~$\mathsf T$ such that, for every tree-ordered $\sigma$-structure $\mathbf M$, $\mathbf M\in\mathsf T(\TInc(\TInc(\mathbf M)))$.
\end{lem}

\begin{proof}
Let $k$ be the maximum arity of relations in $\sigma$.
    The transduction $\mathsf T$ uses unary predicates $V,M_1,\dots,M_k,P$,
  and  is defined by $\nu(x):=V(x)$, $\rho_\prec(x,y):=(x\prec y)$ and, for each $R\in\sigma$ with arity $r$, 
    \[\rho_R(\bar x):=
    \exists \bar y, t\ P(t)\wedge\bigwedge_{i=1}^r\bigl(
     V(x_i)\wedge E(x_i,y_i)\wedge M_i(y_i)\wedge E(y_i,t)\bigr).
    \]

    Let $\mathbf M$ be a tree-ordered $\sigma$-structure.
    The vertices of $\TInc(\TInc(\mathbf M))$ are marked as follows:  the domain $M$ of $\mathbf M$  is marked by $V$, the  pairs $(R,\bar v)$ with $R\in\sigma$ and $\bar v\in R(\mathbf M)$ are marked by $P$, and the  pairs
    $(u,(R,\bar v))$ such that
    $R\in\sigma, \bar v\in R(\mathbf M)$, and  $u=v_i$ are marked $M_i$.
    
    With this marking, it is easily checked that $\mathbf M\in\mathsf T(\TInc(\TInc(\mathbf{M})))$.
\end{proof}

\begin{remark}
    It is direct from the definitions that $\TGaif(\mathbf M)^E=\Gaif(\mathbf M^\sigma)$ and
    $\TInc(\mathbf M)^E=\Inc(\mathbf M^\sigma)$ (See \zcref{fig:example}).
\end{remark}

\begin{figure}[h!t]
    \centering
    \includegraphics[width=.9\columnwidth]{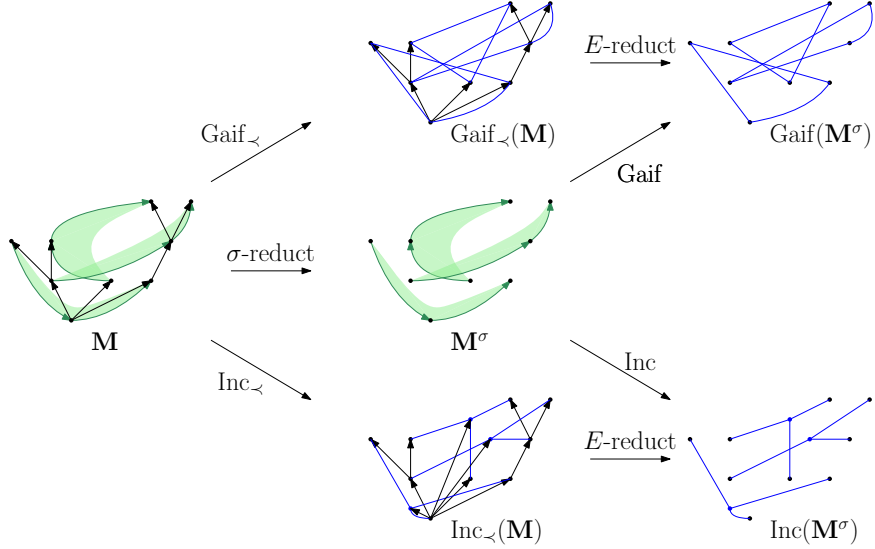}
    \caption{Constructions derived from a tree-ordered  structure $\mathbf M$. (The signature $\sigma$ contains a single $3$-ary relation.)}
    \label{fig:example}
\end{figure}. 

\subsection{The generalized fundamental graph}

 We generalize the notion of a fundamental graph of a tree-ordered graph given in \zcref{sec:fund} to relational structures.
Our generalization is inspired by the notion of hypergraphic polymatroid introduced in \cite{vertigan1993recognizing}. 
This leads to the following notion:

\begin{ndefi}
    The \ndef{flattening} $\Flat(\mathbf M)$ of a tree-ordered $\sigma$-structure  $\mathbf M$  is a pair $(\mathcal H,Y)$, where $Y$ is the edge set of the cover graph of $\mathbf M^\prec$, and $\mathcal H$ is the hypergraph with domain $M$, whose edge set is the union $Y$ and the symmetrization\footnotemark of $\mathbf M^\sigma$.
\end{ndefi}
\footnotetext{Precisely, $X$ is an edge in the symmetrization of $\mathbf M^\sigma$ if there exists at least one relation $R\in\sigma$, such that $X$ is  the support of some tuple in $R(\mathbf M)$.}

Note that if $\mathbf M$ is a tree-ordered graph and $(G,Y)=\Flat(\mathbf M)$, then $G$ is a connected graph having $Y$ as a spanning tree.

Consider a tree-ordered $\sigma$-structure $\mathbf M$ with flattening $(\mathcal H,Y)$. For $A\in E(\mathcal H)\setminus Y$, we consider the subspace $Y\langle A\rangle$ of $2^{Y}$ generated by the minimal sets $F\subseteq Y$ such that $F$ induces a path between two distinct elements of $A$. (Note that the dimension of this subspace is exactly the rank of $A$ in the hypergraphic polymatroid associated to $\mathcal H$.)

\begin{ndefi}[Generalized fundamental graph]
\label{def:TFun}
    Let $\sigma$ be a finite relational signature and let $k$ be the maximum arity of a relation in $\sigma$.

    A \ndef{generalized fundamental graph}  of a tree-ordered $\sigma$-structure $\mathbf M$ with flattening $(\mathcal H,Y)$ is a bipartite binary relational structure with signature $\{E_1,\dots,E_{k-1}\}$, whose domain is the union of $Y$ and $E(\mathcal H)$, where 
    the neighborhoods of $A\in E(\mathcal H)$ in $E_1,\dots,E_{k-1}$ spans $Y\langle A\rangle$.

    We denote by $\Lambda(\mathbf M)$ the set of all generalized fundamental graphs of $\mathbf M$.
\end{ndefi}

\begin{lem}
\label{lem:special_lambda}
Let $\sigma$ be a relational signature. There exists a transduction 
$\mathsf T_{\Lambda}$ such that, for every TOWS $\sigma$-structure $\mathbf M$, $\mathsf T_\Lambda(\TInc(\TInc(\mathbf M)))$ contains a generalized fundamental graph of $\mathbf M$.
\end{lem}
\begin{proof}
    Let $k$ be the maximum arity of relations in $\sigma$ and let $M_1,\dots,M_k,Y,Z$ be unary predicates not in $\sigma$.
    We define the transduction $\mathsf T_\Lambda$ by $\nu(x):=Y(x)\vee Z(x)$ and, for $i\in [k]$, 
    \begin{multline*}
\rho_{E_i}(x,y):=Y(x)\wedge Z(y)\wedge 
    \bigl(\exists v,v',w,w'\ M_1(w)\wedge M_{i+1}(w')\wedge E(w,y)\\
    \wedge E(w',y)\wedge E(w,v)\wedge E(w',v')\wedge (v\meet v' \prec x)\wedge ((x\preceq v)\vee(x\preceq v'))
    \bigr).
    \end{multline*}

    Let $\mathbf M$ be a TOWS $\sigma$-structure. 
    Two vertices $(R,\bar v)$ and $(R',\bar v')$ of $\TInc(\mathbf M)$ (hence of $\TInc(\TInc(\mathbf M))$) are said to be equivalent if
    $R=R'$ and $\bar v'$ is a permutation of $\bar v$.
    We mark vertices of $\TInc(\TInc(\mathbf M))$ as follows: the vertices from the domain~$M$ of $\mathbf M$ except the root are marked $Y$ and a representative of each equivalence class of 
    the vertices of the form $(R,\bar v)$ with $R\in\sigma$ and $\bar v\in R(\mathbf M)$ is marked $Z$. Moreover, vertices of the form $(u,(R,\bar v))$ with $u=v_i$ are marked $M_i$.  
    Each vertex $v\in Y$ corresponds to the unique edge of the cover graph of $\mathbf M^\prec$ with greater incidence~$v$ (this is the cover of the parent of $v$ by $v$). It is easily checked that 
    this marking defines a generalized fundamental graph of $\mathbf M$ (See \zcref{fig:fundamental}).
\end{proof}
\begin{lem}
\label{lem:trans_lambda}
    Let $\sigma$ be a relational signature. There exists a transduction 
$\mathsf T_{\rm ext}$ such that, for every TOWS $\sigma$-structure $\mathbf M$, if $\mathbf N$ is a generalized fundamental graph of $\mathbf M$, then $\mathsf T_{\rm ext}(\mathbf N)$ contains all the generalized fundamental graphs of $\mathbf M$.

    Consequently, $\Lambda(\mathbf M)\subseteq \mathsf T_{\rm ext}\circ\mathsf T_\Lambda(\TInc(\TInc(\mathbf M)))$, where $T_\Lambda$ is the transduction defined in the proof of \zcref{lem:special_lambda}.
\end{lem}
\begin{proof}
    Any generalized fundamental graph of $\mathbf M$ can be obtained by performing a change of basis for each vertex of $\mathbf N$ associated to a hyperedge of $\Flat(\mathbf M)$. This change of bases is easily performed by a transduction as follows (where $k$ is the maximum arity of a relation in $\sigma$) by letting
\[
\rho_{E_i}(x,y):=Y(x)\wedge Z(y)\wedge \bigoplus_{j=1}^k (M_{i,j}(y)\wedge E_j(x,y)),
\]
    where $\oplus$ denotes the exclusive or.
\end{proof}

\begin{figure}[h!t]
    \centering
    \includegraphics[width=\columnwidth]{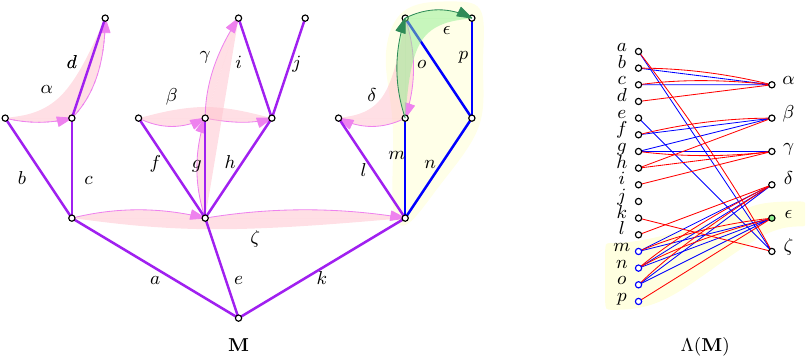}
    \caption{The fundamental graph of a tree-ordered $\sigma$-structure. The dotted lines indicate the order of the incidences in the triples; on $\Lambda(\mathbf M)$, the blue edges correspond to $E_1$ and the red edges to $E_2$.}
    \label{fig:fundamental}
\end{figure}

\begin{lem}
\label{fact:lambda_Gaif}
For every finite relational signature $\sigma$ there is a transduction $\mathsf T$ such that for every TOWS $\sigma$-structure $\mathbf M$ we have
$\Lambda(\TGaif(\mathbf M))\subseteq\mathsf T(\Lambda(\mathbf M))$. 
\end{lem}
\begin{proof}
Let $\sigma$ be a  relational signature and let $k$ be the maximum arity of a relation in $\sigma$.
We define the transduction $\mathsf T$ as the composition of a copy operation $\mathsf C_{\binom{k}{2}+1}$  and the transduction $\mathsf T'$ defined by $\nu(x):=Y(x)\vee Z(x)$ and
\begin{multline*}
\rho_E(x,y):=Y(x)\wedge Z(y)\wedge \\
\exists y'\ \bigl(A(y')\wedge
E(y,y')\wedge\bigvee_{(i,j)\in\binom{k}{2}} (M_{i,j}(y)\wedge(E_{i-1}(x,y')\oplus E_{j-1}(x,y')))\bigr),    
\end{multline*}
where we let $E_0$ to be the empty relation.

    Let $\mathbf M$ be a TOWS $\sigma$-structure and let $\mathbf N$ be a generalized fundamental graph of~$\mathbf M$ constructed as in the proof of \zcref{lem:special_lambda}.

    We mark the copies of $\mathbf N$ in $\mathsf C_{\binom{k}{2}+1}(\mathbf N)$ by predicates $A$ and $M_{i,j}$ for $1\leq i<j\leq k$.
    We further mark by $Y$ the part of  the copy of $\mathbf N$ marked $A$ corresponding to cover edges of $\mathbf M^\prec$. For each edge $uu'$ of $\Gaif(\mathbf M)$ we select a pair $(R,\bar v)$
    associated to a vertex of $\mathbf N$ such that $u,u'\in\bar v$, and we mark $Z$ the clone of this vertex in the copy of $\mathbf N$ marked $M_{\min(i,j),\max(i,j)}$. We denote by $\phi(uu')$ this vertex.
    Then, it is easily checked that $\rho_E(x,\phi(uu'))$ holds if and only if the unique cover of $\mathbf M^\prec$ with greater element $x$ belongs to the fundamental cycle of $uu'$.
\end{proof}
\subsection{{\Tim}s and tree-ordered minors}

{\Tim}s are defined the same way as for  tree-ordered graphs (see \zcref{sec:prelim}):
Let $\mathbf M$ be a tree-ordered structure.
An \ndef{elementary $\prec$-contraction} of a tree-ordered $\sigma$-structure $\mathbf M$ is obtained by identifying two elements $u$ and $v$ of $\mathbf M$, where $(u,v)$ is a cover in $\mathbf M^\prec$. A \ndef{$\prec$-contraction} is a sequence of elementary $\prec$-contractions. 

\begin{ndefi}
A \ndef[of a tree-ordered structure]{\tim} of $\mathbf M$ is a $\prec$-contraction of an induced substructure of $\mathbf M$. 
We denote by  $\Cont(\mathbf M)$ the set of all the {\tim}s of $\mathbf M$.
\end{ndefi}

An \ndef{elementary $\sigma$-deletion} of a tree-ordered $\sigma$-structure $\mathbf M$  is obtained by deleting a tuple of a relation (that is, a tuple in $R(\mathbf M)$ for some $R\in\sigma$). A \ndef{$\sigma$-deletion} is a sequence of elementary $\sigma$-deletions.

\begin{ndefi}
\label{def:Tminor}
A \ndef[of a tree-ordered structure]{tree-ordered minor} of a tree-ordered $\sigma$-structure~$\mathbf M$ is a tree-ordered $\sigma$-structure obtained from $\mathbf M$ by a sequence of  elementary $\prec$-contractions and elementary $\sigma$-deletions. For a class $\mathscr M$ of tree-ordered $\sigma$-structures, we denote by $\Minor(\mathscr M)$ the class of all the tree-ordered minors of structures in~$\mathscr M$.
\end{ndefi}

\section{{\Tim}s and tree-ordered minors  of TOWS structures}
\label{sec:timS}
In this Section we generalize the notions of {\tim} and tree-ordered minors introduced in \zcref{sec:minor_G}.
Let $\mathbf M$ be a tree-ordered structure.
\begin{ndefi}
An \ndef{elementary $\prec$-contraction} of a tree-ordered $\sigma$-structure $\mathbf M$ is obtained by identifying two elements $u$ and $v$ of $\mathbf M$, where $(u,v)$ is a cover in $\mathbf M^\prec$. An \ndef{elementary deletion} is the deletion of a non-root element (as well as all the relations in which the element is involved). Note that both elementary $\prec$-contractions and elementary deletions preserve the property of a structure to be tree-ordered.
An \ndef{\tim} of a TOWS structure $\mathbf M$ is a TOWS structure obtained from $\mathbf M$ by a sequence of deletions followed by a sequence of elementary $\prec$-contractions.
We denote by $\Cont(\mathbf M)$ the set of all the {\tim}s of $\mathbf M$. 
\end{ndefi}

It is easily checked that an {\tim} of an {\tim} of a tree-ordered $\sigma$-structure $\mathbf M$ is an {\tim} of $\mathbf M$.

An \ndef{elementary $\sigma$-deletion} of a tree-ordered $\sigma$-structure $\mathbf M$  is the deletion of a tuple of a relation (that is, a tuple in $R(\mathbf M)$ for some $R\in\sigma$). A \ndef{$\sigma$-deletion} is a sequence of elementary $\sigma$-deletions. 
We denote by $\Mon(\mathscr M)$ the set of all the structures obtained by $\sigma$-deletions on structures in $\mathscr M$. 
A \ndef{tree-ordered minor} of a tree-ordered $\sigma$-structure~$\mathbf M$ is a tree-ordered $\sigma$-structure obtained from $\mathbf M$ by a sequence of  elementary $\prec$-contractions, elementary deletions, and elementary $\sigma$-deletions. We denote by $\Minor(\mathbf M)$ the set of all the tree-ordered minors of $\mathbf M$.
The next fact is direct from the definitions.
\begin{fact}
    \label{fact:cont_minG}
    Let $\mathscr M$ be a class of tree-ordered $\sigma$-structures. 
    Then 
\[
\Minor(\mathscr M)=\Mon(\Cont(\mathscr M))=\Cont(\Mon(\mathscr M)).
\]
\end{fact}

\begin{ndefi}[$\mathsf{Shrink}$]
\label{def:shrink}
We define the interpretation $\mathsf{Shrink}$ \index{Shrink (of a structure)} of tree-ordered $\sigma$-structures in tree-ordered $\sigma\cup\{D,V\}$-structures, where $D$ and $V$ are unary relations, by (for each $R\in\sigma$ with arity $r$):

\begin{align*}    
    \nu(x)&:=V(x)\wedge\neg D(x)\vee (x=r)\\
    \rho_R(\bar x)&:=\exists \bar y\ \biggl(R(\bar y)\wedge\bigwedge_{i=1}^r (x_i\preceq y_i)\wedge\bigwedge_{i=1}^r\bigl(\neg D(y_i)\vee (y_i=r)\bigr)\\
    &\phantom{:=\exists \bar y\ \biggl(\bigwedge_{i=1}^r (x_i\preceq y_i)\ }
    \wedge\bigwedge_{i=1}^r \bigl(\forall z\ (x_i\prec z\preceq y_i)\rightarrow\neg \nu(z)\bigr)\biggr)
\end{align*}
\end{ndefi}

We stress the following easy fact, which justifies to consider that our generalization of $\Cont$ is still a transduction.
\begin{fact}
\label{fact:cont_transG}
    Let $\sigma$ be a relational signature. Then  the  transduction $\mathsf T$ defined as $\mathsf T=\mathsf{Shrink}\circ \Lambda$ (where $\Lambda$ is the monadic expansion adding predicates $D$ and $V$) is
    such that, for every class $\mathscr C$ of tree-ordered $\sigma$-structures we have
    $\Cont(\mathscr C)=\mathsf T(\mathscr C)$.
\end{fact}
\begin{proof}
    Let $\mathbf M$ be a tree-ordered $\sigma$-structure.
    For every expansion $\mathbf M^+$ of $\mathbf M$ by unary relations $V$ and $D$, $\mathsf{Shrink}(\mathbf M^+)\in\Cont(\mathbf M)$.
    Conversely, for every $\mathbf N\in \Cont(\mathbf M)$, by marking $D$ the  deleted vertices and $V$ the minimum vertex of each contracted part, we check that $\mathbf N\in\mathsf T(\mathbf M)$.
\end{proof}

\section{Monadic dependence of classes of TOWS structures}
\label{sec:depS}

In this section, we aim to prove \zcref{thm:mainS}, which we restate now for convenience.
\smallskip

\begin{thm}[store*=thm:mainS,restate-keys={note=see \zcpageref[nocap]{sec:depS}}]
 For a class $\mathscr M$ of TOWS $\sigma$-structures, the following are equivalent:
\begin{enumerate}
    \item $\mathscr M$ is monadically dependent,
    \item $\TGaif(\mathscr M)$ is monadically dependent,
    \item $\TInc(\mathscr M)$ is monadically dependent,
    \item $\Lambda(\mathscr M)$ is monadically dependent,
    \item $\Minor(\mathscr M)$ is monadically dependent, 
    \item $\Minor(\mathscr M)^\sigma$ is nowhere dense. 
\end{enumerate}
\end{thm}
\medskip

\begin{lem}
\label{lem:minG}
    Let $\mathbf N$ be a tree-ordered minor of a tree-ordered $\sigma$-structure $\mathbf M$. Then  $\TGaif(\mathbf N)$ is a tree-ordered minor of  $\TGaif( \mathbf M)$.
\end{lem}
\begin{proof}
    It is sufficient to consider the case where $\mathbf N$ is obtained from $\mathbf M$ by an elementary $\prec$-contraction or an elementary $\sigma$-deletion.

    In the first case, assume that $u\prec: v$ is a cover of $\mathbf M^\prec$. Then  $\TGaif(\mathbf N)$ is clearly obtained by identifying $u$ and $v$. Hence, $\TGaif(\mathbf N)$ is a tree-ordered minor of  $\TGaif(\mathbf M)$.

    In the second case, assume that $\mathbf N$ is obtained from $\mathbf M$ by deleting the tuple $\bar v\in R(\mathbf M)$. Then  $\TGaif(\mathbf N)$ is obtained from $\TGaif(\mathbf M)$ by deleting the pairs $\{v_i,v_j\}$ that do not appear in any other tuple in a relation $R'(\mathbf M)$ with \mbox{$R'\in\sigma$}. In particular the graph $\bigl(\TGaif(\mathbf N)\bigr)^\sigma$ is a subgraph of 
    $\bigl(\TGaif(\mathbf M)\bigr)^\sigma$, while 
    $\bigl(\TGaif(\mathbf N)\bigr)^\prec=\bigl(\TGaif(\mathbf M)\bigr)^\prec$. It follows that $\TGaif(\mathbf N)$ is a tree-ordered minor of  $\TGaif(\mathbf M)$.
\end{proof}

\begin{proof}[Proof of \zcref{thm:mainS}]
We aim to prove the next equivalences.
\begin{mycases}
    \item\label{enum:S-mdependent} $\mathscr M$ is monadically dependent;

    \item\label{enum:S-M-mdependent} $\Minor(\mathscr M)$ is  monadically dependent;

     \item\label{enum:S-TG-mdependent} $\TGaif(\mathscr M)$ is  monadically dependent;
     
    \item\label{enum:S-MG-ND} $\Minor(\TGaif(\mathscr M))^E$ is nowhere dense;
    \item\label{enum:S-GM-ND} $\Gaif(\Minor(\mathscr M)^\sigma)$ is nowhere dense;
    \item\label{enum:S-M-ND} $\Minor(\mathscr M)^\sigma$ is nowhere dense;
    \item\label{enum:S-IM-ND} $\Inc(\Minor(\mathscr M)^\sigma)$ is nowhere dense;
    \item\label{enum:S-MI-ND} $\Minor(\TInc(\mathscr M))^E$  is nowhere dense;
    \item\label{enum:S-I-mD} $\TInc(\mathscr M)$ is  monadically dependent;
    \item\label{enum:S-II-mD} $\TInc(\TInc(\mathscr M))$ is  monadically dependent;
    \item\label{enum:S-Fun} $\Lambda(\mathscr M)$  is  monadically dependent;
    \item\label{enum:S-FunG} $\Lambda(\TGaif(\mathscr M))$  is  monadically dependent;
\end{mycases}

We prove the following implications and equivalences.
\[
\begin{xy}
    \xymatrix@=16pt{
       \zcref[noname]{enum:S-Fun}\ar@2[d]&\zcref[noname]{enum:S-II-mD}\ar@2[l]\ar@2[d]&\zcref[noname]{enum:S-I-mD}\ar@2[l]&\zcref[noname]{enum:S-MI-ND}\ar@2[l]\\
    \zcref[noname]{enum:S-FunG}\ar@2[d]&\zcref[noname]{enum:S-mdependent}\ar@2[dl]\ar@2{<->}[r]&\zcref[noname]{enum:S-M-mdependent}&\zcref[noname]{enum:S-IM-ND}\ar@2{<->}[u]\\
        \zcref[noname]{enum:S-TG-mdependent}\ar@2{<->}[r]&\zcref[noname]{enum:S-MG-ND}\ar@2[r]&\zcref[noname]{enum:S-GM-ND}\ar@2{<->}[r]\ar@2{<->}[ur]&\zcref[noname]{enum:S-M-ND}
    }
\end{xy}
\]

\begin{trivlist}
    \item[\impl{1}{3}:] $\TGaif(\mathscr M)$ is a simple interpretation of $\mathscr M$.
    \item[\equ{3}{4}:] this is a direct consequence of \zcref{thm:mainG} applied to the weakly sparse class $\TGaif(\mathscr M)$ of tree-ordered graphs.
    \item[\impl{4}{5}:] this follows from \zcref{lem:minG}, as it implies
    \[
\Gaif(\Minor(\mathscr M)^\sigma)= \TGaif(\Minor(\mathscr M))^E\subseteq \Minor(\TGaif(\mathscr M))^E.
\] 
 \item[\equ{5}{6}:] this is the definition of nowhere-denseness for relational structures.
 \item[\equ{5}{7}:] this is well known (see, for instance \cite[Proposition 5.7]{Sparsity}).
  \item[\equ{7}{8}:] this follows from the fact that 
  $\Minor(\TInc(\mathscr M))^E$
  is the monotone closure of 
  $\Inc(\Minor(\mathscr M)^\sigma)$.
  \item[\impl{8}{9}:] this follows from \zcref{thm:mainG} applied to the weakly sparse class $\TInc(\mathscr M)$ of tree-ordered graphs.
  \item The above proves the implication \impl{1}{9}. Applying this implication to $\TInc(\mathscr M)$, we deduce  \impl{9}{10}.
\item[\impl{10}{1}:] $\mathscr M$ is a transduction of 
  $\TInc(\TInc(\mathscr M))$ (\zcref{lem:inc2tows}).
     \item[\impl{10}{11}:] $\Lambda(\mathscr M)$ is a transduction of 
  $\TInc(\TInc(\mathscr M))$ (by \zcref{lem:trans_lambda}).
  \item[\impl{11}{12}:]  $\Lambda(\TGaif(\mathscr M))$ is a transduction of $\Lambda(\mathscr M)$ (\zcref{fact:lambda_Gaif}).
  \item[\impl{12}{3}:]
  this is a direct application of \zcref{lem:lambdaG}, since $\TGaif(\mathscr M)$ is a tree-ordered weakly sparse class of graphs.
  \item[\equ{1}{2}:] As $\mathscr M\subseteq \Cont(\mathscr M)$ and
  as $\Cont(\mathscr M))$ is a transduction of $\mathscr M$ (\zcref{fact:cont_trans}), we get that $\mathscr M$ and $\Cont(\mathscr M)$ are transduction-equivalent. Hence,
  (1) is equivalent to the property that $\Cont(\mathscr M)$ is monadically dependent.
  By the above, we deduce that $\Cont(\mathscr M)$ is monadically dependent if and only if $\Minor(\mathscr M)^\sigma$ is nowhere dense.
  As $\Minor(\Mon(\mathscr M))=\Minor(\mathscr M)$, $\Minor(\mathscr M)^\sigma$ is nowhere dense is equivalent to $\Cont(\Mon(\mathscr M))$ monadically dependent. As  $\Minor(\mathscr M)=\Cont(\Mon(\mathscr M))$, we deduce that (1) is equivalent to (2).\qedhere
\end{trivlist}
\end{proof}

\newpage
\part{Applications}
\addtocontents{toc}{\vspace{2pt}}

\section{Sparsification}
\label{sec:sp}

In this section, we make use of \MSO-transductions and \MSO-interpretations, which are defined similarly as
the other transductions and interpretations in this article (which are \FO-transductions and \FO-interpretations), except their
formulas are also allowed to quantify on unary predicates (for more on the use
of monadic second-order logic in graphs, we refer to \cite{courcelle12}).

We now consider yet another transduction, which is a variation of the reduct of tree-minors.
The transduction $\mathsf{Sp}$ is obtained by first adding to the edge set the cover graph of the tree-order, then applying the {\tim} transduction $\Cont$, then taking the $E$-reduct.

\begin{lem}
\label{lem:back_from_Sp}
    Let $\mathscr C$ be a class of TOWS graphs such that
    $\mathsf{Sp}(\mathscr C)$ has bounded expansion.

    Then $\mathscr C$ is an \MSO-transduction of $\mathsf{Sp}(\mathscr C)$.
\end{lem}
\begin{proof}
    Let $\sigma=(R,E,Y)$, where $R$ is a unary relation and $E$ and~$Y$ are two binary relations. As $\mathscr C$ has bounded expansion, there exists a transduction $\mathsf T$ such that for every  $G\in\mathscr C$, $\mathsf T(G)$ is the set of all the sigma structures having $G$ as their Gaifman graph \cite{TWW_perm}.
    In particular, for every TOWS graph $\mathbf M$, 
    $\mathsf T({\rm Sp}(\mathbf M))$ contains the $\sigma$-structure $\mathbf N$ with same domain as $\mathbf M$, where $R(\mathbf N)$ is the root of $\mathbf M$, $(x,y)\in Y(\mathbf N)$ if $\mathbf M\models x\prec:y$, and $E(\mathbf N)=E(\mathbf M)$.
    As the tree-order $\prec$ is an \MSO-interpretation of its its directed cover graph, we deduce (by composition) that $\mathscr C$ is an \MSO-transduction of $\mathsf{Sp}(\mathscr C)$.
\end{proof}

\begin{lem}
\label{lem:analog}
    Let $\mathscr O$ be a class of structures, let $\Pi$ be the set of all the classes that are \MSO-transductions of $\mathscr O$ and let $\Pi_S$ be the set of all the weakly sparse classes in $\Pi$.  
    
    Assume that all the classes in $\Pi_S$ have bounded expansion.
    Then for every class~$\mathscr C$ of TOWS-graphs, the following are equivalent:
    \begin{enumerate}
        \item $\mathscr C\in\Pi$;
        \item every weakly sparse transduction of $\mathscr C$ is in $\Pi_S$;
        \item $\mathsf{Sp}(\mathscr C)\in\Pi_S$.
    \end{enumerate}
\end{lem}
\begin{proof}
    (1)$\Rightarrow$(2) is direct from the definition of $\Pi_s$ (as every transduction is an \MSO-transduction). (2)$\Rightarrow$(3) as $(2)$ implies that~$\mathscr C$ is monadically dependent, hence by \zcref{thm:mainG}, $\mathsf{Sp}(\mathscr C)$ is nowhere dense (thus weakly sparse), hence (by (2)) $\mathsf{Sp}(\mathscr C)\in \Pi_s$.
    (3)$\Rightarrow$(2) as $\mathscr C$ is an \MSO-transduction of $\mathsf{Sp}(\mathscr C)$ if $\mathsf{Sp}(\mathscr C)$ has bounded expansion (\zcref{lem:back_from_Sp}). 
\end{proof}

\begin{thm}[store*=thm:sp,restate-keys={note=see \zcpageref[nocap]{sec:sp}}]
    Let $\mathscr C$ be a  class of TOWS graphs. 
    Then,
    $$\mathscr C\text{ is monadically dependent }\quad\iff\quad\mathsf{Sp}(\mathscr C)\text{ is nowhere dense}.$$
    
    Moreover, we have the following equivalences:
\begin{align*}
    {\rm cw}(\mathscr C)<\infty&\iff\forall\text{ weakly sparse transduction $\mathscr D$ of $\mathscr C$},\ {\rm tw}(\mathscr D)<\infty\\
    &\iff {\rm tw}(\mathsf{Sp}(\mathscr C))<\infty;\\
    {\rm lcw}(\mathscr C)<\infty&\iff\forall\text{ weakly sparse transduction $\mathscr D$ of $\mathscr C$},\ {\rm pw}(\mathscr D)<\infty\\
    &\iff {\rm pw}(\mathsf{Sp}(\mathscr C))<\infty;\\
    {\rm sd}(\mathscr C)<\infty&\iff\forall\text{ weakly sparse transduction $\mathscr D$ of $\mathscr C$},\ {\rm td}(\mathscr D)<\infty\\
    &\iff {\rm td}(\mathsf{Sp}(\mathscr C))<\infty.
\end{align*}
\end{thm}
\begin{proof}
The first equivalence is a direct consequence of \zcref{thm:mainG}.
The other equivalences follow from \zcref{lem:analog}, as
\begin{itemize}
    \item a class has bounded cliquewidth if and only if it is an \MSO-transduction of the class of all trees \cite{courcelle1992monadic} and the weakly sparse classes with bounded cliquewidth are exactly the classes with bounded treewidth \cite{Gurski2000};
    \item a class has bounded linear cliquewidth if and only if it is an \MSO-transduction of the class of all paths \cite{courcelle1992monadic} and the weakly sparse classes with bounded linear cliquewidth are exactly the classes with bounded pathwidth;
    \item a class has shrubdepth at most $h$ if and only if it is an \MSO-transduction of the class of all trees with depth at most $h$, and the weakly sparse classes with bounded shrubdepth are exactly the classes with bounded treedepth. \qedhere
\end{itemize}
\end{proof}

So, within TOWS-graphs, the dense analogue of bounded treewidth and bounded pathwidth are  bounded clique-width and bounded linear-cliquewidth.
While \zcref{lem:analog} cannot be used beyond cliquewidth,
we conjecture that dense analogues within TOWS graphs are the same as within general graphs. 
\begin{conj}
Let $\Pi_S$ be a downset of weakly sparse class properties (meaning that if $\mathscr C\in\Pi_S$ and $\mathscr D$ is a weakly sparse transduction of $\mathscr C$, then $\mathscr D\in\Pi_S$).
Let $\Pi$ be the dense analog of $\Pi_S$, that is, $\mathscr C\in\Pi$ if every weakly sparse transduction of $\mathscr C$ belong to $\Pi_S$.

Then  for every class $\mathscr C$ of TOWS-graphs, we have
\[
\mathscr C\in \Pi\qquad\iff\qquad\mathsf{Sp}(\mathscr C)\in \Pi_S.
\]
\end{conj}

\section{Intractability of {\FO} model checking in TOWS graphs}
\label{sec:model checking}

As an application of our characterization of independent TOWS graphs by substructure
obstructions (\zcref{thm:dependentTOWS}), we show that {\FO} model checking is intractable
for every independent hereditary class of TOWS graphs, under standard complexity-theoretic
assumptions.

An algorithm for {\FO} model checking on a specific graph class~$\mathscr C$ takes as input
a graph $G\in \mathscr C$ and a first-order sentence $\varphi$, and answers whether $G\models 
\varphi$ or not. Since a general algorithm cannot exist for all graphs or sentences under
widely-believed complexity assumptions, we ask for an {\em fixed-parameter-tractable} algorithm 
(or {\FPT} for short),
that is with time
complexity $O_{\mathscr C, \varphi}(|G|^c)$ for some~$c>0$, where the constants hidden by the
big-$O$ notation can depend on both the class and the sentence.

In the weakly sparse setting, the next theorem is a common generalization of the intractability of \FO-model checking for hereditary classes of
ordered graphs of unbounded twin-width\footnote{proved without the weak sparsity assumption} \cite[Theorem 5]{tww4} and of monotone somewhere dense graphs~\cite{dvovrak2010deciding}.

\begin{thm}[store*=thm:intractability,,restate-keys={note=see \zcpageref[nocap]{sec:model checking}}]
    Let $\mathscr C$ a hereditary class of tree-ordered weakly-sparse graphs. If~$\mathscr C$
    is independent, then {\FO}-model checking is not {\FPT}
 on $\mathscr C$,
    under the assumption that $\AW[*]\neq\FPT$.
\end{thm}
\begin{proof}
If $\mathsf I$ is an interpretation of $\sigma'$-structures in $\sigma$-structures and $\varphi$ is a sentence (in the language of $\sigma'$-structures), then there exists a (computable) formula $\mathsf I^*(\varphi)$ (in the language of $\sigma$-structures), such that for every $\sigma$-structure $\mathbf M$ it holds that  $\mathsf I(\mathbf M)\models\varphi$
if and only if $\mathbf M\models\mathsf I^*(\varphi)$.

According to \zcref{cor:interpret}, for  every independent hereditary class of TOWS graphs $\mathscr C$, there exists an integer $h$ such that $\mathscr C$
efficiently interprets the class of all $h$-subdivided bipartite graphs. As  \FO-model checking is $\AW[*]$-complete on these graphs \cite{downey1996parameterized}, we deduce that \FO-model checking is not {\FPT} on any independent hereditary classes of TOWS graphs (assuming
$\FPT\neq\AW[*]$).
\end{proof}

\section{A model-theoretical characterization of classes excluding a minor}
\label{sec:minor}
As a consequence of the above, we propose the following model theoretic characterization of weakly sparse classes of graphs excluding a minor.

We first take time for some definitions and notations.
Let $G$ be a graph. A \ndef{tree-ordering} of $G$ is the expansion of $G$ by a forest-ordering of $G$ defined by a rooted spanning forest of $G$. The \ndef{tree-ordering expansion} $\mathscr C_\Upsilon$ of a class of graphs $\mathscr C$ is the class of all the tree-orderings of all the graphs in $\mathscr C$.

\begin{thm}[store*=thm:excl_minor,,restate-keys={note=see \zcpageref[nocap]{sec:minor}}]
    Let $\mathscr C$ be a weakly sparse class of graphs. Then  the following are equivalent:
    \begin{enumerate}
        \item $\mathscr C$ excludes some graph as a minor;
        \item $\mathscr C_\Upsilon$ is monadically dependent.
    \end{enumerate}
    \end{thm}

\begin{proof}
    Let $\mathbf M$ be a tree-ordering of a graph $G$. Then $\Minor(\mathbf M)^E$ contains only minors of $G$. Thus, 
    $\Minor(\mathscr C_\Upsilon)^E$ is included in the closure of $\mathscr C$ by taking minors. If $\mathscr C$ excludes some graph as a minor, then this closure is a proper minor closed class, hence is nowhere dense. According to \zcref{thm:mainS}, it follows that $\mathscr C_\Upsilon$ is monadically dependent.

    Otherwise, if $\mathscr C$ contains every graph as a minor, there exists, for each $n\in\mathbb N$, some graph $G_n\in\mathscr C$ and some tree-ordering $G_n^\Upsilon$ such that the $1$-subdivision of~$K_n$ is an induced subgraph of a $\prec$-contraction of $G_n^\Upsilon$. 
    Then $\{G_n^\Upsilon\colon n\in\mathbb N\}$ is not monadically dependent, hence $\mathscr C_\Upsilon$ is not monadically dependent. 
\end{proof}

\begin{cor}[store=cor:succ_tree]
The ordered infinite binary tree with vertex set~$2^{<\omega}$ and order $\bar u<\bar v$ if $|\bar u|<|\bar v|$ or $|\bar u|=|\bar v|$ and $\bar u$ precedes $\bar v$ in the lexicographic order is monadically dependent.
\end{cor}
\begin{proof}
    By compactness, it is sufficient to prove that the class of all finite binary trees ordered by $\bar u<\bar v$ if $|\bar u|<|\bar v|$ or $|\bar u|=|\bar v|$ and~$\bar u$ precedes $\bar v$ in the lexicographic order is monadically dependent. This follows from the fact that the flattening of these ordered graphs are all planar.
\end{proof}
\begin{remark}
    It is natural to wonder what weakly sparse class of graphs $\mathscr C$ is such that there is a class $\mathscr C^+$ of tree-orderings of the graphs in $\mathscr C$ that is monadically dependent. (Here we  can select a single tree-ordering for each graph in $\mathscr C$.)

    Clearly, the class of connected cubic graphs does not have this property, as otherwise some expansion of this class to a class of ordered graphs could be obtained by transduction, contradicting the property that this class has unbounded twin-width (See \cite{tww4}).
\end{remark}

    Say that a TOWS graph $\mathbf M$ is an \ndef{internal tree-ordering expansion} of $G$ if  there exists a spanning forest $F$ such that (for an arbitrary rooting of $F$) $\mathbf M$ is the expansion of $G$ by the forest-ordering defined by $F$.
\begin{conj}
    Let $\mathscr C$ be a hereditary weakly sparse class of graphs. Then $\mathscr C$ has a dependent internal tree-ordering expansion if and only if $\mathscr C$ has bounded twin-width.
\end{conj}

\section{Bounded degree TOWS graphs and twin-width}
\label{sec:tows_tww}
We have proved that a class $\mathscr C$ of tree-ordered weakly sparse graphs is monadically dependent if and only if $\Minor(\mathscr C)^E$ is nowhere dense (\zcref{thm:mainG}).
In the case where the cover graph of the tree-orders in the $\prec$-reduct of $\mathscr C$ has bounded degree, more can be said.

\label{p:tows_tww}
\begin{thm}[store*=thm:tows_tww,restate-keys={note=see \zcpageref[nocap]{p:tows_tww}}]
Let $D$ be an integer and let $\mathscr C$ be a class of tree-ordered weakly sparse graphs, whose $\prec$-reduct cover graphs have degree at most~$D$. Then  the following are equivalent:
\begin{enumerate}
    \item $\mathscr C$ is monadically dependent;
        \item $\mathscr C$ has bounded twin-width;
    \item $\Minor(\mathscr C)^E$ is weakly sparse.
\end{enumerate}
\end{thm}

\begin{proof}
    We consider the natural transduction $\mathsf T$ expanding 
    a tree-ordered graph (whose tree-order's cover graph has maximum degree at most $D$) with a linear orders.
    
     \begin{trivlist}
        \item[\equ{1}{2}:] this follows from the application to $\mathsf T(\mathscr C)$ of the characterization of classes with bounded twin-width given in \cite{tww4}.
        \item[\impl{2}{3},] as $\Minor(\mathscr C)^E$ is a transduction of $\mathscr C$, which is nowhere dense (according to \zcref{thm:mainG}).
        \item[\impl{3}{4}] is trivial;
        \item[\impl{4}{2}:] Assume $\mathsf T(\mathscr C)$ has unbounded twin-width. Then, we can find arbitrarily large biclique ordered-minors. Let $I_1,\dots,I_n$ be the corresponding intervals, and let $s_i=\bigwedge I_i$ for $i\in[n]$. 
        Then either we can find a large subset $I\subseteq [n]$ such that $\{s_i\colon i\in I\}$ is an antichain in $\prec$ (in which case contracting the subtrees rooted in the~$s_i$ for $i\in I$ gives a large biclique in $\Minor(\mathbf M)^E$), or we can find a large subset $I\subseteq [n]$ such that $\{s_i\colon i\in I\}$ is a chain in $\prec$.
        Note that, by construction, 
         $s_i$ in the parent of some element of $I_i$, and that $I_i\cup\{s_i\}$ is a subtree of the cover graph of~$\mathbf M^\prec$ whose leaves are leaves of $\mathbf M^\prec$, with the possible exception of the last element~$f_i$ of $I_i$.
        As $I_i\cap I_j=\emptyset$, the subtrees induced by $I_i\cup\{s_i\}$ and $I_j\cup\{s_j\}$ may intersect only if $s_i=s_j$, $s_i\in I_j$, or $s_j\in I_i$.
        There are at most $d$ values $i$ with the same $s_i$ (as each of these values uses a child of $s_i$) and $s_i$ belongs to at most on $I_j$.
        This implies that the conflict graph  is $(D+2)$-degenerate. Thus, we can find a large conflict free set, which will define a biclique in $\Minor(\mathbf M)^E$.\qedhere
    \end{trivlist}
\end{proof}
\addtocontents{toc}{\vspace{.5em}}
\begin{idx}
\printindex
\end{idx}
\newpage

\bibliographystyle{amsplain}
\bibliography{biblio,newbib}

\end{document}